\documentclass{lmcs}
\pdfoutput=1

\usepackage[utf8]{inputenc}

\usepackage{lastpage}
\lmcsdoi{19}{4}{38}
\lmcsheading{}{\pageref{LastPage}}{}{}%
{Oct.~21,~2022}{Dec.~22,~2023}{}


\DeclareSymbolFont{stmry}{U}{stmry}{m}{n}
\DeclareMathDelimiter\llbracket{\mathopen}{stmry}{"4A}
                                          {stmry}{"71}
\DeclareMathDelimiter\rrbracket{\mathclose}{stmry}{"4B}
                                           {stmry}{"79}

\DeclareSymbolFont{symbolsC}{U}{pxsyc}{m}{n}
\DeclareMathSymbol{\coloneqq}{\mathrel}{symbolsC}{"42}

\usepackage[linesnumbered,lined,figure,noend]{algorithm2e}
\usepackage{fontawesome}
\usepackage{xcolor}
\usepackage{float}
\usepackage{tikz}
\usepackage{amssymb}
\usepackage{amsmath}
\usepackage{booktabs}
\usepackage{makecell}
\usetikzlibrary{calc}
\usetikzlibrary{decorations}
\usetikzlibrary{decorations.pathreplacing}

\makeatletter
\renewcommand{\SetKwInOut}[2]{%
    \sbox\algocf@inoutbox{\KwSty{#2}\algocf@typo:}%
    \expandafter\ifx\csname InOutSizeDefined\endcsname\relax% if first time used
    \newcommand\InOutSizeDefined{}\setlength{\inoutsize}{\wd\algocf@inoutbox}%
    \sbox\algocf@inoutbox{\parbox[t]{\inoutsize}{\KwSty{#2\algocf@typo:}\hfill}~}\setlength{\inoutindent}{\wd\algocf@inoutbox}%
    \else% else keep the larger dimension
    \ifdim\wd\algocf@inoutbox>\inoutsize%
    \setlength{\inoutsize}{\wd\algocf@inoutbox}%
    \sbox\algocf@inoutbox{\parbox[t]{\inoutsize}{\KwSty{#2\algocf@typo:}\hfill}~}\setlength{\inoutindent}{\wd\algocf@inoutbox}%
    \fi%
    \fi% the dimension of the box is now defined.
    \algocf@newcommand{#1}[1]{%
        \ifthenelse{\boolean{algocf@inoutnumbered}}{\relax}{\everypar={\relax}}%
        {\let\\\algocf@newinout\hangindent=\inoutindent\hangafter=1\parbox[t]{\inoutsize}{\KwSty{#2\algocf@typo:}\hfill}~##1\par}%
        \algocf@linesnumbered% reset the numbering of the lines
}}%
\makeatother

\newcommand{\aggaux}[3]{[\mathsf{CNT}\,#1.\,#3](#2)}
\newcommand{\aggbound}{\varlist}
\newcommand{\aggresult}{c}
\newcommand{\aggresulta}{c_1}
\newcommand{\aggresultb}{c_2}
\newcommand{\aggresultc}{c'}
\newcommand{\agggroup}{M}
\newcommand{\aggqry}{\aggaux{\aggbound}{\aggresult}{\exqryvarlist}}
\newcommand{\aggqryrestr}{\aggaux{\aggbound}{\aggresultc}%
{\bigvee_{\qryrestr\in\restr} (\exqryvarlist\land\qryrestr)}}
\newcommand{\aggrestrdist}{\bigwedge_{\qryrestr\in\restr}%
\neg(\exists\varlist.\,\exqryvarlist\land\qryrestr)}
\newcommand{\aggrestr}{\bigwedge_{\qryrestr\in\restr} \neg\qryrestr}

\newcommand{\sigsymb}{\sigma}
\newcommand{\consts}{\mathcal{C}}
\newcommand{\varset}{\mathcal{V}}
\newcommand{\predsymbs}{\mathcal{R}}
\newcommand{\auxpredsymb}{\mathsf{A}}
\newcommand{\auxpredtuple}{\mathsf{t}}
\newcommand{\auxvar}{t}
\newcommand{\predsymb}{r}
\newcommand{\termsymb}{t}
\newcommand{\cstsymb}{\mathsf{c}}
\newcommand{\cstsymba}{\mathsf{c}'}
\newcommand{\valsymb}{\alpha}
\newcommand{\valsymba}{\alpha'}
\newcommand{\eqs}{\mathcal{E}}
\newcommand{\apreds}{\mathcal{A}}
\newcommand{\apredsa}{\mathcal{A}_1}
\newcommand{\apredsb}{\mathcal{A}_2}
\newcommand{\qpreds}{\mathcal{G}}
\newcommand{\qpredsa}{\mathcal{G}_1}
\newcommand{\qpredsb}{\mathcal{G}_2}

\newcommand{\qpredsy}{\mathcal{G}_{\vary}}
\newcommand{\cpreds}{\mathcal{G}}
\newcommand{\cpredsa}{\mathcal{G}_1}
\newcommand{\cpredsb}{\mathcal{G}_2}

\newcommand{\domval}{d}
\newcommand{\domvalc}{d'}
\newcommand{\domvallist}{\vec{\domval}}
\newcommand{\dom}{\mathcal{D}}
\newcommand{\doma}{\mathcal{D}_1}
\newcommand{\domb}{\mathcal{D}_2}
\newcommand{\str}{\mathcal{S}}
\newcommand{\stra}{\mathcal{S}_1}
\newcommand{\strb}{\mathcal{S}_2}
\newcommand{\strtrain}{\mathcal{T}}

\newcommand{\listlength}{k}
\newcommand{\len}[1]{\left|#1\right|}
\newcommand{\card}[1]{\left|#1\right|}

\newcommand{\iter}{i}
\newcommand{\idx}{j}

\newcommand{\loopbounddef}[2]{\{\fixqry\in #1\mid #2\in\gens{\fixqry}\}}
\newcommand{\loopbound}[2]{\mathsf{fixbound}(#1, #2)}

\newcommand{\freeloopadef}[1]{\{(\fixqry, \eqconjqry)\in #1\mid\gens{\fixqry}\neq\emptyset\}}
\newcommand{\freeloopa}[1]{\mathsf{fixfree}(#1)}
\newcommand{\freeloopbdef}[2]{\{(\qryinqfin, \eqconjqry)\in #1\mid\hanging{\qryinqfin}{\eqconjqry}\neq\emptyset\lor\allowbreak\fv{\qryinqfin}\cup\fv{\eqconjqry}\neq\fv{#2}\}}
\newcommand{\freeloopb}[2]{\mathsf{inf}(#1, #2)}

\newcommand{\guard}{guard}
\newcommand{\guardqry}[1]{\mathsf{guard}(#1)}

\newcommand{\varle}[1]{\lesssim_{#1}}

\newcommand{\infequivsign}{\stackrel{\infty}{\equiv}}
\newcommand{\infequiv}[2]{{#1}\infequivsign {#2}}

\newcommand{\nats}{\mathbb{N}}

\newcommand{\dgstrat}{\gamma}

\newcommand{\diffop}{-}
\newcommand{\diff}[2]{#1\diffop#2}
\newcommand{\gdiffop}{\triangleright}
\newcommand{\gdiff}[2]{#1\gdiffop#2}

\newcommand{\tupover}[2]{#1(#2)}
\newcommand{\evaltup}[2]{#1[#2]}

\newcommand{\citeauthorVanGT}{Van Gelder and Topor}
\newcommand{\citeauthorAlice}{Abiteboul et al.}
\newcommand{\citeauthorAgg}{Clau{\ss}en et al.}
\newcommand{\citeauthorIeee}{Chomicki and Toman}
\newcommand{\citeauthorMWJ}{Ngo et al.}
\newcommand{\citeauthorAil}{Ailamazyan et al.}
\newcommand{\citeauthorEMHJ}{Escobar-Molano et al.}

\newcommand{\isneg}[1]{\mathsf{neg}(#1)}
\newcommand{\iseq}[1]{\mathsf{eq}(#1)}

\newcommand{\proj}[2]{\tilde{\exists}{#1}.\,{#2}}

\newcommand{\adom}[1]{\mathsf{adom}({#1})}
\newcommand{\adomd}[2]{\mathsf{adom}^{#2}({#1})}
\newcommand{\adomqry}[2]{\mathsf{AD}({#1},\,{#2})}
\newcommand{\arity}{\iota}
\newcommand{\sz}[1]{\mathsf{m}(#1)}
\newcommand{\subs}[1]{\mathsf{sub}(#1)}
\newcommand{\cntsub}[1]{\left|\mathsf{sub}(#1)\right|}
\newcommand{\termphi}[1]{\mathsf{eqneg}({#1})}
\newcommand{\cost}[2]{\mathsf{cost}^{#2}(#1)}

\newcommand{\fv}[1]{\mathsf{fv}(#1)}
\newcommand{\fvseq}[1]{\vec{\mathsf{fv}}(#1)}
\newcommand{\free}[2]{{#1}\in\fv{#2}}
\newcommand{\absent}[2]{{#1}\notin\fv{#2}}
\newcommand{\av}[1]{\mathsf{av}({#1})}

\newcommand{\sattup}[1]{\left\llbracket{#1}\right\rrbracket}
\newcommand{\sattupd}[2]{\sattup{#1}^{#2}}

\newcommand{\flconj}[1]{\mathsf{flat}^{\land}(#1)}
\newcommand{\fldisj}[1]{\mathsf{flat}^{\lor}(#1)}
\newcommand{\flop}[1]{\mathsf{flat}^{\oplus}(#1)}

\newcommand{\edb}[1]{\mathsf{ap}(#1)}
\newcommand{\colatoms}[1]{\mathsf{aps}({#1})}
\newcommand{\transeqs}[1]{\mathsf{eqs}^*({#1})}
\newcommand{\closeatoms}[1]{\mathsf{qps}({#1})}
\newcommand{\closeatomseq}[1]{\overline{\mathsf{qps}}({#1})}
\newcommand{\colpreds}[1]{\mathsf{qps}({#1})}
\newcommand{\colpredsqry}[1]{\mathsf{qps}^{\lor}({#1})}
\newcommand{\coleqs}[2]{\mathsf{eqs}({#1},{#2})}

\newcommand{\cp}[1]{\mathsf{cp}(#1)}
\newcommand{\sconj}[1]{\mathsf{sort}^{\land}(#1)}

\newcommand{\cntname}{\mathsf{optcnt}}
\newcommand{\cnt}[1]{\cntname(#1)}
\newcommand{\ranfra}[1]{\mathsf{ranf2ra}(#1)}
\newcommand{\rasql}[1]{\mathsf{ra2sql}(#1)}
\newcommand{\ranfsqlname}{\mathsf{ranf2sql}}
\newcommand{\ranfsql}[1]{\ranfsqlname(#1)}

\newcommand{\allowedname}{\mathsf{rb}}
\newcommand{\allowed}[1]{\allowedname(#1)}
\newcommand{\splitq}[1]{\mathsf{split}(#1)}
\newcommand{\rw}[1]{\mathsf{rw}({#1})}

\newcommand{\pushnot}[1]{\mathsf{srnf}(#1)}
\newcommand{\enf}[1]{\mathsf{enf}(#1)}
\newcommand{\SRTORANFName}{\mathsf{sr2ranf}}
\newcommand{\allowtosafe}[2]{\SRTORANFName({#1},{#2})}
\newcommand{\allowtosafeqry}[1]{\SRTORANFName({#1})}
\newcommand{\safe}[1]{\mathsf{ranf}(#1)}

\newcommand{\dgname}{\mathsf{dg}}
\newcommand{\dg}[5]{\dgname(#1, #2, #3, #4, #5)}
\newcommand{\dgeqsname}{\mathsf{dg}^\approx}
\newcommand{\dgeqs}[2]{\dgeqsname(#1, #2)}
\newcommand{\poseqs}{\varset^+}
\newcommand{\poseqsa}{\varset_1^+}
\newcommand{\poseqsb}{\varset_2^+}
\newcommand{\negeqs}{\varset^-}
\newcommand{\negeqsa}{\varset_1^-}
\newcommand{\negeqsb}{\varset_2^-}
\newcommand{\vareqsa}{\varset^1}
\newcommand{\vareqsb}{\varset^2}

\newcommand{\x}{1.1}

\newcommand{\z}{2.8}
\newcommand{\w}{0.1}
\newcommand{\h}{0.1}
\newcommand{\cspace}{\;\;}
\newcommand{\highlight}[1]{\setlength{\fboxsep}{0pt}\colorbox{black!15}{$#1$}}

\newcommand{\flleftaux}[3]{\multicolumn{2}{@{}l@{\;}}{#1}&\text{if}&#2#3}
\newcommand{\flleft}[2]{\flleftaux{#1}{#2}{;}}

\newcommand{\flright}{\flleft}

\newcommand{\rulenocond}[1]{\multicolumn{2}{@{}l@{}}{#1;}}
\newcommand{\ruleif}[2]{#1&\text{if}&#2;}
\newcommand{\ruleifdot}[2]{#1&\text{if}&#2.}
\newcommand{\ruleiff}[2]{#1&\text{iff}&#2;}
\newcommand{\rulenocondgen}[1]{#1;}
\newcommand{\ruleifgen}[2]{\multicolumn{3}{@{}l@{}}{#1}\\\multicolumn{3}{@{}l@{}}{\quad\text{if }#2;}}
\newcommand{\ruleifgendot}[2]{\multicolumn{3}{@{}l@{}}{#1}\\\multicolumn{3}{@{}l@{}}{\quad\text{if }#2.}}

\newcommand{\bigor}[1]{{\textstyle\bigvee_{#1}}}

\SetInd{0.2em}{0.4em}

\newcommand{\tildeo}[1]{\tilde{\mathcal{O}}\left(#1\right)}
\newcommand{\emptyseq}{[]}
\newcommand{\eq}{\coloneqq}

\newcommand{\loc}{1000}
\newcommand{\evalsz}{14}
\newcommand{\evalnfv}{2}
\newcommand{\evalmaxn}{4}
\newcommand{\traindg}{2}

\newcommand{\togolf}{300s}
\newcommand{\tosusp}{300s}

\newcommand{\prodparam}{n}
\newcommand{\usrparam}{m}

\newcommand{\eparam}{n}

\newcommand{\smallexp}{\textsc{Small}}
\newcommand{\mediumexp}{\textsc{Medium}}
\newcommand{\largeexp}{\textsc{Large}}
\newcommand{\infexp}{\textsc{Infinite}}
\newcommand{\realexp}{\textsc{Real}}

\newcommand{\psql}{\ensuremath{\mathsf{PostgreSQL}}}
\newcommand{\msql}{\ensuremath{\mathsf{MySQL}}}
\newcommand{\psqlsub}{\ensuremath{^{\mathsf{P}}}}
\newcommand{\msqlsub}{\ensuremath{^{\mathsf{M}}}}
\newcommand{\tool}{\ensuremath{\mathsf{RC2SQL}}}
\newcommand{\toolnonopt}{\ensuremath{\mathsf{RC2SQL}^{-}}}
\newcommand{\vgtool}{\ensuremath{\mathsf{VGT}}}
\newcommand{\vgtoolnonopt}{\ensuremath{\mathsf{VGT}^{-}}}
\newcommand{\vgtna}{\ensuremath{-}}
\newcommand{\ddd}{\ensuremath{\mathsf{DDD}}}
\newcommand{\ldd}{\ensuremath{\mathsf{LDD}}}
\newcommand{\mpreg}{\ensuremath{\mathsf{MonPoly}^\mathsf{REG}}}
\newcommand{\radb}{\ensuremath{\mathsf{radb}}}

\newcommand{\trtime}{Translation time}

\newcommand{\strategya}{\text{Infinite results ($\dgstrat=0$)}}
\newcommand{\strategyb}{\text{Finite results ($\dgstrat=1$)}}

\newcommand{\predbrand}{\mathsf{B}}
\newcommand{\predprod}{\mathsf{P}}
\newcommand{\predscore}{\mathsf{S}}
\newcommand{\predtext}{\mathsf{T}}
\newcommand{\predexa}{\mathsf{P}_3} %ternary
\newcommand{\predexb}{\mathsf{P}_2} %binary
\newcommand{\predexc}{\mathsf{P}_1} %unary

\newcommand{\brandname}{\mathit{brand}}
\newcommand{\username}{\mathit{user}}
\newcommand{\scorename}{\mathit{score}}

\newcommand{\varbrand}{b}
\newcommand{\varuser}{u}
\newcommand{\varscore}{s}
\newcommand{\varprod}{p}
\newcommand{\vartext}{t}

\newcommand{\varx}{x}
\newcommand{\varxs}{\vec{\varx}}
\newcommand{\vary}{y}

\newcommand{\varz}{z}
\newcommand{\varw}{u}
\newcommand{\vart}{v}

\newcommand{\varlistname}{v}
\newcommand{\varlist}{\vec{\varlistname}}

\newcommand{\varlistbname}{w}
\newcommand{\varlistb}{\vec{\varlistbname}}

\newcommand{\qry}{Q}
\newcommand{\subqry}{Q'}
\newcommand{\subqrya}{Q_1'}
\newcommand{\subqryb}{Q_2'}
\newcommand{\subqryc}{Q''}
\newcommand{\qrya}{Q_1}
\newcommand{\qryb}{Q_2}
\newcommand{\qryx}{Q_{\varx}}
\newcommand{\qryy}{Q_{\vary}}
\newcommand{\qryxy}{Q_{\varx\vary}}
\newcommand{\qryvarlista}{Q_{\varlist}}
\newcommand{\qryvarlistb}{Q_{\varlist}'}

\newcommand{\qryconj}{Q^\land}

\newcommand{\exqryx}{Q_{\varx}}
\newcommand{\exqryxc}{Q_{\varx}'}
\newcommand{\exqryy}{Q_{\vary}}
\newcommand{\exqryz}{Q_{\varz}}
\newcommand{\exqryvarlist}{Q_{\varlist}}

\newcommand{\exqryvarlistb}{Q_{\varlistb}}
\newcommand{\exsafeqryvarlistb}{\hat{Q}_{\varlistb}}

\newcommand{\qryrb}{\tilde{Q}}
\newcommand{\qryrbx}{\tilde{Q}'}

\newcommand{\qrysrnf}{Q_{\mathit{srnf}}}
\newcommand{\qryenf}{Q_{\mathit{enf}}}

\newcommand{\eqqry}{Q^\approx}
\newcommand{\eqconjqry}{E}%R^\approx}
\newcommand{\eqconjqryc}{E'}%R'^\approx}

\newcommand{\fixqry}{Q_{\mathit{fix}}}
\newcommand{\fixqryc}{\fixqry'}
\newcommand{\qryinqfin}{Q_{\mathit{f}}}
\newcommand{\qryinqinf}{Q_{\mathit{i}}}

\newcommand{\qryfin}{Q_{\mathit{fin}}}
\newcommand{\qryfina}{Q_{\mathit{fin}}'}
\newcommand{\qryinf}{Q_{\mathit{inf}}}
\newcommand{\qryinfa}{Q_{\mathit{inf}}'}

\newcommand{\qsusp}{Q^{\mathit{susp}}}
\newcommand{\qsuspa}{{Q^{\mathit{susp}}_{\mathit{text}}}}
\newcommand{\qsuspusr}{Q^{\mathit{susp}}_{\mathit{user}}}

\newcommand{\safefin}{\hat{Q}_{\mathit{fin}}}
\newcommand{\safeinf}{\hat{Q}_{\mathit{inf}}}
\newcommand{\safeqry}{\hat{Q}}
\newcommand{\safesubqry}{\hat{Q}'}

\newcommand{\oneatom}{Q_{\mathit{ap}}}
\newcommand{\projoneatom}{Q_{\mathit{qp}}}

\newcommand{\cpreda}{Q_1}
\newcommand{\cpredb}{Q_2}
\newcommand{\cpredc}{Q_3}

\newcommand{\qex}{Q_{\mathit{ex}}}
\newcommand{\qryvgt}{Q_{\mathit{vgt}}}

\newcommand{\gens}[1]{\mathsf{nongens}({#1})}
\newcommand{\eclass}[1]{\mathsf{classes}(#1)}
\newcommand{\hanging}[2]{\mathsf{disjointvars}(#1, #2)}

\newcommand{\varseta}{V}
\newcommand{\varsetb}{V'}

\newcommand{\varseq}{\varlist}

\newcommand{\disjs}{\mathcal{Q}}
\newcommand{\poss}{\mathcal{Q}^+}
\newcommand{\negs}{\mathcal{Q}^-}
\newcommand{\eqset}{\mathcal{Q}^\approx}
\newcommand{\neqset}{\mathcal{Q}^{\not\approx}}

\newcommand{\projatomset}{\mathcal{Q}_{\mathit{qps}}}

\newcommand{\bset}[2]{\mathcal{B}_{#2}(#1)}

\newcommand{\copyvars}[2]{\{\varx\approx\vary\mid\varx\in #1\land
\vary\in#2\}}

\newcommand{\setfin}{\mathcal{Q}_{\mathit{fin}}}
\newcommand{\setinf}{\mathcal{Q}_{\mathit{inf}}}

\newcommand{\restr}{\disjs}
\newcommand{\qryrestr}{\overline{\qry}}
\newcommand{\conjrestr}{\bigwedge_{\qryrestr\in\restr} \qryrestr}
\newcommand{\sumszrestr}{\sum_{\qryrestr\in\restr} \sz{\qryrestr}}
\newcommand{\restra}{\overline{\disjs}}
\newcommand{\conjrestra}{\bigwedge_{\qryrestr\in\restra} \qryrestr}
\newcommand{\restrf}[1]{\mathcal{Q}_{#1}}

\newcommand{\posset}[1]{{\mathcal{T}^+_{#1}}}
\newcommand{\possetc}[1]{{\overline{\mathcal{T}}^+_{#1}}}
\newcommand{\negset}[1]{{\mathcal{T}^-_{#1}}}
\newcommand{\negsetc}[1]{{\overline{\mathcal{T}}^-_{#1}}}
\newcommand{\auxseta}[1]{{\mathcal{T}^1_{#1}}}
\newcommand{\auxsetb}[1]{{\mathcal{T}^2_{#1}}}

\newcommand{\nondetobj}{X}
\newcommand{\nondetset}{\mathcal{X}}
\newcommand{\nondetfun}{f}
\newcommand{\nondetchoice}{%
We denote the non-deterministic choice of an object $X$ from a non-empty set
$\mathcal{X}$ as $X\leftarrow\mathcal{X}$.}

\newcommand{\setminimal}[2]{\forall{#1}\in {#2}.\,
\fv{{#2}\setminus\{{#1}\}}\neq\fv{#2}}
\newcommand{\projminimaldef}{\setminimal{\projoneatom}{\projatomset}}
\newcommand{\projminimal}{\mathsf{minimal}(\projatomset)}

\newcommand{\RC}{RC\xspace}
\newcommand{\safeprop}{RANF}

\newcommand{\fbvs}{(free and bound) variables}
\newcommand{\pwdist}{pairwise distinct \fbvs}
\newcommand{\infequivprop}{inf-equivalent}
\newcommand{\di}{domain-independent}

\newcommand{\ethz}{(\textsc{vgt}_1)}
\newcommand{\vgtrw}{(\textsc{vgt}_2)}
\newcommand{\transnew}{\faStar}
\newcommand{\transex}{\faStar$\exists$}
\newcommand{\transfree}{\faStarO}
\newcommand{\lagg}{\#}
\newcommand{\laggg}{\#\#}
\newcommand{\rulea}{(R1)}
\newcommand{\ruleb}{(R2)}
\newcommand{\rulec}{(R3)}
\newcommand{\propfv}{\textsc{fv}}
\newcommand{\propfin}{\textsc{eval}}
\newcommand{\propfvlab}{\label{lab:fv}}
\newcommand{\propfinlab}{\label{lab:eval}}
\newcommand{\propfvref}{\hyperref[lab:fv]{\propfv}}
\newcommand{\propfinref}{\hyperref[lab:eval]{\propfin}}
\newcommand{\propcon}{\textsc{con}}
\newcommand{\propcst}{\textsc{cst}}
\newcommand{\proppred}{\textsc{var}}
\newcommand{\proprep}{\textsc{rep}}

\newcommand{\gen}[2]{\mathsf{gen}({#1},{#2})}
\newcommand{\tgen}[3]{\mathsf{gen}({#1},{#2},{#3})}
\newcommand{\tcov}[3]{\mathsf{cov}({#1},{#2},{#3})}
\newcommand{\xgen}[2]{\mathsf{gen}_{\mathsf{vgt}}({#1},{#2})}
\newcommand{\vgen}[3]{\mathsf{gen}_{\mathsf{vgt}}({#1},{#2},{#3})}
\newcommand{\xcon}[2]{\mathsf{con}_{\mathsf{vgt}}({#1},{#2})}
\newcommand{\con}[3]{\mathsf{con}_{\mathsf{vgt}}({#1},{#2},{#3})}

\renewcommand{\emptyset}{\varnothing}
\newcommand{\sconjtwo}[2]{\bigwedge\nolimits^\approx({#1},{#2})}

\theoremstyle{defC}\newtheorem{exampleC}[exa]{Example}
\theoremstyle{plain}\newtheorem{lemma}[lem]{Lemma} %\crefname{lemma}{Lemma}{lemma}
\theoremstyle{plain}\newtheorem{theorem}[thm]{Theorem}
\theoremstyle{definition}\newtheorem{definition}[defi]{Definition}
\theoremstyle{definition}\newtheorem{example}[exa]{Example}

\newsavebox\isabelleqedbox
\savebox\isabelleqedbox{%
    \begin{tikzpicture}[x=0.8mm, y=0.8mm, 
        baseline=-0.3mm, line join=round]
        \begin{scope}[yslant=-0.5]
            \draw (0,0) rectangle +(1,1);
            \draw (2,1) rectangle +(1,1);
            \draw (1,2) rectangle +(1,1);
        \end{scope}
        \begin{scope}[yslant=0.5]
            \filldraw (1,-1) rectangle +(1,1);
            \filldraw (3,-2) rectangle +(1,1);
            \filldraw (2,0) rectangle +(1,1);
        \end{scope}
        \begin{scope}[yslant=0.5,xslant=-1]
            \draw (1,0) rectangle +(1,1);
            \draw (2,-1) rectangle +(1,1);
            \draw (3,1) rectangle +(1,1);
        \end{scope}
    \end{tikzpicture}%
}
\newcommand{\isabelleqed}{%
  \usebox{\isabelleqedbox}%
}

\begin{document}

\title[Efficient Evaluation of Arbitrary Relational Calculus Queries]{Efficient Evaluation of\texorpdfstring{\\}{} Arbitrary Relational Calculus Queries}

\author[M.~Raszyk]{Martin Raszyk\lmcsorcid{0000-0003-3018-2557}}[a]
\author[D.~Basin]{David Basin\lmcsorcid{0000-0003-2952-939X}}[b]
\author[S.~Krsti\'c]{Sr\dj{}an Krsti\'c\lmcsorcid{0000-0001-8314-2589}}[b]
\author[D.~Traytel]{Dmitriy Traytel\lmcsorcid{0000-0001-7982-2768}}[c]

\renewcommand{\S}{Section~}

\address{DFINITY, Zurich, Switzerland}	\email{martin.raszyk@dfinity.org}  

\address{Department of Computer Science, ETH Zürich, Zurich, Switzerland}	\email{\{basin, srdan.krstic\}@inf.ethz.ch}

\address{Department of Computer Science, University of Copenhagen, Copenhagen, Denmark}
\email{traytel@di.ku.dk}

\begin{abstract}

The relational calculus (RC) is a concise, declarative query language.
However,
existing RC query evaluation approaches are inefficient and
often deviate from established algorithms based on finite tables used in
database management systems. We devise a new translation of an arbitrary
RC query into two safe-range queries, for which the finiteness
of the query's evaluation result is guaranteed. Assuming an infinite domain, the
two queries have the following meaning: The first is closed and characterizes
the original query's relative safety, i.e., whether given a fixed database, the
original query evaluates to a finite relation. The second safe-range query is
equivalent to the original query, if the latter is relatively safe. We compose
our translation with other, more standard ones to ultimately obtain two SQL
queries. This allows us to use standard database management systems to evaluate
arbitrary RC queries. We show that our translation improves the
time complexity over existing approaches, which we also empirically confirm in
both realistic and synthetic experiments.
\vspace*{-2ex}
\end{abstract}

\maketitle

\section{Introduction}\label{sec:intro}

Codd's theorem states that all domain-independent queries of the relational
calculus (RC) can be expressed in relational algebra (RA)~\cite{DBLP:persons/Codd72}. A popular interpretation 
of this result is that RA suffices to express all interesting
queries. This interpretation justifies why SQL evolved as the
practical database query language with the RA as its mathematical foundation.
SQL is declarative and abstracts over the actual RA expression used to evaluate a query. Yet, SQL's syntax %is heavily influenced by RA and
inherits RA's deliberate syntactic limitations, such as union-compatibility, which ensure domain independence.
%, such as the lack of arbitrary negation, union, and universal quantification.
RC does not have such syntactic limitations, which arguably makes it a more attractive declarative query language than both RA and SQL. The main problem of RC is that it is not immediately clear how to evaluate even domain-independent queries, much less how to handle the domain-dependent (i.e., not domain-independent) ones.

%declarative: queries can focus on what the result should be rather than how to compute it. Being declarative is a burden on the query
%evaluation engine, since it is not immediately clear how to evaluate queries, but it liberates
%the query author by making many queries significantly easier to
%express.

\looseness=-1
As a running example, consider a shop in which brands (unary
finite relation $\predbrand$ of brands) sell products (binary finite relation $\predprod$ relating brands and
products) and products are reviewed by users with a score (ternary
finite relation $\predscore$ relating products, users, and scores). We consider a brand \emph{suspicious} if there is a user and a score such that all the
brand's products were reviewed by that user with that score. An RC query
computing suspicious brands is %
\[\qsusp \eq \predbrand(\varbrand) \land \exists \varuser, \varscore.\;\forall \varprod.\;\predprod(\varbrand,\varprod) \longrightarrow \predscore(\varprod, \varuser, \varscore).\]
This query is domain independent and follows closely our informal description.
It is not, however, clear how to evaluate it because its second conjunct is 
domain dependent as it is satisfied for every brand that does not occur in $\predprod$.
 Finding suspicious brands using RA or SQL is a challenge, which only
the best students from an undergraduate database course will accomplish. We give away an RA answer next (where $\diffop$ is the set difference operator
and $\gdiffop$ is the anti-join, also known as the \emph{generalized} difference operator~\cite{DBLP:books/aw/AbiteboulHV95}):
\[\begin{array}{@{}l@{}}
\pi_{\brandname}(\diff{(\highlight{\pi_{\username,\scorename}(\predscore)} \times \predbrand)}{\pi_{\brandname, \username, \scorename}(\gdiff{(\highlight{\pi_{\username, \scorename}(\predscore)} \times \predprod)}{\predscore})}) \cup
(\diff{\predbrand}{\pi_{\brandname}(\predprod)}).
\end{array}\]

The highlighted expressions $\pi_{\username, \scorename}(\predscore)$ are called \emph{generators}. They ensure that the left operands
of the anti-join and set difference operators include or have the same columns (i.e., are union-compatible) as the
corresponding right operands. 
(Following Codd~\cite{DBLP:persons/Codd72}, one could also use the active domain to obtain canonical, but far less efficient, generators.)

\looseness=-1
Van Gelder and Topor~\cite{DBLP:conf/pods/GelderT87,DBLP:journals/tods/GelderT91} present a translation from a
decidable class of domain-independent \RC queries, called \emph{evaluable}, to
RA expressions. Their translation of the evaluable $\qsusp$ query would yield different
generators, replacing both highlighted parts by $\pi_{\username}(\predscore) \times \pi_{\scorename}(\predscore)$. That one can avoid this Cartesian product as shown above is subtle: Replacing only the first highlighted generator with the product
results in an inequivalent RA expression.

\looseness=-1
Once we have identified suspicious brands, we may want to obtain the users whose scoring made the brands suspicious. In RC, omitting $\varuser$'s quantifier from $\qsusp$ achieves just that:
\[\qsuspusr \eq \predbrand(\varbrand) \land \exists \varscore.\;\forall \varprod.\;\predprod(\varbrand,\varprod) \longrightarrow \predscore(\varprod, \varuser, \varscore).\]
In contrast, RA cannot express the same property as it is domain dependent (hence also not evaluable and thus out of scope for Van Gelder and Topor's translation): $\qsuspusr$ is satisfied for every user if a brand has no products, i.e., it does not occur in $\predprod$. Yet, $\qsuspusr$  is satisfied for finitely many users on every database instance where $\predprod$ contains at least one row for every brand from the relation $\predbrand$,
in other words $\qsuspusr$ is \emph{relatively safe}
on such database instances.

How does one evaluate queries that are not evaluable or even
domain dependent? The main approaches from the 
literature~(\S\ref{sec:rel}) are either to use
variants of the active domain
semantics~\cite{DBLP:journals/jacm/BenediktL00,DBLP:journals/acta/HullS94,ail} or to abandon finite
relations entirely and evaluate queries using finite
representations of infinite (but well-behaved) relations such as systems of
constraints~\cite{DBLP:books/sp/Revesz02} or automatic
structures~\cite{DBLP:journals/mst/BlumensathG04}. These approaches favor
expressiveness over efficiency. But unlike query translations, they cannot benefit
from decades of practical database research and engineering.

In this work, we translate arbitrary \RC{} queries
to RA expressions under the assumption of an infinite domain.
To deal with queries that are domain dependent,
our translation produces two RA expressions, instead of a single equivalent one. The
first RA expression characterizes the original RC query's relative safety, the decidable question of whether the query evaluates to a 
finite relation for a given database, which can be the case even for
a domain-dependent query, e.g., $\qsuspusr$. If the original query is relatively safe on a given
database, i.e., produces some finite result, then the second RA expression evaluates to
the same finite result. Taken together, the two RA expressions
solve the \emph{query capturability} problem~\cite{DBLP:conf/lics/AvronH91}: they allow us to enumerate the
original RC query's finite evaluation result, or to learn that it would be infinite using RA
operations on the unmodified database.

% % OLD VERSION
% \begin{figure}[t]
%   \begin{tikzpicture}
%   \tikzset{pt/.style={draw,rounded corners=3pt,inner sep=5pt}, link/.style={-latex,line
%   width=0.1mm,draw=black}, path/.style=={-latex,line width=0.1mm,draw=black}}
%   \node[pt] (rc) at (-3*\z,0) {\RC{}};
%   \node[pt] (sr) at (-2*\z,0) {Safe-range \RC{}};
%   \node[pt] (srnf) at (-1*\z,0) {SRNF};
%   \node[pt] (ranf) at (0,0) {\safeprop};
%   \node[pt] (ra) at (1*\z,0) {RA};
%   \node[pt] (sql) at (2*\z,0) {SQL};
  
%   \path (rc) -- node[anchor=south,yshift=1, label=above:{\makecell[l]{\\~\S\ref{sec:translation}}}]{}(sr);
%   \draw[link] (rc) -- (sr.170);
%   \draw[link] (rc) -- (sr.190);
%   \draw[link] (sr) -- node[anchor=south,yshift=1,label=above:{\makecell[l]{\\~\S\ref{sec:srnf}}}]{}(srnf);
%   \draw[link] (srnf) -- node[anchor=south, yshift=1,align=left,label=above:{\makecell[l]{\S\ref{sec:ranf},
%   \\ \S\ref{sec:ranf-detail}}}]{}(ranf);
%   \path (ranf) edge [loop below] node[yshift=5,anchor=north,label=below:{~\S\ref{sec:aggs}}]{} (ranf);
%   \draw[link] (ranf) -- node[anchor=south,yshift=1,label=above:{\makecell[l]{\\~\S\ref{sec:ranf2ra}}}]{}(ra);
%   \draw[link] (ra) -- node[anchor=south,yshift=1,label=above:{\makecell[l]{\\~\S\ref{sec:ra2sql}}}]{}
%     (sql);
  
%     %  \node at ([yshift=-3mm, xshift=-15mm]ranf.180) {
%     %   $\allowtosafeqry{\qry}$};
%   \end{tikzpicture}
%   \vspace{-1ex}
%   \caption{Overview of our translation.}\label{fig:overview}
%   %\vspace{-1.5ex}
%   \end{figure}
\begin{figure}[t]
  \small
  \begin{tikzpicture}
  \tikzset{pt/.style={draw,rounded corners=3pt,inner sep=5pt}, link/.style={-latex,line
  width=0.1mm,draw=black}, path/.style=={-latex,line width=0.1mm,draw=black}}
  \node[pt] (rc) at (-3*\z,0) {\makecell[c]{\RC{}\\(\S\ref{sec:rc})}};
  \node[pt] (sr) at (-1.8*\z,0) {\makecell[c]{Safe-range \RC{}\\(\S\ref{sec:eval_allowed})}};
  \node[pt] (srnf) at (-0.7*\z,0) {\makecell[c]{SRNF\\(\S\ref{sec:srnf})}};
  \node[pt] (ranf) at (0.35*\z,0) {\makecell[c]{\safeprop\\(\S\ref{sec:ranf})}};
  \node[pt] (ra) at (1.22*\z,0) {RA};
  \node[pt] (sql) at (1.85*\z,0) {SQL};
  
  \path (rc) -- node[anchor=south,yshift=8, label=above:{\makecell[l]{\S\ref{sec:translation}}}]{}(sr);
  \draw[link] (rc) -- (sr.170);
  \draw[link] (rc) -- (sr.190);
  \draw[link] (sr) -- node[anchor=south,yshift=8,label=above:{\makecell[l]{\S\ref{sec:srnf-detail}}}]{}(srnf);
  \draw[link] (srnf) -- node[anchor=south, yshift=8,align=left,label=above:{\makecell[l]{
  \S\ref{sec:ranf-detail}}}]{}(ranf);
  \path (ranf) edge [loop below] node[yshift=4,anchor=north,label=below:{\S\ref{sec:aggs}}]{} (ranf);
  \draw[link] (ranf) -- node[anchor=south,yshift=8,label=above:{\makecell[l]{\S\ref{sec:ranf2ra}}}\;]{}(ra);
  \draw[link] (ra) -- node[anchor=south,yshift=8,label=above:{\makecell[l]{\;\S\ref{sec:ra2sql}}}]{}
    (sql);
  
    %  \node at ([yshift=-3mm, xshift=-15mm]ranf.180) {
    %   $\allowtosafeqry{\qry}$};
  \end{tikzpicture}
\vspace{-4ex}
\caption{Overview of our translation.}\label{fig:overview}
%\vspace{-1.5ex}
\end{figure}

Figure~\ref{fig:overview} summarizes our translation's steps and the sections where they are presented. Starting from an RC query, 
it produces two SQL queries via transformations to safe-range queries, the safe-range 
normal form (SRNF), the relational algebra normal form (RANF), and RA, respectively~(\S\ref{sec:prelim}). 
This article's main contribution is the first step: translating an RC query into two
safe-range RC queries~(\S\ref{sec:translation}), which fundamentally differs
from Van Gelder and Topor's approach and produces better generators, like
$\pi_{\username, \scorename}(\predscore)$ above. Our generators strictly improve the time complexity of query evaluation~(\S\ref{sec:complex}).

\looseness=-1
After the standard transformations from safe-range to RANF queries and from there to RA expressions, we translate the RA expressions into SQL using the \radb{} tool~\cite{radb} (\S\ref{sec:detail}).
We leverage various ideas from the literature to optimize the overall result. For example, we
generalize Clau\ss{}en et al.~\cite{DBLP:conf/vldb/ClaussenKMP97}'s approach to avoid evaluating
Cartesian products like $\pi_{\username, \scorename}(\predscore) \times \predprod$ in RANF queries
by using count aggregations (\S\ref{sec:aggs}).

\looseness=-1
The translation to SQL enables any standard database management system (DBMS) to
evaluate RC queries. We implement our translation and then use either \psql{} or 
\msql{} for query evaluation.
Using a real Amazon review dataset~\cite{DBLP:conf/emnlp/NiLM19} and our synthetic 
benchmark that generates hard database instances for random RC queries (\S\ref{sec:golf}),
we evaluate our translation's performance~(\S\ref{sec:eeval}). The evaluation shows that 
our approach 
outperforms Van Gelder and Topor's translation (which also uses a standard DBMS for evaluation)
and other RC evaluation approaches based on constraint databases and structure reduction. 

In summary, our three main contributions are as follows:
\begin{itemize}
\item We devise a translation of an arbitrary \RC{} query
into a pair of RA expressions as described above.
The time complexity of evaluating our translation's results
improves upon Van Gelder and Topor's approach~\cite{DBLP:journals/tods/GelderT91}.

\item  We implement our translation and extend it to produce SQL
queries. The resulting tool \tool{} makes RC a viable input language for any standard
DBMS.  We evaluate our tool on synthetic and real
data and confirm that our translation's improved asymptotic
time complexity carries over into practice.

\item To challenge \tool{} (and its competitors) in our evaluation, we devise
the \emph{Data Golf} benchmark that generates hard database instances for
randomly generated RC queries.
\end{itemize}

\looseness=-1
This article extends our ICDT 2022 conference paper~\cite{DBLP:conf/icdt/RaszykBKT22} with a 
more complete description of the translation. In particular, it describes the 
steps that follow our main contribution -- the translation of RC queries into two 
safe-range queries. 
In addition, we formally verify the functional correctness (but not the complexity analysis) of the
main contribution using the Isabelle/HOL proof assistant~\cite{DBLP:journals/afp/RaszykT22}. The
theorems and examples that have been verified in Isabelle are marked with a special symbol
(\isabelleqed). The formalization helped us identify and correct a technical oversight in 
the algorithm from the conference paper (even though the problem was compensated for by the 
subsequent steps of the translation in our implementation). 
%Our new algorithm description (\S\ref{sec:translation}) incorporates insights from the formalization, in particular more precise termination arguments and invariants. 

\section{Related Work}\label{sec:rel}

We recall Trakhtenbrot's theorem and the fundamental notions of \emph{capturability} and \emph{data complexity.}
Given an RC query over a \emph{finite} domain,
Trakhtenbrot~\cite{trakhtenbrot1950impossibility}
showed that it is undecidable whether there exists
a (finite) structure and a variable assignment satisfying the query.
In contrast, the question of whether a \emph{fixed} structure and a \emph{fixed} variable assignment satisfies the given RC query is decidable~\cite{ail}.

Kifer~\cite{DBLP:conf/jcdkb/Kifer88} calls a query class
capturable if there is an algorithm that,
given a query in the class and a database instance,
enumerates the query's evaluation result, i.e., all tuples satisfying the query.
Avron and Hirshfeld~\cite{DBLP:conf/lics/AvronH91} observe that Kifer's notion is restricted because it requires every query
in a capturable class to be domain independent.
Hence, they propose an alternative definition that we also use: A query class is capturable if there is an algorithm that,
given a query in the class, a (finite or infinite) domain, and a database instance,
determines whether the query's evaluation result on the database instance over the domain is finite
and enumerates the result in this case.
Our work solves Avron and Hirshfeld's capturability problem
additionally assuming an infinite domain.

\looseness=-1
Data complexity~\cite{DBLP:conf/stoc/Vardi82} is the complexity of recognizing
if a tuple satisfies a fixed query
over a database, as a function of the database size.
Our capturability algorithm provides an upper bound on the data complexity
for RC queries over an infinite domain that have a finite evaluation result
(but it cannot decide if a tuple
belongs to a query's result if the result is infinite).

Next, we group related approaches to evaluating RC queries into three categories.

\textbf{Structure reduction.}  The classical approach
to handling arbitrary RC queries is to evaluate them under a finite
structure~\cite{DBLP:books/sp/Libkin04}. The core question here is whether the evaluation
produces the same result as defined by the natural semantics, which typically considers
infinite domains. Codd's theorem~\cite{DBLP:persons/Codd72}
affirmatively answers this question for domain-independent queries, restricting the structure to the
\emph{active domain}. Ailamazyan et al.~\cite{ail} show
that RC is a capturable query class by extending the active domain with a few additional elements, whose
number depends only on the query,
and evaluating the query over this finite domain.
\emph{Natural--active collapse} results~\cite{DBLP:journals/jacm/BenediktL00} generalize
Ailamazyan et al.'s~\cite{ail} result to extensions of RC
(e.g., with order relations)
by combining the structure reduction with a translation-based approach.
Hull and Su~\cite{DBLP:journals/acta/HullS94} study several semantics
of RC that guarantee the finiteness of the query's evaluation result.
In particular, the ``output-restricted unlimited interpretation''
only restricts the query's evaluation result to tuples
that only contain elements in the active domain,
but the quantified variables still range over the (finite or infinite)
underlying domain.
Our work is inspired by all these theoretical landmarks,
in particular~Hull and Su's work~(\S\ref{sec:var_restr}).
Yet we avoid using (extended) active domains, which make query evaluation impractical.% due to a prohibitively high time complexity.

\textbf{Query translation.} Another strategy is to translate a given query into one that can 
be evaluated efficiently, for example as a sequence of RA operations on finite tables. 
Van Gelder and Topor
pioneered this approach~\cite{DBLP:conf/pods/GelderT87,DBLP:journals/tods/GelderT91} for \RC{}. A core
component of their translation is the choice of generators, which replace the
active domain restrictions from structure reduction approaches and
thereby improve the time complexity.
Extensions to scalar and complex function symbols have also been
studied~\cite{DBLP:conf/pods/Escobar-MolanoHJ93, DBLP:journals/isci/LiuYL08}. All these approaches focus on syntactic classes
of RC, for which domain independence is given, e.g., the \emph{evaluable} queries of Van Gelder and Topor
%\begin{short}\cite[Definition 5.2]{DBLP:journals/tods/GelderT91}\end{short}%
(Appendix~\ref{sec:eval}). Our approach is inspired
by Van Gelder and Topor's work but generalizes it to handle arbitrary RC
queries at the cost of assuming an infinite domain. Also,
we further improve the time complexity of \citeauthorVanGT's approach
by choosing better generators.

\textbf{Evaluation with infinite relations.}  Constraint
databases~\cite{DBLP:books/sp/Revesz02} obviate the need for using RA operations on
finite tables. This yields significant expressiveness gains as domain independence need 
not be assumed. 
Yet the efficiency of the quantifier elimination procedures employed
cannot compare with the simple evaluation of the RA's projection operation. Similarly, 
automatic structures~\cite{DBLP:journals/mst/BlumensathG04} can represent the results 
of arbitrary RC queries finitely, but struggle with large quantities of data. 
We demonstrate this in our evaluation where we compare our translation to several modern incarnations of the 
above approaches, all based on binary decision diagrams~\cite{DBLP:conf/csl/MollerLAH99,DBLP:conf/cade/Moller02,DBLP:conf/fmcad/ChakiGS09,mona-tool,DBLP:journals/jacm/BasinKMZ15}.

\section{Preliminaries}\label{sec:prelim}

We introduce the \RC syntax and semantics and
define relevant classes of \RC queries.

\vspace{-1ex}
\subsection{Relational Calculus}\label{sec:rc}

A signature $\sigsymb$ is a triple $(\consts,\predsymbs,\arity)$,
where $\consts$ and $\predsymbs$ are disjoint finite sets of constant and predicate symbols,
and the function
$\arity:\predsymbs\rightarrow\nats$ maps each predicate symbol
$\predsymb\in \predsymbs$
to its arity $\arity(\predsymb)$. Let $\sigsymb=(\consts,\predsymbs,\arity)$
be a signature and $\varset$ a countably infinite set of variables
disjoint from $\consts\cup \predsymbs$. The following grammar defines the syntax of \RC queries:
\[
\qry ::= \bot\mid\top\mid \varx\approx \termsymb\mid \predsymb(\termsymb_1, \ldots, \termsymb_{\arity(\predsymb)})\mid
\neg \qry\mid \qry\lor \qry\mid \qry\land \qry\mid \exists \varx.\, \qry.
\]
Here, $\predsymb\in \predsymbs$ is a predicate symbol,
$\termsymb, \termsymb_1, \ldots, \termsymb_{\arity(\predsymb)}\in \varset\cup \consts$ are terms, and
$\varx\in \varset$ is a variable.
We write $\exists \varlist.\,\qry$ for $\exists \varlistname_1.\ldots\exists \varlistname_\listlength.\,\qry$ and $\forall \varlist.\,\qry$ for $\neg\exists \varlist.\,\neg \qry$,
where $\varlist$ is a variable sequence $\varlistname_1,\ldots, \varlistname_\listlength$.
If $\listlength = 0$, then both $\exists \varlist.\,\qry$ and $\forall \varlist.\,\qry$ denote just $\qry$.
Quantifiers have lower precedence than
conjunctions and disjunctions, e.g.,
$\exists\varx.\,\qrya\land\qryb$ means
$\exists\varx.\,(\qrya\land\qryb)$.
We use $\approx$ to denote the equality of terms in \RC{}
to distinguish it from $=$, which denotes syntactic object identity.
We also write $\qrya\longrightarrow\nobreak\qryb$
for $\neg\qrya\lor\qryb$.
However, writing $\qrya\lor \qryb$ for
$\neg (\neg \qrya\land\neg \qryb)$ would complicate later definitions, e.g., the safe-range queries (\S\ref{sec:eval_allowed}).

\looseness=-1
We define the subquery partial order $\sqsubseteq$ on queries
as the (reflexive and transitive) subterm relation
on the datatype of \RC{} queries. For example,
$\qrya$ is a subquery of the query $\qrya\land\neg\exists \vary.\,\qryb$.
We denote by $\subs{\qry}$ the set of subqueries of a query $\qry$,
by $\fv{\qry}$ the set of \emph{free variables} in $\qry$,
and by $\av{\qry}$ be the set of all \fbvs{} in a query $\qry$.
Furthermore, we denote by $\fvseq{\qry}$
the sequence of free variables in $\qry$
based on some fixed ordering of variables.
We lift this notation to sets of queries in the standard way.
A query $\qry$ with no free variables, i.e., $\fv{\qry}=\emptyset$,
is called \emph{closed}.
Queries of the form $\predsymb(\termsymb_1, \ldots, \termsymb_{\arity(\predsymb)})$ and $\varx\approx \cstsymb$
are called \emph{atomic predicates}.
We define the predicate $\edb{\cdot}$ characterizing atomic predicates,
i.e., $\edb{\qry}$ is true iff $\qry$ is an atomic predicate.
Queries of the form
$\exists\varlist.\,\predsymb(\termsymb_1, \ldots, \termsymb_{\arity(\predsymb)})$
and $\exists\varlist.\,\varx\approx \cstsymb$
are called \emph{quantified predicates}.
We denote by $\proj{\varx}{\qry}$ the query obtained by
existentially quantifying a variable $\varx$ from a query $\qry$
if $\varx$ is free in $\qry$, i.e.,
$\proj{\varx}{\qry} \eq \exists \varx.\,\qry$ if $\varx\in\fv{\qry}$
and $\proj{\varx}{\qry}\eq \qry$ otherwise.
We lift this notation to sets of queries in the standard way.
We use $\proj{\varx}{\qry}$ (instead of $\exists \varx.\,\qry$)
when constructing a query
to avoid introducing bound variables that never occur in $\qry$.

\looseness=-1
A structure $\str$ over a signature $(\consts, \predsymbs, \arity)$ consists of
a non-empty domain $\dom$
and interpretations $\cstsymb^\str\in \dom$
and $\predsymb^\str\subseteq \dom^{\arity(\predsymb)}$,
for each $\cstsymb\in \consts$ and $\predsymb\in \predsymbs$.
We assume that all the relations $\predsymb^\str$ are \emph{finite}.
Note that this assumption does \emph{not} yield a finite structure
(as defined in finite model theory~\cite{DBLP:books/sp/Libkin04}) since
the domain $\dom$ can still be infinite.
A (\emph{variable}) \emph{assignment} is a mapping
$\valsymb: \varset\rightarrow \dom$.
We extend~$\valsymb$ 
to constant symbols $\cstsymb\in \consts$ with $\valsymb(\cstsymb)=\cstsymb^\str$.
We write $\valsymb[\varx\mapsto \domval]$ for the assignment
that maps $\varx$ to $\domval \in \dom$ and is otherwise identical to $\valsymb$.
We lift this notation to sequences $\varxs$ and $\domvallist$
of pairwise distinct variables and arbitrary domain elements
of the same length.
The semantics of \RC queries for a structure $\str$ and an assignment $\valsymb$ is defined in Figure~\ref{fig:rc}.
\begin{figure}
\small
$\begin{array}{@{}l@{\;\;}c@{\;\;}l@{\;}|@{\;}l@{\;\;}c@{\;\;}l@{}}
\rulenocond{(\str, \valsymb) \not\models \bot;\, (\str, \valsymb) \models \top}&&
\ruleiff{(\str, \valsymb) \models (\varx \approx \termsymb)}{\valsymb(\varx)=\valsymb(\termsymb)}\\
\ruleiff{(\str, \valsymb) \models \predsymb(\termsymb_1, \ldots, \termsymb_{\arity(\predsymb)})}%
{(\valsymb(\termsymb_1), \ldots, \valsymb(\termsymb_{\arity(\predsymb)}))\in \predsymb^\str}&
\ruleiff{(\str, \valsymb) \models (\neg \qry)}%
{(\str, \valsymb) \not\models \qry}\\
\ruleiff{(\str, \valsymb) \models (\qrya\lor \qryb)}%
{(\str, \valsymb) \models \qrya\text{ or }(\str, \valsymb) \models \qryb}&
(\str, \valsymb) \models (\exists \varx.\,\qry)&\text{iff}&%
(\str, \valsymb[\varx\mapsto \domval]) \models \qry,\\
\ruleiff{(\str, \valsymb) \models (\qrya\land \qryb)}%
{(\str, \valsymb) \models \qrya\text{ and }(\str, \valsymb) \models \qryb}&
&&\text{for some }\domval\in \dom.\\
\end{array}$
\vspace{-1ex}
\caption{The semantics of \RC{}.}
\label{fig:rc}
%\vspace{-.8ex}
\end{figure}
We write $\valsymb\models \qry$ for $(\str, \valsymb)\models \qry$
if the structure $\str$ is fixed in the given context.
For a fixed $\str$, only
the assignments to $\qry$'s free variables influence $\valsymb\models \qry$, i.e.,
$\valsymb\models \qry$ is equivalent to $\valsymba\models \qry$,
for every variable assignment $\valsymba$
that agrees with $\valsymb$ on $\fv{\qry}$.
For closed queries $\qry$,
we write $\models \qry$ and say that $\qry$ holds, since closed
queries either hold for all variable assignments
or for none of them.
We call a finite sequence $\domvallist$ of domain elements
$\domval_1, \ldots\, \domval_\listlength\in\dom$
a \emph{tuple}.
Given a query $\qry$ and a structure $\str$,
we denote the set of satisfying tuples for $\qry$ by
\[
\sattup{\qry}^{\str}=\{\domvallist\in\dom^{\len{\fvseq{\qry}}}\mid
\text{there exists an assignment}\,\valsymb\,\text{such that}\,
(\str,\valsymb[\fvseq{\qry}\mapsto\domvallist])\models \qry\}.
\]
\looseness=-1
We omit $\str$ from $\sattup{\qry}^{\str}$
if $\str$ is fixed. We call
the values from $\sattup{\qry}$ assigned to $\varx\!\in\!\fv{\qry}$ \emph{column} $\varx$.

The \emph{active domain} $\adomd{\qry}{\str}$ of a query $\qry$ and a structure $\str$
is a subset of the domain $\dom$ containing the interpretations $\cstsymb^\str$
of all constant symbols that occur in $\qry$
and the values in the relations $\predsymb^\str$
interpreting all predicate symbols that occur in $\qry$. Since $\consts$ and $\predsymbs$ are
finite and all $\predsymb^\str$ are finite relations of a finite arity $\arity(\predsymb)$,
the active domain $\adomd{\qry}{\str}$ is also a finite set.
We omit $\str$ from $\adomd{\qry}{\str}$ if $\str$ is fixed in the given context.

Queries $\qrya$ and $\qryb$ over the same signature
are \emph{equivalent}, written $\qrya\equiv \qryb$, if
$(\str, \valsymb)\models \qrya\Longleftrightarrow
(\str, \valsymb)\models \qryb$, for every $\str$ and~$\valsymb$.
Queries $\qrya$ and $\qryb$ over the same signature
are \emph{\infequivprop}, written $\infequiv{\qrya}{\qryb}$, if
$(\str, \valsymb)\models \qrya\Longleftrightarrow
(\str, \valsymb)\models \qryb$, for every $\str$
with an \emph{infinite} domain $\dom$ and every $\valsymb$.
Clearly, equivalent queries are also \infequivprop{}.

A query $\qry$ is \emph{\di} if
$\sattup{\qry}^{\stra}=\sattup{\qry}^{\strb}$ holds for every two
structures $\stra$ and $\strb$ that agree on the interpretations of
constants ($\cstsymb^{\stra}=\cstsymb^{\strb}$) and predicates
($\predsymb^{\stra}=\predsymb^{\strb}$), while their domains $\doma$
and $\domb$ may differ. Agreement on the interpretations implies
$\adomd{\qry}{\stra}=\adomd{\qry}{\strb}\subseteq\doma \cap \domb$. It is undecidable whether an \RC
query is \di{}~\cite{DBLP:journals/jacm/Paola69a,DBLP:journals/ipl/Vardi81}.

\begin{figure}[t]
  %\vspace{-1ex}
  \small
  $\begin{array}{lcllcllcllcllcl}
    \varx\approx \varx&\equiv&\top,&\neg\bot&\equiv&\top,&\neg\top&\equiv&\bot,&
    \exists \varx.\,\bot&\equiv&\bot,&\exists \varx.\,\top&\equiv&\top,\\
    \qry\land\bot&\equiv&\bot,&\bot\land \qry&\equiv&\bot,&\qry\land\top&\equiv&\qry,&\top\land \qry&\equiv&\qry,\\
    \qry\lor\bot&\equiv&\qry,&\bot\lor \qry&\equiv&\qry,&\qry\lor\top&\equiv&\top,&\top\lor \qry&\equiv&\top.
  \end{array}$
  \vspace{-1ex}
  \caption{Constant propagation rules.}\label{fig:cp}
  \vspace{-1ex}
  \end{figure}
  
  \looseness=-1
  We denote by $\cp{\qry}$ the query obtained from a query $\qry$
  by exhaustively applying the rules in Figure~\ref{fig:cp}.
  Note that $\cp{\qry}$ is either of the form $\bot$ or $\top$
  or contains no $\bot$ or $\top$ subqueries.
  
  \begin{definition}
  The substitution of the form $\qry[\varx\mapsto \vary]$ is
  the query $\cp{\subqry}$, where $\subqry$ is obtained from a query~$\qry$
  by replacing all occurrences of the free variable $\varx$
  by the variable $\vary$, potentially also renaming bound
  variables to avoid capture.
  \end{definition}
  
  \begin{definition}
  The substitution of the form $\qry[\varx/\bot]$ is
  the query $\cp{\subqry}$, where $\subqry$ is obtained from a query $\qry$
  by replacing with $\bot$ every atomic predicate or equality
  containing the free variable $\varx$,
  except for $(\varx\approx \varx)$ which is replaced by $\top$.
  \end{definition}

We lift the substitution notation to sets of queries in the standard way.

The function $\flop{\qry}$,
where $\oplus\in\{\lor, \land\}$, computes a set of queries
by ``flattening'' the operator $\oplus$:
$\flop{\qry}\eq\flop{\qrya}\cup\flop{\qryb}$
if $\qry=\qrya\oplus \qryb$ and $\flop{\qry}\eq\{\qry\}$ otherwise.

\vspace{-1ex}
\subsection{Safe-Range Queries}\label{sec:eval_allowed}

\begin{figure}
  \begin{minipage}[t]{0.40\textwidth}
  \begin{figure}[H]
  \small
  $
    \begin{array}{@{}l@{\;}c@{\;}l@{}}
    \rulenocondgen{\tgen{\varx}{\bot}{\emptyset}}\\
    \ruleif{\tgen{\varx}{\qry}{\{\qry\}}}{\edb{\qry}\,\text{and}\,\free{\varx}{\qry}}\\
    \ruleif{\tgen{\varx}{\neg\neg \qry}{\qpreds}}{\tgen{\varx}{\qry}{\qpreds}}\\
    \ruleifgen{\tgen{\varx}{\neg (\qrya\lor \qryb)}{\qpreds}}{\tgen{\varx}{(\neg \qrya)\land(\neg \qryb)}{\qpreds}}\\
    \ruleifgen{\tgen{\varx}{\neg (\qrya\land \qryb)}{\qpreds}}{\tgen{\varx}{(\neg \qrya)\lor(\neg \qryb)}{\qpreds}}\\
    \ruleifgen{\tgen{\varx}{\qrya\lor \qryb}{\qpredsa\cup \qpredsb}}{\tgen{\varx}{\qrya}{\qpredsa}\,\text{and}\,\tgen{\varx}{\qryb}{\qpredsb}}\\
    \ruleifgen{\tgen{\varx}{\qrya\land \qryb}{\qpreds}}{\tgen{\varx}{\qrya}{\qpreds}\,\text{or}\,\tgen{\varx}{\qryb}{\qpreds}}\\
    \ruleifgen{\tgen{\varx}{\qry\land \varx\approx \vary}{\qpreds[\vary\mapsto \varx]}}{\tgen{\vary}{\qry}{\qpreds}}\\
    \ruleifgen{\tgen{\varx}{\qry\land \vary\approx \varx}{\qpreds[\vary\mapsto \varx]}}{\tgen{\vary}{\qry}{\qpreds}}\\
    \ruleifgendot{\tgen{\varx}{\exists \vary.\,\exqryy}{\proj{\vary}{\qpreds}}}{\varx\neq \vary\,\text{and}\,\tgen{\varx}{\exqryy}{\qpreds}}
    \end{array}$
    \caption{The \emph{generated} relation.}
    \label{fig:gen_con}
    \end{figure}
  \end{minipage}%
  \begin{minipage}[t]{0.6\textwidth}
  \begin{figure}[H]
  \small
  $
  \begin{array}{@{}l@{}c@{}c@{\;}l@{}}
  \tcov{\varx}{\varx \approx  \varx}{\emptyset};\\
  \flright{\tcov{\varx}{\qry}{\emptyset}}{\absent{\varx}{\qry}}\\
  \flleft{\tcov{\varx}{\varx \approx  \vary}{\{\varx \approx  \vary\}}}{\varx\neq \vary}\\
  \flleft{\tcov{\varx}{\vary \approx  \varx}{\{\varx \approx  \vary\}}}{\varx\neq \vary}\\
  \flleft{\tcov{\varx}{\qry}{\{\qry\}}}{\edb{\qry}\,\text{and}\,\free{\varx}{\qry}}\\
  \flleft{\tcov{\varx}{\neg \qry}{\cpreds}}{\tcov{\varx}{\qry}{\cpreds}}\\
  \flleft{\tcov{\varx}{\qrya\lor \qryb}{\cpredsa\cup \cpredsb}}{\tcov{\varx}{\qrya}{\cpredsa}\,\text{and}\,\tcov{\varx}{\qryb}{\cpredsb}}\\
  \flleft{\tcov{\varx}{\qrya\lor \qryb}{\cpreds}}{\tcov{\varx}{\qrya}{\cpreds}\,\text{and}\, \qrya[\varx/\bot]=\top}\\
  \flleft{\tcov{\varx}{\qrya\lor \qryb}{\cpreds}}{\tcov{\varx}{\qryb}{\cpreds}\,\text{and}\, \qryb[\varx/\bot]=\top}\\
  \flleft{\tcov{\varx}{\qrya\land \qryb}{\cpredsa\cup \cpredsb}}{\tcov{\varx}{\qrya}{\cpredsa}\,\text{and}\,\tcov{\varx}{\qryb}{\cpredsb}}\\
  \flleft{\tcov{\varx}{\qrya\land \qryb}{\cpreds}}{\tcov{\varx}{\qrya}{\cpreds}\,\text{and}\, \qrya[\varx/\bot]=\bot}\\
  \flleft{\tcov{\varx}{\qrya\land \qryb}{\cpreds}}{\tcov{\varx}{\qryb}{\cpreds}\,\text{and}\, \qryb[\varx/\bot]=\bot}\\
  \multicolumn{4}{@{}l@{}}{\tcov{\varx}{\exists \vary.\,\exqryy}{\proj{\vary}{\cpreds}}}\\
  \multicolumn{4}{@{}l@{}}{\qquad\text{if }\varx\neq \vary\,\text{and}\,\tcov{\varx}{\exqryy}{\cpreds}\,\text{and}\, (\varx\approx \vary)\notin \cpreds;}\\
  \multicolumn{4}{@{}l@{}}{\tcov{\varx}{\exists \vary.\,\exqryy}{\proj{\vary}{(\cpreds\setminus \{\varx\approx \vary\})}\cup \qpredsy[\vary\mapsto \varx]}}\\
  \multicolumn{4}{@{}l@{}}{\qquad\text{if }\varx\neq \vary\,\text{and}\, \tcov{\varx}{\exqryy}{\cpreds}\,\text{and}\,\tgen{\vary}{\exqryy}{\qpredsy}.}\\~
  \end{array}
  $
  \caption{The \emph{covered} relation.}
  \label{fig:cov_extra}
  \end{figure}
  \end{minipage}
  \end{figure}

The class of \emph{safe-range} queries~\cite{DBLP:books/aw/AbiteboulHV95} is
a decidable subset of \di{} \RC queries.
Its definition is based on the notion
of the range-restricted variables of a query.
A variable is called \emph{range restricted}
if ``its possible values all lie within the active domain of
the query''~\cite{DBLP:books/aw/AbiteboulHV95}. Intuitively, atomic predicates restrict the possible
values of a variable that occurs in them as a term.
%A similar idea applies to 
%equalities between a variable and a constant.
An equality $\varx\approx \vary$ can extend the set of 
range-restricted variables in a conjunction $\qry\land \varx\approx \vary$:
If $\varx$ or $\vary$ is range restricted in $\qry$,
then both $\varx$ and $\vary$ are range restricted in $\qry\land \varx\approx \vary$.

\looseness=-1
We formalize range-restricted variables
using the \emph{generated} relation $\tgen{\varx}{\qry}{\qpreds}$,
defined in Figure~\ref{fig:gen_con}.
Specifically, $\tgen{\varx}{\qry}{\qpreds}$ holds
if $\varx$ is a range-restricted variable in $\qry$
and every satisfying assignment for $\qry$
satisfies some quantified predicate, referred to as \emph{generator}, from $\qpreds$.
A similar definition by Van Gelder and Topor~\cite[Figure~5]{DBLP:journals/tods/GelderT91}
uses a set of atomic (not quantified) predicates $\apreds$ as generators and defines the rule
%\begin{short}$\tgen{\varx}{\exists \vary.\,\exqryy}{\qpreds}$\end{short}%
$\vgen{\varx}{\exists \vary.\,\exqryy}{\apreds}$
if $\varx\neq\vary$ and
%\begin{short}$\tgen{\varx}{\exqryy}{\qpreds}$\end{short}%
$\vgen{\varx}{\exqryy}{\apreds}$
(Appendix~\ref{sec:eval}, Figure~\ref{fig:gen_con_vgt}).
In contrast, we modify the rule's conclusion
to existentially quantify the variable $\vary$
in all queries in $\qpreds$ where $\vary$ is free:
$\tgen{\varx}{\exists \vary.\,\exqryy}{\proj{\vary}{\qpreds}}$.
Hence, $\tgen{\varx}{\qry}{\qpreds}$ implies $\fv{\qpreds}\subseteq\fv{\qry}$.
We now formalize these relationships.

\begin{lemC}[{\isabelleqed}]\label{lem:gen_wit}
Let $\qry$ be a query, $\varx \in \fv{\qry}$, and $\qpreds$ be a set of quantified predicates
such that $\tgen{\varx}{\qry}{\qpreds}$.
Then (i) for every $\projoneatom\in \qpreds$,
we have $\varx\in\fv{\projoneatom}$ and $\fv{\projoneatom}\subseteq\fv{\qry}$, (ii)~for every $\valsymb$ such that $\valsymb\models \qry$,
there exists a $\projoneatom\in \qpreds$ such that $\valsymb\models \projoneatom$,
and (iii)~$\qry[\varx/\bot]=\bot$.
\end{lemC}

\begin{definition}\label{def:st}
\looseness=-1
Let $\gen{\varx}{\qry}$ hold iff 
$\tgen{\varx}{\qry}{\qpreds}$ holds for some $\qpreds$.
Let $\gens{\qry}\eq\{\varx\in\fv{\qry} \mid \gen{\varx}{\qry}\,\text{does not hold}\}$
be the set of free variables in a query $\qry$
that are not range restricted.
A query $\qry$ has \emph{range-restricted free variables} if every free variable
of $\qry$ is range restricted, i.e., $\gens{\qry}=\emptyset$.
A query $\qry$ has \emph{range-restricted bound variables}
if the bound variable $\vary$ in every subquery $\exists \vary.\,\exqryy$ of $\qry$
is range restricted, i.e., $\gen{\vary}{\exqryy}$ holds.
A query is \emph{safe range} if it has range-restricted free
and range-restricted bound variables.
\end{definition}

\vspace{-1ex}
\subsection{Safe-Range Normal Form}\label{sec:srnf}

\looseness=-1
A query $\qry$ is in safe-range normal form (SRNF) if
the query $\subqry$ in every subquery $\neg\subqry$ of $\qry$
is an atomic predicate, equality,
or an existentially quantified query~\cite{DBLP:books/aw/AbiteboulHV95}.
In \S\ref{sec:srnf-detail} we define function $\pushnot{\qry}$
that returns a SRNF query equivalent to a query $\qry$.
Intuitively, the function $\pushnot{\qry}$ proceeds by
pushing negations downwards~\cite[\S 5.4]{DBLP:books/aw/AbiteboulHV95},
distributing existential quantifiers over disjunction~\cite[Rule~(T9)]{DBLP:journals/tods/GelderT91},
and dropping bound variables that never occur~\cite[Definition~9.2]{DBLP:journals/tods/GelderT91}.
We include the last two rules to optimize the time complexity of evaluating
the resulting query.

If a query $\qry$ is safe range, then $\pushnot{\qry}$ is also safe range.

\vspace{-1ex}
\subsection{Relational Algebra Normal Form}\label{sec:ranf}
\begin{figure}[t]
  \[
  \begin{array}{@{}l@{\;}c@{\;}l@{}}
  \rulenocond{\safe{\bot};\;\safe{\top}}\\
  \ruleif{\safe{\qry}}{\edb{\qry}}\\
  \ruleif{\safe{\neg \qry}}{\safe{\qry}\,\text{and}\,\fv{\qry}=\emptyset}\\
  \ruleif{\safe{\qrya\lor \qryb}}{\safe{\qrya}\,\text{and}\,\safe{\qryb}\,\text{and}\,\fv{\qrya}=\fv{\qryb}}\\
  \ruleif{\safe{\qrya\land \qryb}}{\safe{\qrya}\,\text{and}\,\safe{\qryb}}\\
  \ruleif{\safe{\qrya\land \neg \qryb}}{\safe{\qrya}\,\text{and}\,\safe{\qryb}\,\text{and}\,\fv{\qryb}\subseteq\fv{\qrya}}\\
  \ruleif{\safe{\qry\land (\varx\approx \vary)}}{\safe{\qry}\,\text{and}\, \{\varx, \vary\}\cap\fv{\qry}\neq\emptyset}\\
  \ruleif{\safe{\qry\land \neg (\varx \approx  \vary)}}{\safe{\qry}\,\text{and}\, \{\varx, \vary\}\subseteq\fv{\qry}}\\
  \ruleifdot{\safe{\exists \varx.\,\exqryx}}{\safe{\exqryx}\,\text{and}\, \varx\in\fv{\exqryx}}\\
  \end{array}
  \]
  \vspace{-1.2em}
  \caption{Characterization of \safeprop{} queries.}
  \label{fig:safe}
\end{figure}

Relation algebra normal form (RANF)
is a class of safe-range queries that can be easily mapped to RA%
 \cite{DBLP:books/aw/AbiteboulHV95}
and evaluated using the RA operations
for projection, column duplication,
selection, set union, binary join, and anti-join.

Figure~\ref{fig:safe} defines the predicate $\safe{\cdot}$ 
characterizing RANF queries.
The translation of safe-range queries (\S\ref{sec:eval_allowed})
to equivalent \safeprop{} queries proceeds via SRNF (\S\ref{sec:srnf}).
A safe-range query in SRNF can be translated to an equivalent RANF
query by subquery rewriting using the following 
rules~\cite[Algorithm~5.4.7]{DBLP:books/aw/AbiteboulHV95}:
\[
\begin{array}{lclr}
\qry\land(\qrya\lor\qryb)&\equiv&(\qry\land\qrya)\lor(\qry\land\qryb),&
\rulea\\
\qry\land(\exists \varx.\,\exqryx)&\equiv&(\exists \varx.\,\qry\land\exqryx),&\ruleb\\
\qry\land\neg\subqry&\equiv&\qry\land\neg(\qry\land\subqry).&\rulec
\end{array}
\]
Subquery rewriting is a nondeterministic process
(as the rewrite rules can be applied in an arbitrary order)
that impacts the performance of evaluating the resulting RANF query.
We translate a safe-range query in SRNF to an equivalent RANF query
by a recursive function $\allowtosafeqry{\cdot}$  
inspired by the rules $\rulea\text{--}\rulec$ and
fully specified in Figure~\ref{alg:allowtosafe} in \S\ref{sec:ranf-detail}.

\vspace{-1ex}
\subsection{Query Cost}\label{sec:cost}

To assess the time complexity of evaluating a \safeprop{} query $\qry$,
we define the \emph{cost} of $\qry$ over a structure $\str$,
denoted $\cost{\qry}{\str}$,
to be the sum of intermediate result sizes over all \safeprop{}
subqueries of $\qry$.
Formally,
$\cost{\qry}{\str}\eq\sum_{\subqry\sqsubseteq \qry,\ \safe{\subqry}}\card{\sattupd{\subqry}{\str}}\cdot\card{\fv{\subqry}}$.
This corresponds to evaluating $\qry$ following its \safeprop{} structure
(\S\ref{sec:ranf}, Figure~\ref{fig:safe})
using the RA operations.
The complexity of these operations
is linear in the combined input and output size (ignoring
logarithmic factors due to set operations).
The output size (the number of tuples times the number of variables)
is counted in
$\card{\sattupd{\subqry}{\str}}\cdot \card{\fv{\subqry}}$
and the input size is counted as the output size
for the input subqueries.
Repeated subqueries are only considered once, 
which does not affect the asymptotics of query cost. In practice,
the evaluation results for common subqueries can be reused.

\section{Query Translation}\label{sec:translation}

\looseness=-1
Our approach to evaluating
an arbitrary \RC query $\qry$ over a fixed structure $\str$ with an infinite domain $\dom$ proceeds by translating $\qry$ into
a pair of safe-range queries $(\qryfin, \qryinf)$ such that
\begin{enumerate}[(\propfin)]
\item[(\propfv)] \propfvlab $\fv{\qryfin}=\fv{\qry}$ unless $\qryfin$ is syntactically equal to $\bot$; $\fv{\qryinf}=\emptyset$;
\item[(\propfin)] \propfinlab
$\sattup{\qry}$ is an infinite set if $\qryinf$ holds;
otherwise $\sattup{\qry}=\sattup{\qryfin}$
is a finite set.
\end{enumerate}
Since the queries $\qryfin$ and $\qryinf$ are safe range,
they are \di{} and thus
$\sattup{\qryfin}$ is a finite set. In particular,
$\sattup{\qry}$ is a finite set
if $\qryinf$ does not hold.
Our translation generalizes Hull and Su's case distinction
that restricts bound variables~\cite{DBLP:journals/acta/HullS94}
to restrict all variables.
Moreover, we use Van Gelder and Topor's idea
to replace the active domain
by a smaller set (generator) specific to each
variable~\cite{DBLP:journals/tods/GelderT91}
while further improving the generators.
Unless explicitly noted, in the rest of the article we assume a fixed structure $\str$.

\subsection{Restricting One Variable}\label{sec:var_restr}

Let $\varx$ be a free variable in a query $\qryrb$
with range-restricted bound variables.
This assumption on $\qryrb$
will be established by translating an arbitrary query $\qry$ bottom-up
(\S\ref{sec:bound}).
In this section, we develop
a translation of $\qryrb$
into an equivalent query $\qryrbx$ that satisfies the following:
\begin{itemize}
\item $\qryrbx$ has range-restricted bound variables;
\item $\qryrbx$ is a disjunction; $\varx$ is range restricted in the first disjunct; the remaining disjuncts are all binary conjunctions of a query not containing $\varx$ with a query of a special form containing $\varx$.
The special form, central to our translation,
is either an equality $\varx\approx\vary$
or a query satisfied by infinitely many values of $\varx$
for all values of the remaining free variables.
\end{itemize}

\newcommand{\mystrut}{\vphantom{\qryrb\bigvee_{\oneatom\in\qpreds}}}
\newcommand{\tspace}{6pt}
\def\myskip{3pt}
\looseness=-1
We now restate Hull and Su's~\cite{DBLP:journals/acta/HullS94}
and Van Gelder and Topor's~\cite{DBLP:journals/tods/GelderT91} 
approaches %to restricting a query's variables 
using our notation in order to clarify how we generalize both approaches. 
% To clarify how we generalize Hull and Su's~\cite{DBLP:journals/acta/HullS94}
% and Van Gelder and Topor's~\cite{DBLP:journals/tods/GelderT91}
% approaches to restricting a query's variables, we first restate their approaches using 
% our notation.
In particular, Hull and Su's approach is already stated in a 
generalized way that restricts a \emph{free} variable.

\paragraph{Hull and Su.}
    Let $\qryrb$ be a query
    with range-restricted bound variables and $\varx \in \fv{\qryrb}$.
    Then\vspace*{\tspace}
\[
\begin{array}{@{}l@{\;}c@{\;}l@{\;}l@{}l@{}}
\mystrut\qryrb&\equiv&\left(\qryrb\land{}
\adomqry{\varx}{\qryrb}\right)\lor%{}\\[1.5\jot]\mystrut&&
\left(\bigor{\vary\in \fv{\qryrb}\setminus\{\varx\}} (\qryrb[\varx\mapsto \vary]\land\varx\approx \vary)\right)\lor{}\\[1.5\jot]
\mystrut&&\left(\qryrb[\varx/\bot]\land\neg(
\adomqry{\varx}{\qryrb}\lor
\bigor{\vary\in \fv{\qryrb}\setminus\{\varx\}} \varx\approx \vary)\right).
\end{array}
\]
\looseness=-1
Here $\adomqry{\varx}{\qryrb}$ stands for an \RC query
with a single free variable $\varx$
that is satisfied by an assignment $\valsymb$
if and only if $\valsymb(\varx)\in\adom{\qryrb}$.
Hull and Su's translation distinguishes
the following three cases for a fixed assignment $\valsymb$ 
(each corresponding to a top-level disjunct above):
\begin{itemize}
    \item if $\adomqry{\varx}{\qryrb}$ holds (and hence $\valsymb(\varx)\in\adom{\qryrb}$), then we do not alter the query $\qryrb$;
    \item if $\varx\approx \vary$ holds for some free variable $\vary\in\fv{\qryrb}\setminus\{\varx\}$,
    then $\varx$ can be replaced by $\vary$ in $\qryrb$;
    \item \looseness=-1 
    otherwise, $\qryrb$ is equivalent to $\qryrb[\varx/\bot]$. Specifically, all 
    atomic predicates having $\varx$ free
    can be replaced by $\bot$ (as $\valsymb(\varx)\notin\adom{\qryrb}$),
    all equalities $\varx\approx \vary$ and $\vary\approx \varx$ for $\vary\in\fv{\qryrb}\setminus\{\varx\}$
    can be replaced by $\bot$ (as $\valsymb(\varx)\neq\valsymb(\vary)$),
    and all equalities $\varx\approx \varz$ for a bound variable $\varz$
    can be replaced by $\bot$ (as $\valsymb(\varx)\notin\adom{\qryrb}$
    and $\varz$ is range restricted
    in its subquery $\exists \varz.\,\exqryz$, by assumption). In the last case,
    $\gen{\varz}{\exqryz}$ holds and thus, for all $\valsymba$ extending $\valsymb$, we have
    $\valsymba\models\exists \varz.\,\exqryz$ if and only if
    there exists a $\domval\in\adom{\exqryz}\subseteq\adom{\qryrb}$ such that $\valsymba[\varz\mapsto\domval]\models\exqryz$.
\end{itemize}

\begin{example}
Consider the query $\qry\eq \predbrand(\vary)\lor \varx\approx \vary$. Then $\adomqry{\varx}{\qry} = \predbrand(\varx)$ and following Hull and Su we obtain that $\qry$ is equivalent to the disjunction of the following three queries:
\begin{itemize}
\item $(\predbrand(\vary)\lor \varx\approx \vary) \land \predbrand(\varx)$, which is equivalent to $\predbrand(\varx) \land \predbrand(\vary)$;
\item $(\predbrand(\vary)\lor \varx\approx \vary)[\varx\mapsto \vary] \land \varx\approx \vary$, which is syntactically equal to $\varx\approx \vary$ due to constant propagation that is part of the substitution operator;
\item $(\predbrand(\vary)\lor \varx\approx \vary)[\varx/\bot] \land \lnot (\predbrand(\varx) \lor \varx\approx \vary)$, which is equivalent to $\predbrand(\vary) \land \lnot \predbrand(\varx) \land \lnot \varx\approx \vary$.
\end{itemize}
Note that in this example, each disjunct covers a different subset of $\qry$'s satisfying assignments and all three disjuncts are necessary to cover all of $\qry$'s satisfying assignments.
\end{example}

\looseness=-1
\paragraph{Van Gelder and Topor.}
Let $\qryrb$ be an \emph{evaluable} query
with range-restricted bound variables, $\varx \in \fv{\qryrb}$.
Then there exists
a set $\apreds$ of atomic predicates
such that
\[
\begin{array}{@{}l@{\;}c@{\;}l@{\;}l@{}l@{}}
\mystrut\qryrb&\equiv&\left(\qryrb\land
\bigvee_{\oneatom\in\apreds}\exists \fvseq{\oneatom}\setminus
\{\varx\}.\;\oneatom\right)\lor\left(\qryrb[\varx/\bot]\right).
\end{array}
\]
Note that $\exists\fvseq{\qry}\setminus\{\varx\}.\,\qry$ is the query
in which all free variables of $\qry$ except for $\varx$ are existentially quantified.
Van Gelder and Topor restrict their attention to evaluable queries, which do not contain equalities between variables. (They only discuss an incomplete approach to supporting such equalities~\cite[Appendix A]{DBLP:journals/tods/GelderT91}.) Thus, their translation lacks the corresponding disjuncts that Hull and Su have.

%Given a set of atomic predicates $\apreds$, we write
%$\exists\vec{\valsymb}.\,\apreds$ for $\bigor{\oneatom\in \apreds} \exists\vec{\valsymb}.\,\oneatom$.
To avoid enumerating the entire active domain $\adom{\qryrb}$,
Van Gelder and Topor replace the query $\adomqry{\varx}{\qryrb}$ used by Hull and Su by the query $\bigvee_{\oneatom\in\apreds}\exists \fvseq{\oneatom}\setminus
\{\varx\}.\;\oneatom$ constructed from the atomic predicates from $\apreds$.
Because their translation must yield an equivalent query
(for every finite or infinite domain), $\apreds$ and $\qryrb$ must satisfy, for all $\valsymb$,
\[
\begin{array}{@{}l@{\;}c@{\;}l@{\;\;}l@{\;\;\;\;}l@{}}
\valsymb\models\neg\bigvee_{\oneatom\in\apreds}\exists \fvseq{\oneatom}\setminus
\{\varx\}.\;\oneatom&\Longrightarrow&
(\valsymb\models\qryrb\Longleftrightarrow\valsymb\models \qryrb[\varx/\bot])&\ethz&\text{and}\\
\valsymb\models\qryrb[\varx/\bot]&\Longrightarrow& \valsymb\models\forall \varx.\,\qryrb&\vgtrw.
\end{array}
\]
Note that $\vgtrw{}$
does not hold for the query $\qryrb\eq \neg \predbrand(\varx)$
and thus Van Gelder and Topor
only consider a proper subset of all \RC{} queries, called evaluable. For evaluable queries,
Van Gelder and Topor use the \emph{constrained} relation
%\begin{short}$\mathsf{con}$ \cite[Figure~5]{DBLP:journals/tods/GelderT91}\end{short}%
$\con{\varx}{\qry}{\apreds}$, defined in Appendix~\ref{sec:eval}, Figure~\ref{fig:gen_con_vgt}, to construct a set of atomic predicates $\apreds$ that satisfies $\ethz$.
% such a set might also exist for queries that are not syntactically evaluable

\looseness=-1
\paragraph{Our Translation.}
Let $\qryrb$ be a query
with range-restricted bound variables, $\varx \in \fv{\qryrb}$.
Then there exists a set $\apreds$ of atomic predicates and a set of equalities $\eqs$ 
such that
\[
\begin{array}{@{}l@{\;}c@{\;}l@{\;}l@{}l@{}}
    \mystrut\qryrb&\equiv&\left(\qryrb\land
    \bigvee_{\oneatom\in\apreds}\exists \fvseq{\oneatom}\setminus
    \fv{\qryrb}.\;\oneatom\right)\lor%{}\\[1.5\jot]\mystrut&&
    \left(\bigor{\varx\approx\vary\,\in\,\eqs} (\qryrb[\varx\mapsto \vary]\land\varx\approx \vary)\right)\lor{}\\[1.5\jot]
    \mystrut&&\left(\qryrb[\varx/\bot]\land\neg(
    (\bigvee_{\oneatom\in\apreds}\exists \fvseq{\oneatom}\setminus
    \fv{\qryrb}.\;\oneatom)
    {}\lor
    \bigor{\varx \approx \vary\,\in\,\eqs}\, \varx\approx \vary)\right).
\end{array}
\]
In contrast to Van Gelder and Topor, we only require that $\apreds$ satisfies $\ethz$
in our translation, which also allows us to translate non-evaluable queries, such as $\qryrb := \lnot \predbrand(\varx)$ above.
Note that we also existentially quantify only these variables that are not free in $\qryrb$,
%($\exists \fvseq{\oneatom}\setminus \fv{\qryrb}.\;\oneatom$), 
whereas Van Gelder and Topor quantify all variables except $\varx$. For our introductory example $\qsusp$, this modification allows our translation to use the quantified predicate $\exists \varprod.\;\predscore(\varprod, \varuser, \varscore)$ to restrict both $\varuser$ and $\varprod$ simultaneously. In contrast, Van Gelder and Topor's approach restricts them separately using $\exists \varprod, \varuser.\;\predscore(\varprod, \varuser, \varscore)$ and $\exists \varprod, \varscore.\;\predscore(\varprod, \varuser, \varscore)$, so that the Cartesian product of these quantified predicates may need to be computed in their translated queries.

In contrast to Hull and Su, we do not consider the equalities of $\varx$ with all other free 
variables in $\qryrb$, but only such equalities $\eqs$ that occur in $\qryrb$. 
%($\exists \fvseq{\oneatom}\setminus \{\varx\}.\;\oneatom$).
We jointly compute the sets $\apreds$ and $\eqs$ using the
\emph{covered} relation $\tcov{\varx}{\qry}{\cpreds}$ 
(in contrast to %\emph{constrained} relation
%\begin{short}$\mathsf{con}$ \cite[Figure~5]{DBLP:journals/tods/GelderT91}\end{short}%
$\con{\varx}{\qry}{\apreds}$ relation). %defined in Appendix~\ref{sec:eval}, Figure~\ref{fig:gen_con_vgt})
Figure~\ref{fig:cov_extra} shows the definition of
this relation.
The set $\qpreds$ computed by the covered relation contains atomic predicates 
that satisfy $\ethz$ and are already quantified as described above. 
The set also contains the relevant equalities that can be used in our translation.
% I don't see how the next sentence refines the previous one => dropping More precisely
For every variable $\varx$ and query $\qryrb$
with range-restricted bound variables,
there exists at least one set of quantified predicates and equalities $\cpreds$ such that $\tcov{\varx}{\qryrb}{\cpreds}$
and $\ethz$ holds for the set of atomic predicate subqueries in $\cpreds$ (i.e., 
for $\{\subqry\mid \edb{\subqry} \land \exists \qry\in\qpreds.\; \subqry\sqsubseteq\qry\}$).
%Unlike the generator set $\qpreds$ in $\tgen{\varx}{\qry}{\qpreds}$, 
As the \emph{cover} set $\cpreds$ in $\tcov{\varx}{\qry}{\cpreds}$ may contain both
quantified predicates and equalities between two variables, we define a function $\colpreds{\cpreds}$ that collects
all \emph{generators}, i.e., quantified predicates
and a function $\coleqs{\varx}{\cpreds}$ that collects
all \emph{variables} $\vary$ distinct from $\varx$ occurring
in equalities of the form $\varx\approx \vary$.
We use $\colpredsqry{\cpreds}$ to denote
the query $\bigor{\projoneatom\in\colpreds{\cpreds}} \projoneatom$.
We state the soundness and completeness of the relation
$\tcov{\varx}{\qryrb}{\cpreds}$
in the next lemma, which follows by induction
on the derivation of $\tcov{\varx}{\qryrb}{\cpreds}$.

\begin{lemC}[{\isabelleqed}]\label{lem:cov_complete}
Let $\qryrb$ be a query
with range-restricted bound variables, $\varx \in \fv{\qryrb}$.
\begin{description}
\item[Completeness] Then there exists a set $\cpreds$
of quantified predicates and equalities
such that
$\tcov{\varx}{\qryrb}{\cpreds}$ holds and,
\item[Soundness] for any $\cpreds$ satisfying $\tcov{\varx}{\qryrb}{\cpreds}$
and all $\valsymb$,
\[
\valsymb\models\neg(\colpredsqry{\cpreds}\lor\bigor{\vary\in\coleqs{\varx}{\cpreds}} \varx\approx \vary)
\Longrightarrow
(\valsymb\models\qryrb \Longleftrightarrow \valsymb\models\qryrb[\varx/\bot]).
\]
\end{description}
\end{lemC}
Finally, to preserve the dependencies between the variable $\varx$
and the remaining free variables of $\qry$
occurring in the quantified predicates from $\colpreds{\cpreds}$,
we do not project $\colpreds{\cpreds}$ on the single variable $\varx$,
i.e., we restrict $\varx$ by $\colpredsqry{\cpreds}$
instead of $\exists\fvseq{\qry}\setminus\{\varx\}.\,\colpreds{\cpreds}$
as by Van Gelder and Topor.
From Lemma~\ref{lem:cov_complete},
we derive our optimized translation characterized by the following lemma.
\begin{lemC}[{\isabelleqed}]\label{lem:rw}
Let $\qryrb$ be a query
with range-restricted bound variables, $\varx\in\fv{\qryrb}$,
and $\cpreds$ be such that
$\tcov{\varx}{\qryrb}{\cpreds}$ holds.
Then $\varx\in\fv{\projoneatom}$ and $\fv{\projoneatom}\subseteq\fv{\qryrb}$,
for every $\projoneatom\in \colpreds{\cpreds}$,
and
\begin{equation}
\begin{array}{@{}l@{\,}l@{}}
\qryrb\equiv&\left(\qryrb\land \colpredsqry{\cpreds}\right)\lor
\left(\bigor{\vary\in\coleqs{\varx}{\cpreds}} (\qryrb[\varx\mapsto \vary]\land \varx\approx \vary)\right)\lor {}\\
&\left(\qryrb[\varx/\bot]\land\neg(\colpredsqry{\cpreds}\lor
\bigor{\vary\in\coleqs{\varx}{\cpreds}} \varx\approx \vary)\right).\tag{\transnew}
\end{array}
\end{equation}
\end{lemC}

\looseness=-1
Note that $\varx$ is only guaranteed to be range restricted
in (\transnew)'s first disjunct.
However, it only occurs in the remaining disjuncts in subqueries of a special form
that are conjoined at the top-level to the disjuncts.
These subqueries of a special form are equalities of the form $\varx\approx\vary$
or negations of a disjunction of quantified predicates with a free occurrence of $\varx$
and equalities of the form $\varx\approx \vary$.
We will show how to handle such occurrences
in \S\ref{sec:bound} and \S\ref{sec:free}.
Moreover, the negation of the disjunction can be omitted
if $\vgtrw$ holds.

\subsection{Restricting Bound Variables}\label{sec:bound}

Let $\varx$ be a free variable in a query $\qryrb$ with range-restricted bound variables.
Suppose that the variable $\varx$ is not range restricted, i.e.,
$\gen{\varx}{\qryrb}$ does not hold.
To translate $\exists \varx.\,\qryrb$ into an \infequivprop{} query
with range-restricted bound variables ($\exists \varx.\,\qryrb$ does not have range-re\-strict\-ed bound variables
precisely because $\varx$ is not range restricted in $\qryrb$),
we first apply (\transnew) to $\qryrb$
and distribute the existential quantifier binding $\varx$ over disjunction. 
Next we observe that
\[
\exists \varx.\,(\qryrb[\varx\mapsto \vary]\land \varx\approx \vary)
\equiv \qryrb[\varx\mapsto \vary]\land\exists \varx.\,(\varx\approx \vary)
\equiv \qryrb[\varx\mapsto \vary],
\]
\noindent where the first equivalence follows because $\varx$ does not occur
free in $\qryrb[\varx\mapsto \vary]$ and the second equivalence follows
from the straightforward validity of $\exists \varx.\,(\varx\approx \vary)$.
Moreover, we observe the following inf-equivalence (recall: an equivalence that holds for infinite domains only): 
\[
\exists \varx.\,(\qryrb[\varx/\bot]\land \neg(\colpredsqry{\cpreds}\lor\bigor{\vary\in\coleqs{\varx}{\cpreds}} \varx\approx \vary))\infequivsign \qryrb[\varx/\bot]
\]
\looseness=-1
because $\varx$ is not free in $\qryrb[\varx/\bot]$
and there exists a value $\domval$ for $\varx$ in the infinite domain $\dom$
such that $\varx\neq \vary$ holds for all finitely many $\vary\in\coleqs{\varx}{\cpreds}$
and $\domval$ is not among the finitely many values interpreting
the quantified predicates in $\colpreds{\cpreds}$.
Altogether, we obtain the following lemma.
\begin{lemC}[{\isabelleqed}]\label{lem:rw_bound}
Let $\qryrb$ be a query
with range-restricted bound variables, $\varx \in \fv{\qryrb}$, and $\cpreds$ be a set of quantified predicates and equalities such that
$\tcov{\varx}{\qryrb}{\cpreds}$ holds. Then
\[
\exists \varx.\,\qryrb\infequivsign\left(\exists \varx.\,\qryrb\land \colpredsqry{\cpreds}\right)\lor\left(\bigor{\vary\in\coleqs{\varx}{\cpreds}} (\qryrb[\varx\mapsto \vary])\right)\lor \left(\qryrb[\varx/\bot]\right).\tag{\transex}
\]
\end{lemC}

\begin{figure}[t]
\begin{minipage}{0.5\linewidth}
\begin{algorithm}[H]
  \SetKwInOut{Input}{input}
  \SetKwInOut{Output}{output}
  \SetKwProg{Fn}{function}{$\,=$}{}
  \SetKwProg{lFn}{auxiliary function}{$\,=$}{}
  \Input{An \RC{} query $\qry$.}
  \Output{A query $\qryrb$ with range-restricted bound variables
  such that $\qry\infequivsign{} \qryrb$.}
  \BlankLine
  \lFn{$\loopbound{\disjs}{\varx}$}{$\loopbounddef{\disjs}{\varx}$}
  \BlankLine
  \Fn{$\allowed{\qry}$}{
  \Switch{$\qry$}{
    \lCase{$\neg \subqry$}{\Return $\neg\allowed{\subqry}$}
    \lCase{$\subqrya\lor \subqryb$}{\Return $\allowed{\subqrya}\lor\allowed{\subqryb}$}
    \lCase{$\subqrya\land \subqryb$}{\Return $\allowed{\subqrya}\land\allowed{\subqryb}$}
    \Case{$\exists \varx.\,\exqryx$}{\label{alg:allowed:exflat}
      $\disjs\eq\fldisj{\allowed{\exqryx}}$\;\label{alg:allowed:rec}
      \While{$\loopbound{\disjs}{\varx}\neq\emptyset$}{
      \label{alg:allowed:loopbeg}
        $\fixqry\leftarrow\loopbound{\disjs}{\varx}$\;
        $\cpreds\gets \{\cpreds\mid\tcov{\varx}{\fixqry}{\cpreds}\}$\;\label{alg:allowed:c}
        $\disjs\eq(\disjs\setminus\{\fixqry\})\cup\allowbreak
        \{\fixqry\land\colpredsqry{\cpreds}\}
        \cup\allowbreak\bigcup_{\vary\in\coleqs{\varx}{\cpreds}} \{\fixqry[\varx\mapsto \vary]\}\cup\allowbreak
        \{\fixqry[\varx/\bot]\}$\;
        \label{alg:allowed:ext}
      }\label{alg:allowed:loopend}
      \Return $\bigor{\qryrb\in\disjs} \proj{\varx}{\qryrb}$\;
    }
    \lOther{\Return $\qry$}
  }}
\caption{Restricting bound variables.}
    \label{alg:allowed}
\end{algorithm}
\end{minipage}%
\hspace*{-0.045\linewidth}
\begin{minipage}{0.576\linewidth}
\begin{algorithm}[H]
  \SetKwInOut{Input}{input}
  \SetKwInOut{Output}{output}
  \SetKwProg{Fn}{function}{$\,=$}{}
  \SetKwProg{lFn}{auxiliary function}{$\,=$}{}
  \Input{An \RC{} query $\qry$.}
  \Output{Safe-range query pair $(\qryfin, \qryinf)$
    for which (\propfvref{}) and (\propfinref{}) hold.}
  \BlankLine
  \lFn{$\freeloopa{\setfin}$}{$\freeloopadef{\setfin}$}
  \lFn{$\freeloopb{\setfin}{\qry}$}{$\freeloopbdef{\setfin}{\qry}$}
  \BlankLine
  \Fn{$\splitq{\qry}$}{
  $\setfin\eq\{(\allowed{\qry}, \emptyset)\}; \setinf\eq\emptyset$\;
  \While{$\freeloopa{\setfin}\neq\emptyset$}{\label{alg:free:loopbeg}
    $(\fixqry, \eqconjqry)\leftarrow \freeloopa{\setfin}$\;
    \label{alg:free:choosequery}
    $\varx\leftarrow \gens{\fixqry}$\;\label{alg:free:choosevar}
    $\cpreds\leftarrow \{\cpreds\mid\tcov{\varx}{\fixqry}{\cpreds}\}$\;\label{alg:free:chooserestr}
    $\setfin\eq(\setfin\setminus\{(\fixqry, \eqconjqry)\})\cup\allowbreak
    \{(\fixqry\land \colpredsqry{\cpreds}, \eqconjqry)\}\cup\allowbreak
    \bigcup_{\vary\in\coleqs{\varx}{\cpreds}} \{(\fixqry[\varx\mapsto \vary], \eqconjqry\cup \{(\varx, \vary)\})\}$\;%
    $\setinf\eq\setinf\cup\{\fixqry[\varx/\bot]\}$\;
  }\label{alg:free:loopend}
  \While{$\freeloopb{\setfin}{\qry}\neq\emptyset$}{\label{alg:free:hangingbeg}
	$(\qryinqfin, \eqconjqry)\leftarrow \freeloopb{\setfin}{\qry}$\;
    \label{alg:free:choosequery2}
    $\setfin\eq\setfin\setminus\{(\qryinqfin, \eqconjqry)\}$\;
    $\setinf\eq\setinf\cup\{\qryinqfin \land (\bigwedge_{\qry \in \approx(\eqconjqry)}.\;\qry)\}$\;
  }\label{alg:free:hangingend}
  \Return $(\bigor{(\qryinqfin, \eqconjqry)\in\setfin} (\sconjtwo{\qryinqfin} {\eqconjqry}),\allowbreak
  \allowed{\bigor{\qryinqinf\in\setinf} \exists \fvseq{\qryinqinf}.\,\qryinqinf})$\;
  }
    \caption{Restricting free variables.}
    \label{alg:free}
\end{algorithm}
\end{minipage}\kern-.026\textwidth
\end{figure}

Our approach for restricting all bound variables recursively applies Lemma~\ref{lem:rw_bound}.
Because the set $\cpreds$ such that $\tcov{\varx}{\qry}{\cpreds}$ holds is not
necessarily unique,
we introduce the following (general) notation.
\nondetchoice{}
We define the recursive function $\allowed{\qry}$ in Figure~\ref{alg:allowed},
where $\allowedname$ stands for $\mathsf{r}$ange-restrict $\mathsf{b}$ound variables.
The function converts an arbitrary \RC query $\qry$
into an \infequivprop{} query with range-restricted bound variables.
We proceed by describing the case $\exists\varx.\,\exqryx$.
First, $\allowed{\exqryx}$ is recursively applied
on Line~\ref{alg:allowed:rec}
to establish the precondition of Lemma~\ref{lem:rw_bound}
that the translated query has range-restricted bound variables.
Because existential quantification distributes over disjunction,
we flatten disjunction in $\allowed{\exqryx}$
and process the individual disjuncts independently.
We apply (\transex) to every disjunct $\fixqry$ in which the variable $\varx$
is not already range restricted.
For every $\fixqryc$ added to $\disjs$
after applying (\transex) to $\fixqry$
the variable $\varx$ is either range restricted or does not occur
in $\fixqryc$, i.e., $\varx\notin\gens{\fixqryc}$.
This entails the termination of the loop
on Lines~\ref{alg:allowed:loopbeg}--\ref{alg:allowed:loopend}.

%In our Isabelle formalization, we prove the correctness of $\allowed{-}$ by induction. In particular, we prove the following invariant for the while loop in the $\exists\varx.\,\exqryx$ case:
%
%\[
%\begin{array}{@{}lcl@{}}
%\mathsf{INV}_\allowedname\;\varx\;\exqryx\;\disjs&\text{iff}& \disjs\text{ is a finite set of queries}\\
%&\wedge&  \infequiv{\exists\varx.\,\exqryx}{\bigor{\qry \in \disjs}.\;\proj{\varx}{\qry}}\\
%&\wedge&  \forall\qry \in \disjs.\;\fv{\qry} \subseteq \fv{\exqryx} \land \qry\text{ has range-restricted bound variables}\\
%\end{array}\]

\begin{exampleC}[{\isabelleqed}]\label{ex:rw_bound}
Consider the query
$\qsuspusr \eq \predbrand(\varbrand) \land \exists \varscore.\;\forall \varprod.\;\predprod(\varbrand,\varprod) \longrightarrow \predscore(\varprod, \varuser, \varscore)$
from \S\ref{sec:intro}.
Restricting its bound variables yields the query
\[\allowed{\qsuspusr}=\predbrand(\varbrand)\land\biggl(\Bigl(\exists \varscore.\,\bigl(\neg \exists \varprod.\,
\predprod(\varbrand,\varprod) \land\neg
\predscore(\varprod, \varuser, \varscore)\bigr)\land
\bigl(\exists\varprod.\,\predscore(\varprod, \varuser, \varscore)\bigr)\Bigr)\lor
\Bigl(\neg \exists \varprod.\,\predprod(\varbrand,\varprod)\Bigr)\biggr).\]
The bound variable $\varprod$ is already range restricted in $\qsuspusr$
and thus only $\varscore$ must be restricted.
Applying (\transnew{}) to restrict $\varscore$
in $\neg\exists \varprod.\,\predprod(\varbrand,\varprod) \land\neg
\predscore(\varprod, \varuser, \varscore)$,
then existentially quantifying $\varscore$, 
and distributing the existential quantifier over disjunction
would yield the first disjunct in
$\allowed{\qsuspusr}$ above and
$\exists \varscore.\,(\neg \exists \varprod.\,\predprod(\varbrand,\varprod))
\land\neg (\exists \varprod.\,\predscore(\varprod, \varuser, \varscore))$
as the second disjunct.
Because there exists some value in the infinite domain $\dom$
that does not belong to the finite interpretation of the atomic predicate
$\predscore(\varprod, \varuser, \varscore)$, the query
$\exists\varscore.\,\neg (\exists \varprod.\,
\predscore(\varprod, \varuser, \varscore))$
is a tautology over $\dom$. 
Hence, $\exists \varscore.\,(\neg \exists \varprod.\,\predprod(\varbrand,\varprod))
\land\neg (\exists \varprod.\,\predscore(\varprod, \varuser, \varscore))$
is \infequivprop{} to
$\neg \exists \varprod.\,\predprod(\varbrand,\varprod)$, i.e.,
the second disjunct in $\allowed{\qsuspusr}$.
This reasoning justifies that instead of (\transnew{}) our algorithm applies (\transex{}) to restrict $\varscore$
in $\exists\varscore.\,\neg\exists \varprod.\,\predprod(\varbrand,\varprod)
\land\neg\predscore(\varprod, \varuser, \varscore)$.
\end{exampleC}

\subsection{Restricting Free Variables}\label{sec:free}

Given an arbitrary query $\qry$,
we translate the \infequivprop{} query $\allowed{\qry}$
with range-restricted bound variables
into a pair of safe-range queries $(\qryfin, \qryinf)$ such that
our translation's main properties (\propfvref{}) and (\propfinref{}) hold.
Our translation is based on the following lemma.
\begin{lemC}[{\isabelleqed}]\label{lem:rw_free}
% Let a structure $\str$ with an infinite domain $\dom$ be fixed.
Let $\varx$ be a free variable in a query $\qryrb$
with range-restricted bound variables and
let $\tcov{\varx}{\qryrb}{\cpreds}$ for a set of quantified predicates and equalities $\cpreds$.
If $\qryrb[\varx/\bot]$ is not satisfied by any tuple, then
\begin{equation}
\sattup{\qryrb}=\sattup{\left(\qryrb\land\colpredsqry{\cpreds}\right)
\lor\left(\bigor{\vary\in\coleqs{\varx}{\cpreds}} (\qryrb[\varx\mapsto \vary]\land \varx\approx \vary)\right)}.\tag{\transfree}
\end{equation}
If $\qryrb[\varx/\bot]$ is satisfied by some tuple,
then $\sattup{\smash{\qryrb}}$ is an infinite set.
\end{lemC}
\begin{proof}
If $\qryrb[\varx/\bot]$ is not satisfied by any tuple, then
(\transfree) follows from (\transnew).
If $\qryrb[\varx/\bot]$ is satisfied by some tuple,
then the last disjunct in (\transnew) applied to $\qryrb$
is satisfied by infinitely many tuples obtained
by assigning $\varx$ some value from the infinite domain $\dom$
such that $\varx\neq \vary$ holds for all finitely many
$\vary\in\coleqs{\varx}{\cpreds}$
and $\varx$ does not appear among the finitely many values interpreting
the quantified predicates from $\colpreds{\cpreds}$.
\end{proof}

We remark that $\sattup{\smash{\qryrb}}$ might be an infinite set of tuples
even if $\qryrb[\varx/\bot]$ is never satisfied, for some $\varx$.
This is because $\qryrb[\vary/\bot]$ might be satisfied by some tuple,
for some $\vary$, in which case Lemma~\ref{lem:rw_free} (for $\vary$) implies
that $\sattup{\smash{\qryrb}}$ is an infinite set of tuples.
Still, (\transfree) can be applied to $\qryrb$ for $\varx$
resulting in a query satisfied by the same infinite set of tuples.

\looseness=-1
Our approach is implemented by the function $\splitq{\qry}$
defined in Figure~\ref{alg:free}.
In the following, we describe this function
and justify its correctness,
formalized by the input/output specification.
In $\splitq{\qry}$, we represent the queries $\qryfin$ and $\qryinf$
using a set $\setfin$ of pairs consisting of a query and a relation representing a set of equalities 
and a set $\setinf$ of queries such that
\begin{align*}
\qryfin&\eq\bigor{(\qryinqfin, \eqconjqry)\in\setfin} (\sconjtwo{\qryinqfin}{\eqconjqry}),&
\qryinf&\eq\bigor{\qryinqinf\in\setinf} \exists\fvseq{\qryinqinf}.\,\qryinqinf,
\end{align*}
and, for every $(\qryinqfin, \eqconjqry)\in\setfin$, the relation $\eqconjqry$ represents a set of equalities between variables. Hereby, $\sconjtwo{\qryinqfin}{\eqconjqry}$ is a query that is equivalent to $\bigwedge_{\qry \in \{\qryinqfin\} \cup {{\approx}(\eqconjqry)}}. \;\qry$ where ${\approx}(\eqconjqry)$ abbreviates $\{\varx \approx \vary \mid (\varx,\vary) \in \eqconjqry\}$. However, the $\sconjtwo{\qryinqfin}{\eqconjqry}$ operator carefully assembles the conjunction to ensure that the resulting query is safe range (whenever possible). In particular, the operator must iteratively conjoin the equalities from ${\approx}(\eqconjqry)$ to $\qryinqfin$ in a left-associative fashion and always pick next an equation for which one of the variables is free in $\qryinqfin$ or in the equalities conjoined so far, if such an equation exists. (If no such equation exists, the operator is free to conjoin the remaining equations in an arbitrary order.)

Our algorithm proceeds as follows. As long as there exists some $(\fixqry, \eqconjqry)\in\setfin$
such that $\gens{\fixqry}\neq\emptyset$,
we apply (\transfree) to $\fixqry$
and add the query $\fixqry[\varx/\bot]$ to $\setinf$.
We remark that if we applied (\transfree)
to the entire disjunct $\sconjtwo{\fixqry} {\eqconjqry}$,
the loop on Lines~\ref{alg:free:loopbeg}--\ref{alg:free:loopend} might not terminate.
Note that, for every $(\fixqryc, \eqconjqryc)$ added to $\setfin$
after applying (\transfree) to $\fixqry$,
$\gens{\fixqryc}$ is a proper subset of $\gens{\fixqry}$.
This entails the termination of the loop
on Lines~\ref{alg:free:loopbeg}--\ref{alg:free:loopend}.
Finally, if $\sattup{\fixqry}$ is an infinite set of tuples,
then $\sattup{\sconjtwo{\fixqry}{\eqconjqry}}$ is an infinite set of tuples too.
This is because the equalities in $\eqconjqry$ merely duplicate columns
of the query $\fixqry$.
Hence, it indeed suffices to apply (\transfree) to $\fixqry$
instead of $\sconjtwo{\fixqry}{\eqconjqry}$.

After the loop on Lines~\ref{alg:free:loopbeg}--\ref{alg:free:loopend}
in~Figure~\ref{alg:free} terminates,
for every $(\qryinqfin, \eqconjqry)\in\setfin$, the query $\qryinqfin$ is safe range
and $\eqconjqry$ is a conjunction of equalities such that
$\fv{\qryinqfin}\cup\fv{\eqconjqry}\subseteq\fv{\qry}$.
However, the query $\sconjtwo{\qryinqfin}{\eqconjqry}$ does not have to be safe range, e.g.,
if $\qryinqfin\eq \predbrand(\varx)$ and
$\eqconjqry\eq\{(\varx, \vary), (\varw, \vart)\}$.
Given a relation $\eqconjqry$, let $\eclass{\eqconjqry}$ be the set of
equivalence classes of free variables $\fv{\eqset}$
with respect to the (partial) equivalence closure of $\eqconjqry$, i.e., the smallest symmetric and transitive relation that contains $\eqconjqry$. For instance,
$\eclass{\{(\varx, \vary), (\vary, \varz), (\varw, \vart)\}}=
\{\{\varx, \vary, \varz\}, \{\varw, \vart\}\}$.
Let
$\hanging{\qryinqfin}{\eqconjqry}\eq\bigcup_{\varseta\in\eclass{\eqconjqry}, \varseta\cap\fv{\qryinqfin}=\emptyset} \varseta$
be the set of all variables in equivalence classes from $\eclass{\eqconjqry}$
that are disjoint from $\qryinqfin$'s free variables.
Then, $\sconjtwo{\qryinqfin}{\eqconjqry}$ is safe range if and only if
$\hanging{\qryinqfin}{\eqconjqry}=\emptyset$.\footnote{This statement contained the error we
discovered while formalizing the result presented in our conference paper~\cite{DBLP:conf/icdt/RaszykBKT22}.
There we had wrongly used the naive conjunction $\qryinqfin \land (\bigwedge_{\qry \in
\approx(\eqconjqry)}.\;\qry)$, which will not be safe range whenever $\eqconjqry$ has more than one
element, instead of the more carefully constructed $\sconjtwo{\qryinqfin}{\eqconjqry}$.}

\looseness=-1
Now if $\hanging{\qryinqfin}{\eqconjqry}\neq\emptyset$ and the query $\sconjtwo{\qryinqfin}{\eqconjqry}$ is satisfied by some tuple,
then $\sattup{\sconjtwo{\qryinqfin}{\eqconjqry}}$ is an infinite set of tuples
because all equivalence classes of variables in $\hanging{\qryinqfin}{\eqconjqry}\neq\emptyset$
can be assigned arbitrary values from the infinite domain $\dom$.
In our example with $\qryinqfin\eq \predbrand(\varx)$ and $\eqconjqry\eq\{(\varx, \vary), (\varw, \vart)\}$,
we have $\hanging{\qryinqfin}{\eqconjqry}=\{\varw, \vart\}\neq\emptyset$.
Moreover, if $\fv{\qryinqfin}\cup\fv{\eqconjqry}\neq\fv{\qry}$ and $\sconjtwo{\qryinqfin}{\eqconjqry}$ is satisfied by
some tuple, then this tuple can be extended to infinitely many tuples
over $\fv{\qry}$ by choosing arbitrary values from the infinite domain $\dom$
for the variables in the non-empty set $\fv{\qry}\setminus(\fv{\qryinqfin}\cup\fv{\eqconjqry})$.
Hence, for every $(\qryinqfin, \eqconjqry)\in\setfin$ with $\hanging{\qryinqfin}{\eqconjqry}\neq\emptyset$
or $\fv{\qryinqfin}\cup\fv{\eqconjqry}\neq\fv{\qry}$,
we remove $(\qryinqfin, \eqconjqry)$ from $\setfin$ and
add $\sconjtwo{\qryinqfin}{\eqconjqry}$ to $\setinf$.
Note that we only remove pairs from $\setfin$,
hence the loop on Lines~\ref{alg:free:hangingbeg}--\ref{alg:free:hangingend}
terminates.
Afterwards, the query $\qryfin$ is safe range.
However, the query $\qryinf$ does not have to be safe range.
Indeed, every query $\qryinqinf\in\setinf$ has range-restricted bound variables,
but not all the free variables of $\qryinqinf$ need be range restricted
and thus the query $\exists\fvseq{\qryinqinf}.\,\qryinqinf$ does not have to be safe range.
But the query $\qryinf$ is closed and thus the \infequivprop{} query
$\allowed{\qryinf}$ with range-restricted bound variables
is safe range.

\begin{lemC}[{\isabelleqed}]\label{lem:rw_wf}
Let $\qry$ be an \RC query
and $\splitq{\qry}=(\qryfin, \qryinf)$.
Then the queries $\qryfin$ and $\qryinf$ are safe range;
$\fv{\qryfin}=\fv{\qry}$ unless $\qryfin$ is syntactically equal to $\bot$;
and $\fv{\qryinf}=\emptyset$.
\end{lemC}

\begin{lemC}[{\isabelleqed}]\label{lem:rw_sound}
% Let a structure $\str$ with an infinite domain $\dom$ be fixed.
Let $\qry$ be an \RC query
and $\splitq{\qry}=(\qryfin, \qryinf)$.
If $\models \qryinf$,
then $\sattup{\qry}$ is an infinite set.
Otherwise,
$\sattup{\qry}=\sattup{\qryfin}$ is a finite set.
\end{lemC}

By Lemma~\ref{lem:rw_wf}, $\qryfin$ is a safe-range
(and thus also domain-independent) query.
Hence, for the fixed structure, %$\str$,
the tuples in $\sattup{\qryfin}$ only contain elements in the active domain
$\adom{\qryfin}$, i.e.,
$\sattup{\qryfin}=\sattup{\qryfin}\cap\adom{\qryfin}^{\card{\fv{\qryfin}}}$.
Our translation does not introduce
new constants in $\qryfin$ and thus $\adom{\qryfin}\subseteq\adom{\qry}$.
Hence, by Lemma~\ref{lem:rw_sound}, if $\not\models\qryinf$,
then $\sattup{\qryfin}$ is equal to the
``output-restricted unlimited interpretation''~\cite{DBLP:journals/acta/HullS94} of $\qry$, i.e.,
$\sattup{\qryfin}=\sattup{\qry}\cap\adom{\qry}^{\card{\fv{\qry}}}$.
In contrast, if $\models\qryinf$, then
$\sattup{\qryfin}=\sattup{\qry}\cap\adom{\qry}^{\card{\fv{\qry}}}$
does not necessarily hold.
For instance, for $\qry\eq\neg\predbrand(\varx)$, our translation yields
$\splitq{\qry}=(\bot, \top)$.
In this case, we have $\qryinf=\top$ and thus $\models\qryinf$
because $\neg\predbrand(\varx)$ is satisfied by infinitely many
tuples over an infinite domain.
However, if $\predbrand(\varx)$ is never satisfied, then
$\sattup{\qryfin}=\emptyset$ is not equal to
$\sattup{\qry}\cap\adom{\qry}^{\card{\fv{\qry}}}$.

\looseness=-1
Next, we demonstrate different aspects of our translation on a few examples. Thereby, we use a
mildly modified algorithm that performs constant propagation after all steps that could introduce
constants $\top$ or $\bot$ in a subquery. This optimization keeps the queries small, but is not
necessary for termination and correctness. (In contrast, the constant propagation that is part of
the substitution operators $\qry[\varx\mapsto \vary]$ and $\qry[\varx/\bot]$ is necessary.) We have
verified in Isabelle that our results hold for the modified algorithm. That is, for all above
theorems, we proved two variants: one with and one without additional constant propagation steps.

\begin{exampleC}[{\isabelleqed}]
Consider the query $\qry\eq\predbrand(\varx)\lor\predprod(\varx, \vary)$.
The variable $\vary$ is not range restricted in $\qry$ and thus
$\splitq{\qry}$ restricts $\vary$ by a conjunction of $\qry$
with $\predprod(\varx, \vary)$.
However, if $\qry[\vary/\bot]=\predbrand(\varx)$ is satisfied
by some tuple, then $\sattup{\qry}$ contains infinitely many
tuples.
Hence, $\splitq{\qry}=((\predbrand(\varx)\lor\predprod(\varx, \vary))\land\predprod(\varx, \vary), \exists\varx.\,\predbrand(\varx))$.
Because $\qryfin=(\predbrand(\varx)\lor\predprod(\varx, \vary))\land\predprod(\varx, \vary)$ is only used if $\not\models\qryinf$, i.e.,
if $\predbrand(\varx)$ is never satisfied, we could simplify
$\qryfin$ to $\predprod(\varx, \vary)$. However,
our translation does not implement such heuristic simplifications.
\end{exampleC}

\begin{exampleC}[{\isabelleqed}]
Consider the query $\qry\eq \predbrand(\varx)\land \varw\approx \vart$.
The variables $\varw$ and $\vart$ are not range restricted in $\qry$
and thus $\splitq{\qry}$ chooses one of these variables
(e.g., $\varw$) and restricts it by splitting
$\qry$ into $\qryinqfin=\predbrand(\varx)$ and $\eqconjqry=\{(\varw,\vart)\}$.
Now, all variables are range restricted in $\qryinqfin$,
but the variables in $\qryinqfin$ and $\eqconjqry$ are disjoint.
Hence, $\sattup{\qry}$ contains infinitely many tuples
whenever $\qryinqfin$ is satisfied by some tuple.
In contrast, $\sattup{\qry}=\emptyset$
if $\qryinqfin$ is never satisfied.
Hence, we have $\splitq{\qry}=(\bot,\exists\varx.\,\predbrand(x))$.
\end{exampleC}

\begin{exampleC}[{\isabelleqed}]\label{ex:rw_free}
Consider the query
$\qsuspusr \eq \predbrand(\varbrand) \land \exists \varscore.\;\forall \varprod.\;\predprod(\varbrand,\varprod) \longrightarrow \predscore(\varprod, \varuser, \varscore)$
from \S\ref{sec:intro}.
Restricting its bound variables yields the query
$\allowed{\qsuspusr}=\predbrand(\varbrand)\land((\exists \varscore.\,(\neg \exists \varprod.\,
\predprod(\varbrand,\varprod) \land\neg
\predscore(\varprod, \varuser, \varscore))\land
(\exists\varprod.\,\predscore(\varprod, \varuser, \varscore)))\lor
(\neg \exists \varprod.\,\predprod(\varbrand,\varprod)))$
derived in Example~\ref{ex:rw_bound}.
Splitting $\qsuspusr$ yields
\[\splitq{\qsuspusr}=\Bigl(\allowed{\qsuspusr}\land
\bigl(\exists\varscore,\varprod.\,\predscore(\varprod, \varuser, \varscore)\bigr),
\exists\varbrand.\,\predbrand(\varbrand)\land\neg \exists \varprod.\,\predprod(\varbrand,\varprod)\Bigr).\]
To understand $\splitq{\qsuspusr}$,
we apply (\transnew) to $\allowed{\qsuspusr}$
for the free variable $\varuser$:
\[
\allowed{\qsuspusr}\equiv \Bigl(\allowed{\qsuspusr}\land
\bigl(\exists\varscore,\varprod.\,\predscore(\varprod, \varuser, \varscore)\bigr)\Bigr)\lor
\Bigl(\predbrand(\varbrand)\land\bigl(\neg \exists \varprod.\,\predprod(\varbrand,\varprod)\bigr)\land\neg
\exists\varscore,\varprod.\,\predscore(\varprod, \varuser, \varscore)\Bigr).
\]
\looseness=-1
If the subquery $\predbrand(\varbrand)\land(\neg \exists \varprod.\,\predprod(\varbrand,\varprod))$
from the second disjunct is satisfied for some $\varbrand$,
then $\qsuspusr$ is satisfied by infinitely many values for $\varuser$
from the infinite domain $\dom$
that do not belong to the finite interpretation of
$\predscore(\varprod, \varuser, \varscore)$ and thus satisfy the subquery
$\neg\exists\varscore,\varprod.\,\predscore(\varprod, \varuser, \varscore)$.
Hence, $\sattupd{\qsuspusr}{\str}=\sattupd{\allowed{\qsuspusr}}{\str}$
is an infinite set of tuples whenever
$\predbrand(\varbrand)\land\neg \exists \varprod.\,\predprod(\varbrand,\varprod)$ is satisfied for some $\varbrand$.
In contrast, if $\predbrand(\varbrand)\land\neg \exists \varprod.\,\predprod(\varbrand,\varprod)$
is not satisfied for any $\varbrand$,
then $\qsuspusr$ is equivalent to $\allowed{\qsuspusr}\land
(\exists\varscore,\varprod.\,\predscore(\varprod, \varuser, \varscore))$
obtained also by applying (\transfree) to $\qsuspusr$
for the free variable $\varuser$.
\end{exampleC}

\section{Complexity Analysis}\label{sec:complex}

\looseness=-1
We analyze the time complexity of capturing $\qry$, i.e.,
checking if $\sattup{\qry}$ is finite and enumerating $\sattup{\qry}$ in this case.
To bound the asymptotic time complexity of capturing a fixed~$\qry$,
we need to apply an additional standard translation step to both queries produced 
by our translation to obtain two RANF queries. Query cost (\S\ref{sec:cost}) can 
then be applied to the resulting two queries to bound computation time based on 
the cardinalities of subquery evaluation results.
\begin{definition}\label{def:rw}
  Let $\qry$ be an \RC{} query and $\splitq{\qry}=(\qryfin, \qryinf)$.
  Let $\safefin\eq\allowtosafeqry{\qryfin}$ and
  $\safeinf\eq\allowtosafeqry{\qryinf}$ be the equivalent \safeprop{} queries.
  We define $\rw{\qry}\eq(\safefin, \safeinf)$.
\end{definition}
Since function $\allowtosafeqry{\cdot}$ is a standard translation step, we present it 
in \S\ref{sec:detail} (see Figure~\ref{alg:allowtosafe}). Note that the proof of 
Lemma~\ref{lem:closeatoms_sub} relies on its algorithmic details.

We ignore the (constant) time complexity of computing
$\rw{\qry}=(\safefin, \safeinf)$ and focus on
the time complexity of evaluating the \safeprop{} queries $\safefin$ and
$\safeinf$, i.e., the query cost of $\safefin$ and $\safeinf$.
Without loss of generality, we assume that the input query
$\qry$ has \pwdist{}
to derive a set of quantified predicates from $\qry$'s atomic predicates
and formulate our time complexity bound.
Still, the \safeprop{} queries $\safefin$ and
$\safeinf$ computed by our translation need not have \pwdist{}.

We define the relation $\varle{\qry}$
on $\av{\qry}$ such that $\varx\varle{\qry}\vary$ iff
the scope of an occurrence of $\varx \in \av{\qry}$ is contained
in the scope of an occurrence of $\vary \in \av{\qry}$.
Formally, we define $\varx\varle{\qry}\vary$ iff $\vary\in\fv{\qry}$
or $\exists \varx.\,\exqryx\sqsubseteq\exists \vary.\,\exqryy\sqsubseteq \qry$
for some $\exqryx$ and $\exqryy$.
Note that $\varle{\qry}$ is a preorder %(also called quasiorder)
on all variables and a partial order on the bound variables
for every query with \pwdist{}.

Let $\colatoms{\qry}$ be the set of all atomic predicates in a query $\qry$.
We denote by $\closeatomseq{\qry}$ the set of quantified predicates
obtained from $\colatoms{\qry}$
by performing the variable substitution $\varx\mapsto \vary$,
where $\varx$ and $\vary$ are related by equalities in $\qry$
and $\varx\varle{\qry}\vary$,
and existentially quantifying from a quantified predicate $\projoneatom$
the innermost bound variable $\varx$ in $\qry$
that is free in $\projoneatom$.
Let $\transeqs{\qry}$ be the transitive closure
of equalities occurring in $\qry$.
Formally, we define $\closeatomseq{\qry}$ by:
\begin{itemize}
\item $\oneatom\in\closeatomseq{\qry}$ if $\oneatom\in\colatoms{\qry}$;
\item $\projoneatom[\varx\mapsto \vary]\in\closeatomseq{\qry}$ if
$\projoneatom\in\closeatomseq{\qry}$,
$(\varx, \vary)\in\transeqs{\qry}$,
and $\varx\varle{\qry}\vary$;
\item $\exists\varx.\,\projoneatom\in\closeatomseq{\qry}$
if $\projoneatom\in\closeatomseq{\qry}$, $\varx\in\fv{\projoneatom}\setminus\fv{\qry}$,
and $\varx\varle{\qry}\vary$ for all $\vary\in\fv{\projoneatom}$.
\end{itemize}
%When restricting a variable by a disjunction of quantified predicates
%$\colpredsqry{\cpreds}$ for some $\cpreds$,
%we only introduce quantified predicates $\projoneatom\in\closeatomseq{\qry}$.

\looseness=-1
We bound the time complexity of capturing $\qry$
by considering subsets $\projatomset$ of quantified predicates
$\closeatomseq{\qry}$ that are \emph{minimal} in the sense
that every quantified predicate in $\projatomset$ contains a unique free variable
that is not free in any other quantified predicate in $\projatomset$.
Formally, we define $\projminimal\eq\projminimaldef$.
Every minimal subset $\projatomset$
of quantified predicates $\closeatomseq{\qry}$
contributes the product of the numbers of tuples satisfying
each quantified predicate $\projoneatom\in\projatomset$
to the overall bound (that product is an upper bound on the
number of tuples satisfying the join over all $\projoneatom\in\projatomset$).
Similarly to~\citeauthorMWJ~\cite{DBLP:journals/sigmod/NgoRR13},
we use the notation $\tildeo{\cdot}$ to hide
logarithmic factors incurred by set operations.

\begin{theorem}\label{lem:data_complexity}
Let $\qry$ be a fixed \RC query with \pwdist{}.
The time complexity of capturing $\qry$, i.e.,
checking if $\sattup{\qry}$ is finite
and enumerating $\sattup{\qry}$ in this case, is in
$\textstyle\tildeo{\sum_{\projatomset\subseteq\closeatomseq{\qry},\projminimal}
\prod_{\projoneatom\in\projatomset} \card{\sattup{\projoneatom}}}.$
\end{theorem}
Before we prove Theorem~\ref{lem:data_complexity} we first provide some 
examples to reinforce the intuition behind our claim.
Examples~\ref{ex:vgt} and~\ref{ex:fin-dom} show that the time complexity from
Theorem~\ref{lem:data_complexity} cannot be achieved
by the translation of Van Gelder and Topor~\cite{DBLP:journals/tods/GelderT91}
or over finite domains.
Example~\ref{ex:cart-prod} shows how equalities affect the bound in Theorem~\ref{lem:data_complexity}.

\begin{example}\label{ex:vgt}
Consider the query
$\qry\eq \predbrand(\varbrand) \land \exists \varuser, \varscore.\,\neg\exists \varprod.\,
\predprod(\varbrand, \varprod) \land\neg \predscore(\varprod,\varuser,\varscore)$, equivalent to $\qsusp$ from \S\ref{sec:intro}.
Then $\colatoms{\qry}=\{\predbrand(\varbrand), \predprod(\varbrand, \varprod), \predscore(\varprod,\varuser,\varscore)\}$ and
$\closeatomseq{\qry}=\{\predbrand(\varbrand), \allowbreak\predprod(\varbrand, \varprod),\exists \varprod.\,\predprod(\varbrand, \varprod),
\predscore(\varprod,\varuser,\varscore), \exists \varprod.\,\predscore(\varprod,\varuser,\varscore), \exists \varscore, \varprod.\,\predscore(\varprod,\varuser,\varscore),\allowbreak
\exists \varuser, \varscore, \varprod.\,\predscore(\varprod,\varuser,\varscore)\}$.
The translated query $\qryvgt$ by Van Gelder and Topor~\cite{DBLP:journals/tods/GelderT91}
\[
\begin{array}{@{}l@{}}
\Bigl(\highlight{\bigl(\exists \varscore, \varprod.\,\predscore(\varprod,\varuser,\varscore)\bigr) \land \bigl(\exists \varuser, \varprod.\,\predscore(\varprod,\varuser,\varscore)\bigr)} \land \predbrand(\varbrand)\Bigr) \land{} \\\neg\exists \varprod.\,
\Bigl(\highlight{\bigl(\exists \varscore, \varprod.\,\predscore(\varprod,\varuser,\varscore)\bigr) \land \bigl(\exists \varuser, \varprod.\,\predscore(\varprod,\varuser,\varscore)\bigr)} \land \predprod(\varbrand, \varprod)\Bigr) \land \neg \predscore(\varprod,\varuser,\varscore)
\end{array}
\]
restricts the variables $\varuser$ and $\varscore$ by $\exists \varscore, \varprod.\,\predscore(\varprod,\varuser,\varscore)$ and
$\exists \varuser, \varprod.\,\predscore(\varprod,\varuser,\varscore)$, respectively.
Note that this corresponds to the RA expression shown in \S\ref{sec:intro} with the 
highlighted generators replaced with $\pi_{\username}(\predscore) \times \pi_{\scorename}(\predscore)$.

Consider an interpretation of
$\predbrand$ by $\{(\cstsymba)\mid\cstsymba\in\{1, \ldots, \prodparam\}\}$,
$\predprod$ by $\{(\cstsymba, \cstsymba)\mid\cstsymba\in\{1, \ldots, \prodparam\}\}$,
and $\predscore$ by
$\{(\cstsymb, \cstsymba, \cstsymba)\mid\cstsymb\in\{1, \ldots, \prodparam\}, \cstsymba\in\{1, \ldots, \usrparam\}\}$,
$\prodparam, \usrparam\in\nats$. Computing the join of
$\predprod(\varbrand, \varprod)$,
$\exists \varscore, \varprod.\,\predscore(\varprod,\varuser,\varscore)$, and
 $\exists \varuser, \varprod.\,\predscore(\varprod,\varuser,\varscore)$, which is a Cartesian product, 
results in a time complexity in $\Omega(\prodparam\cdot\usrparam^2)$ for $\qryvgt$.
In contrast, Theorem~\ref{lem:data_complexity} yields
an asymptotically better time complexity in $\tildeo{\prodparam+\usrparam+\prodparam\cdot\usrparam}$
for our translation, more precisely:
\[\tildeo{\card{\sattup{\predbrand(\varbrand)}}+\card{\sattup{\predprod(\varbrand,\varprod)}}+\allowbreak\card{\sattup{\predscore(\varprod,\varuser,\varscore)}}+
(\card{\sattup{\predbrand(\varbrand)}}+\card{\sattup{\predprod(\varbrand,\varprod)}})\cdot\card{\sattup{\predscore(\varprod,\varuser,\varscore)}}},\]
which corresponds to the complexity of evaluating the RA expression shown in \S\ref{sec:intro}.
\end{example}

\begin{example}\label{ex:fin-dom}
The query $\neg \predscore(\varx, \vary, \varz)$ is satisfied by a finite set of tuples
over a finite domain~$\dom$ (as is every query over a finite domain).
For an interpretation of $\predscore$ by
$\{(\cstsymb, \cstsymb, \cstsymb)\mid \cstsymb\in\dom\}$,
the equality $\card{\dom}=\card{\sattup{\predscore(\varx, \vary, \varz)}}$
holds and the number of satisfying tuples is
\[
\card{\sattup{\neg \predscore(\varx, \vary, \varz)}}=\card{\dom}^3-\card{\sattup{\predscore(\varx, \vary, \varz)}}=
\card{\sattup{\predscore(\varx, \vary, \varz)}}^3-\card{\sattup{\predscore(\varx, \vary, \varz)}}\in\Omega(\card{\sattup{\predscore(\varx, \vary, \varz)}}^3),
\]
which exceeds the bound $\tildeo{\card{\sattup{\predscore(\varx,\vary,\varz)}}}$ of Theorem~\ref{lem:data_complexity}. Hence, our infinite domain assumption is crucial for achieving the better complexity bound.
\end{example}

\begin{example}\label{ex:cart-prod}
Consider the following query over the infinite domain $\dom=\nats$ of natural numbers:
\[
\begin{array}{@{}l@{\,}l@{}}
\qry\eq\forall \varw.&(\varw\approx0\lor \varw\approx1\lor \varw\approx2)\longrightarrow\\
&(\exists \vart.\,\predbrand(\vart)\land(\varw\approx0\longrightarrow \varx\approx\vart)
\land(\varw\approx1\longrightarrow \vary\approx\vart)
\land(\varw\approx2\longrightarrow \varz\approx\vart)).
\end{array}
\]
Note that this query is equivalent to $\qry\equiv \predbrand(\varx)\land \predbrand(\vary)\land \predbrand(\varz)$
and thus it is satisfied by a finite set of tuples of size
$\card{\sattup{\predbrand(\varx)}}\cdot\card{\sattup{\predbrand(\vary)}}\cdot\card{\sattup{\predbrand(\varz)}}=\card{\sattup{\predbrand(\varx)}}^3$.
The set of atomic predicates of~$\qry$ is $\colatoms{\qry}=\{\predbrand(\vart)\}$
and it must be closed
under the equalities occurring in $\qry$ to yield a valid bound
in Theorem~\ref{lem:data_complexity}.
In this case, $\closeatomseq{\qry}=\{\predbrand(\vart), \exists\vart.\,\predbrand(\vart), \predbrand(\varx), \predbrand(\vary), \predbrand(\varz)\}$
and the bound in Theorem~\ref{lem:data_complexity} is
$\card{\sattup{\predbrand(\vart)}}\cdot\card{\sattup{\predbrand(\varx)}}\cdot\card{\sattup{\predbrand(\vary)}}\cdot\card{\sattup{\predbrand(\varz)}}=
\card{\sattup{\predbrand(\varx)}}^4$. In particular, this bound is not tight, but it still
reflects the complexity of evaluating the \safeprop{} queries produced
by our translation as it does not derive
the equivalence $\qry\equiv \predbrand(\varx)\land \predbrand(\vary)\land \predbrand(\varz)$.
\end{example}

\looseness=-1
Now, to prove Theorem~\ref{lem:data_complexity}, we need to 
introduce guard queries and the set of quantified predicates of a query.
Given a \safeprop{} query $\safeqry$,
we define a \emph{\guard{}} query $\guardqry{\safeqry}$
that is implied by $\safeqry$, i.e.,
$\guardqry{\safeqry}$ can be used to over-approximate
the set of satisfying tuples for $\safeqry$.
We use this over-approximation in our proof
of Theorem~\ref{lem:data_complexity}.
The \guard{} query $\guardqry{\safeqry}$ has a simple structure:
it is the disjunction of conjunctions of quantified predicates
and equalities.

We now define the set of quantified predicates $\closeatoms{\qry}$
occurring in the \guard{} query $\guardqry{\qry}$.
For an atomic predicate $\oneatom\in\colatoms{\qry}$,
let $\bset{\oneatom}{\qry}$ be the set of sequences of bound variables
for all occurrences of~$\oneatom$ in~$\qry$. For example, consider a query $\qex\eq
((\exists \varz.\,(\exists \vary,\varz.\,\predexa(\varx, \vary, \varz))\land \predexb(\vary, \varz))\land \predexc(\varz))\lor \predexa(\varx, \vary, \varz)$.
Then $\colatoms{\qex}=\{\predexc(\varz), \predexb(\vary, \varz),\allowbreak
\predexa(\varx, \vary, \varz)\}$ and
$\bset{\predexa(\varx, \vary, \varz)}{\qex}=\{[\vary, \varz], \emptyseq\}$,
where $\emptyseq$ denotes the empty sequence
corresponding to the occurrence of $\predexa(\varx, \vary, \varz)$ in $\qex$
for which the variables $\varx, \vary, \varz$ are all free in $\qex$.
Note that the variable $\varz$ in the other occurrence of $\predexa(\varx, \vary, \varz)$
in $\qex$ is bound to the innermost quantifier. Hence,
neither $[\varz, \vary]$ nor $[\varz, \vary, \varz]$
are in $\bset{\predexa(\varx, \vary, \varz)}{\qex}$.
Furthermore, let $\closeatoms{\qry}$ be the set of the quantified predicates
obtained by existentially quantifying sequences of bound variables in
$\bset{\oneatom}{\subqry}$ from the atomic predicates $\oneatom\in \colatoms{\subqry}$
in all subqueries $\subqry$ of $\qry$.
Formally,
$\closeatoms{\qry}\eq\bigcup_{\substack{\subqry\sqsubseteq \qry,
\oneatom\in\colatoms{\subqry}}}
\{\exists\varlist.\,\oneatom\mid\varlist\in\bset{\oneatom}{\subqry}\}$.
For instance,
$\closeatoms{\qex}=\{\predexa(\varx, \vary, \varz), \exists \varz.\,\predexa(\varx, \vary, \varz), \exists \vary, \varz.\,\predexa(\varx, \vary, \varz),
\predexb(\vary, \varz), \exists \varz.\,\predexb(\vary, \varz), \predexc(\varz)\}$.

\looseness=-1
A crucial property of our translation, which is central
for the proof of Theorem~\ref{lem:data_complexity},
is the relationship between the quantified predicates
$\closeatoms{\safeqry}$ for a \safeprop{} query
$\safeqry$ produced by our translation and the original query $\qry$.
The relationship is formalized in the following lemma.

\begin{lemma}\label{lem:closeatoms_sub}
Let $\qry$ be an \RC query with \pwdist{}
and let $\rw{\qry}=(\safefin, \safeinf)$.
Let $\safeqry\in\{\safefin, \safeinf\}$.
Then $\closeatoms{\safeqry}\subseteq\closeatomseq{\qry}$.
\end{lemma}

\begin{proof}
Let $\splitq{\qry}=(\qryfin, \qryinf)$.
We observe that
$\colatoms{\qryfin}\subseteq\closeatomseq{\qry}$,
$\transeqs{\qryfin}\subseteq\transeqs{\qry}$,
$\varle{\qryfin}\subseteq\varle{\qry}$,
$\colatoms{\qryinf}\subseteq\closeatomseq{\qry}$,
$\transeqs{\qryinf}\subseteq\transeqs{\qry}$,
and $\varle{\qryinf}\subseteq\varle{\qry}$.
Hence, $\closeatomseq{\qryfin}\subseteq\closeatomseq{\qry}$
and $\closeatomseq{\qryinf}\subseteq\closeatomseq{\qry}$.

Next we observe that $\closeatoms{\subqry}\subseteq\closeatomseq{\subqry}$
for every query $\subqry$.
Finally, we show that $\closeatoms{\safefin}\subseteq\closeatoms{\qryfin}$
and $\closeatoms{\safeinf}\subseteq\closeatoms{\qryinf}$.
We observe that
$\bset{\oneatom}{\cp{\subqry}}\subseteq\bset{\oneatom}{\subqry}$,
$\bset{\oneatom}{\pushnot{\subqry}}\subseteq\bset{\oneatom}{\subqry}$,
and then
$\closeatoms{\cp{\subqry}}\subseteq\closeatoms{\subqry}$,
$\closeatoms{\pushnot{\subqry}}\subseteq\closeatoms{\subqry}$,
for every query $\subqry$.

\looseness=-1
Assume that $\subqry\land\conjrestr$ is a safe-range query
in which no variable occurs both free and bound,
no bound variables shadow each other, i.e.,
there are no subqueries $\exists\varx.\,\exqryx\sqsubseteq\exqryxc$
and $\exists\varx.\,\exqryxc\sqsubseteq\subqry\land\conjrestr$,
and every two subqueries $\exists\varx.\,\exqryx\sqsubseteq\qrya$ and
$\exists\varx.\,\exqryxc\sqsubseteq\qryb$ such that
$\qrya\land\qryb\sqsubseteq\subqry\land\conjrestr$
have the property that $\exists\varx.\,\exqryx$
or $\exists\varx.\,\exqryxc$ is a quantified predicate.
Then the free variables in $\conjrestr$ never clash with the bound variables
in $\subqry$, i.e., Line~\ref{alg:allowtosafe:capture}
in Figure~\ref{alg:allowtosafe} is never executed.
Next we observe that
$\bset{\oneatom}{\allowtosafe{\subqry}{\restr}}\subseteq
\bset{\oneatom}{\subqry\land\conjrestr}$ (this subset relation only holds when considering queries modulo $\alpha$-equivalence, i.e., queries that have the same binding structure but differ in the used bound variable names are considered to be equal) and then
$\closeatoms{\allowtosafe{\subqry}{\restr}}\subseteq
\closeatoms{\subqry\land\conjrestr}$.
Because $\qryfin$, $\qryinf$ have the assumed properties and $\closeatoms{\pushnot{\subqry}}\subseteq\closeatoms{\subqry}$,
for every query $\subqry$,
we get $\closeatoms{\safefin}=\closeatoms{\allowtosafeqry{\qryfin}}\subseteq
\closeatoms{\qryfin}$ and $\closeatoms{\safeinf}=
\closeatoms{\allowtosafeqry{\qryinf}}\subseteq\closeatoms{\qryinf}$.
\end{proof}

Recall Example~\ref{ex:vgt}.
The query $\exists \varuser, \varprod.\,\predscore(\varprod,\varuser,\varscore)$
is in $\closeatoms{\qryvgt}$, but not in $\closeatomseq{\qry}$.
Hence, $\closeatoms{\qryvgt}\subseteq\closeatomseq{\qry}$,
i.e., an analogue of Lemma~\ref{lem:closeatoms_sub} for \citeauthorVanGT's
translation, does not hold.

% Let a structure $\str$ be fixed.
Every tuple satisfying a \safeprop{} query $\safeqry$
belongs to the set of tuples satisfying
the join over some minimal subset $\projatomset\subseteq\closeatoms{\safeqry}$
of quantified predicates and satisfying
equalities duplicating some of $\projatomset$'s columns.
Hence, we define the \guard{} query $\guardqry{\safeqry}$ as follows:
\[
\guardqry{\safeqry}\eq
\bigvee_{\substack{\projatomset\subseteq\closeatoms{\safeqry}, \projminimal,\\
\eqset\subseteq\copyvars{\fv{\projatomset}}{\fv{\safeqry}},\\
\fv{\projatomset}\cup\fv{\eqset}=\fv{\safeqry}}}
\left(\bigwedge_{\projoneatom\in\projatomset}\projoneatom
\land\bigwedge_{\eqqry\in\eqset} \eqqry
\right).
\]
Note that $\copyvars{\varseta}{\varsetb}$
denotes the set of all equalities $\varx\approx \vary$ between variables
$\varx\in\varseta$ and $\vary\in\varsetb$.
We express the correctness of the \guard{} query in the following lemma.
\begin{lemma}\label{lem:guard_qry}
Let $\safeqry$ be a \safeprop{} query. Then,
for all variable assignments $\valsymb$,
\[\valsymb\models\safeqry\Longrightarrow\valsymb\models\guardqry{\safeqry}.\]
Moreover, $\fv{\guardqry{\safeqry}}=\fv{\safeqry}$
unless $\guardqry{\safeqry}=\bot$.
Hence, $\sattup{\smash{\safeqry}}$ satisfies
\[
\sattup{\safeqry}
\subseteq
\bigcup_{\substack{\projatomset\subseteq\closeatoms{\safeqry}, \projminimal,\\
\eqset\subseteq\copyvars{\fv{\projatomset}}{\fv{\safeqry}},\\
\fv{\projatomset}\cup\fv{\eqset}=\fv{\safeqry}}}
\sattup{\bigwedge_{\projoneatom\in\projatomset}\projoneatom
\land\bigwedge_{\eqqry\in\eqset} \eqqry
}.
\]
\end{lemma}

\begin{proof}
The statement follows by well-founded induction
over the definition of $\safe{\safeqry}$.
\end{proof}

We now derive a bound on $\card{\sattup{\smash{\safesubqry}}}$,
for an arbitrary \safeprop{} subquery
$\safesubqry\sqsubseteq \safeqry$, $\safeqry\in\{\safefin, \safeinf\}$.

\begin{lemma}\label{lem:complex_subquery_subset}
Let $\qry$ be an \RC query with \pwdist{}
and let $\rw{\qry}=(\safefin, \safeinf)$.
Let $\safesubqry\sqsubseteq \safeqry$ be a \safeprop{} subquery
of $\safeqry\in\{\safefin, \safeinf\}$.
Then
\[\textstyle
\card{\sattup{\safesubqry}}\leq
\sum_{\substack{\projatomset\subseteq\closeatomseq{\qry}, \projminimal}}
2^{\card{\av{\safeqry}}}\cdot
\prod_{\projoneatom\in\projatomset}\card{\sattup{\projoneatom}}.
\]
\end{lemma}

\begin{proof}
Applying Lemma~\ref{lem:guard_qry}
to the \safeprop{} query $\safesubqry$ yields
\[
\sattup{\safesubqry}
\subseteq
\bigcup_{\substack{\projatomset\subseteq\closeatoms{\safesubqry}, \projminimal,\\
\eqset\subseteq\copyvars{\fv{\projatomset}}{\fv{\safesubqry}},\\
\fv{\projatomset}\cup\fv{\eqset}=\fv{\safesubqry}}}
\sattup{\bigwedge_{\projoneatom\in\projatomset}\projoneatom
\land\bigwedge_{\eqqry\in\eqset} \eqqry
}.
\]
We observe that
$\card{\sattup{\smash{\bigwedge_{\projoneatom\in\projatomset} \projoneatom
\land\bigwedge_{\eqqry\in\eqset} \eqqry}}}\leq
\card{\sattup{\smash{\bigwedge_{\projoneatom\in\projatomset} \projoneatom}}}\leq
\prod_{\projoneatom\in\projatomset} \card{\sattup{\projoneatom}}$
where the first inequality follows from the fact that
equalities $\eqqry\in\eqset$ can only restrict a set of tuples
and duplicate columns.
Because $\safesubqry$ is a subquery of $\safeqry$,
it follows that $\closeatoms{\safesubqry}\subseteq\closeatoms{\safeqry}$. Lemma~\ref{lem:closeatoms_sub} yields
$\closeatoms{\safeqry}\subseteq \closeatomseq{\qry}$.
Hence, we derive $\closeatoms{\safesubqry}\subseteq\closeatomseq{\qry}$.

The number of equalities
in $\copyvars{\fv{\projatomset}}{\fv{\safesubqry}}$
is at most
\[
\card{\fv{\projatomset}}\cdot\card{\fv{\safesubqry}}
\leq \card{\fv{\safesubqry}}^2
\leq \card{\av{\safeqry}}^2.
\]
The first inequality holds because
$\fv{\projatomset}\cup\fv{\eqset}=\fv{\safesubqry}$ and thus
$\fv{\projatomset}\subseteq\fv{\safesubqry}$.
The second inequality holds because
the variables in a subquery $\safesubqry$ of $\safeqry$
are in $\av{\safeqry}$.
Hence, the number of subsets
$\eqset\subseteq\copyvars{\fv{\projatomset}}{\fv{\safesubqry}}$
is at most $2^{\card{\av{\safeqry}}^2}$.
\end{proof}

\looseness=-1
\!We now bound the query cost of a \safeprop{} query
$\safeqry\!\in\!\{\safefin, \safeinf\}$
over the fixed structure $\str$.

\begin{lemma}\label{lem:cost_estimate}
Let $\qry$ be an \RC query with \pwdist{}
and let $\rw{\qry}=(\safefin, \safeinf)$.
Let $\safeqry\in\{\safefin, \safeinf\}$.
Then
\[\textstyle
\cost{\safeqry}{\str}\leq
\cntsub{\safeqry}\cdot\card{\av{\safeqry}}\cdot 2^{\card{\av{\safeqry}}}\cdot
\sum_{\substack{\projatomset\subseteq\closeatomseq{\qry}, \projminimal}}
\prod_{\projoneatom\in\projatomset}\card{\sattup{\projoneatom}}.
\]
\end{lemma}

\begin{proof}
Recall that $\cntsub{\smash{\safeqry}}$ denotes the number of subqueries
of the query $\safeqry$ and thus bounds the number of
\safeprop{} subqueries $\safesubqry$ of the query $\safeqry$.
For every subquery $\safesubqry$ of $\safeqry$,
we first use the fact that $\card{\fv{\smash{\safesubqry}}}\leq\card{\av{\smash{\safeqry}}}$ to bound
$\card{\sattup{\smash{\safesubqry}}}\cdot\card{\fv{\smash{\safesubqry}}}\leq
\card{\sattup{\smash{\safesubqry}}}\cdot\card{\av{\smash{\safeqry}}}$.
Then we use the estimation of $\card{\sattup{\smash{\safesubqry}}}$
by Lemma~\ref{lem:complex_subquery_subset}.
\end{proof}

Finally, we prove Theorem~\ref{lem:data_complexity}.

\begin{proof}[Proof of Theorem~\ref{lem:data_complexity}]
We derive Theorem~\ref{lem:data_complexity} from Lemma~\ref{lem:cost_estimate}
and the fact that
the quantities $\cntsub{\safeqry}$, $\card{\av{\safeqry}}$,
and $2^{\card{\av{\safeqry}}^2}$
only depend on the query $\qry$
and thus they do not contribute
to the asymptotic time complexity of capturing a fixed query $\qry$.
\end{proof}

\section{Implementation}
\label{sec:detail}

\looseness=-1
We have implemented our translation \tool{} consisting of roughly $\loc$ lines of 
OCaml code~\cite{artifact}. It consists of multiple translation steps that take
an arbitrary relational calculus (RC) query and produce two SQL queries. 

\looseness=-1
Figure~\ref{fig:implementation} summarizes the order of the translation steps and 
the functions that implement them.
The function $\splitq{\cdot}$ (\S\ref{sec:free}), applied in the first step, is the 
main part of our translation. Recall that it takes an arbitrary RC query and 
returns two safe-range RC queries.
Next, the function $\pushnot{\cdot}$ (\S\ref{sec:srnf-detail}) converts both queries 
to safe-range normal form (SRNF), followed by the function 
$\allowtosafe{\cdot}{\cdot}$ (\S\ref{sec:ranf-detail}) 
that converts SRNF queries into relation algebra normal form (RANF). 
Both normal forms were defined in \S\ref{sec:prelim}.
For simplicity, we define a function $\allowtosafeqry{\cdot}$ that combines the 
previous two  functions and can be applied to any safe-range RC query.
In addition to the worst-case complexity, we further improve our translation's 
average-case complexity by implementing the optimizations inspired 
by~\citeauthorAgg~\cite{DBLP:conf/vldb/ClaussenKMP97}. The function 
$\cnt{\cdot}$ (\S\ref{sec:aggs}) implements these optimizations on the RANF queries.
Finally, to derive SQL queries from the RANF queries
we first obtain equivalent relational algebra (RA) expressions following a (slightly modified) 
standard approach~\cite{DBLP:books/aw/AbiteboulHV95} implemented by the function
$\ranfra{\cdot}$ (\S\ref{sec:ranf2ra}).
To translate the RA expressions into SQL, we reuse a publicly available RA 
interpreter \radb~\cite{radb} (\S\ref{sec:ra2sql}). We name the composition of the 
last two steps $\ranfsql{\cdot}$.

To resolve the nondeterministic choices present in our algorithms (\S\ref{sec:nondet})
we always choose the alternative with the lowest query cost. The query cost is estimated
by using a sample structure of constant size, called a \emph{training database}.
A good training database should preserve the relative ordering of queries by 
their cost over the actual database as much as possible.
Nevertheless, our translation satisfies the correctness
and worst-case complexity claims independently of the choice of the training database.

Overall, the translation is formally defined as 
\begin{center}
$  \tool{}(\qry)\eq(\qryfina,\qryinfa) $
\end{center}
where
%\[
%\begin{array}{@{}r@{\;}c@{\;}l@{}}
$\qryfina \eq \ranfsql{\cnt{\allowtosafeqry{\qryfin}}}$,  
$\qryinfa \eq \ranfsql{\cnt{\allowtosafeqry{\qryinf}}}$, and
$(\qryfin,\qryinf) \eq \splitq{\qry}$.
\pagebreak[2]
%\end{array}
%\]

\begin{figure}[t]
  \small
  \begin{tikzpicture}
  \tikzset{pt/.style={draw,rounded corners=3pt,inner sep=5pt, fill=white}, link/.style={-latex,line
  width=0.1mm,draw=black}, path/.style=={-latex,line width=0.1mm,draw=black},
  brace/.style={line width=0.1mm,draw=black, decoration={brace},decorate}}
  \node[pt] (rc) at (-3*\z,0.05) {\RC{}};

  \node[pt] (sr') at (-1.8*\z+\w,\h) {\phantom{Safe-range \RC{}}};
  \draw[link] (rc) -- (sr'.180);
  
  \node[pt] (sr) at (-1.8*\z,0) {Safe-range \RC{}};
  \draw[link] (rc) -- node[anchor=south,yshift=4, label=above:{\makecell[l]{$\splitq{\cdot}$}}]{}(sr.180);  

  \node[pt] (srnf') at (-0.7*\z+\w,\h) {\phantom{SRNF}};
  \draw[link] (sr') -- (srnf');
  \node[pt] (srnf) at (-0.7*\z,0) {SRNF};
  \draw[link] (sr) -- node (srrcsrnf) [anchor=south,yshift=4,label=above:{\makecell[l]{$\pushnot{\cdot}$}}]{}(srnf);
  
  \node[pt] (ranf') at (0.35*\z+\w,\h) {\phantom{\safeprop}};
  \draw[link] (srnf') -- (ranf');
  \path (ranf') edge [loop above] (ranf');

  \node[pt] (ranf) at (0.35*\z,0) {\safeprop};
  \draw[link] (srnf) -- node (srnfranf) [anchor=south, yshift=4,align=left,label=above:{\makecell[l]{
  $\allowtosafe{\cdot}{\cdot}$}}]{}(ranf);
  \path (ranf) edge [loop above] node[yshift=6,anchor=north,label=above:{$\cnt{\cdot}$}]{} (ranf);

  \node[pt] (ra') at (1.22*\z+\w,\h) {\phantom{RA}};
  \draw[link] (ranf') -- (ra');

  \node[pt] (ra) at (1.22*\z,0) {RA};
  \draw[link] (ranf) -- node (ranfra) [anchor=south,yshift=4,label=above:{\makecell[l]{$\ranfra{\cdot}$}}]{}(ra);

  \node[pt] (sql') at (1.85*\z+\w,\h) {\phantom{SQL}};
  \draw[link] (ra') --  (sql');

  \node[pt] (sql) at (1.85*\z,0) {SQL};
  
  \draw[link] (ra) -- node (rasql) [anchor=south,yshift=4,label=above:{\makecell[l]{$\rasql{\cdot}$}}]{}
    (sql);
  
  \node at ([yshift=45]srnf) {$\allowtosafeqry{\cdot}$};
  \node at ([yshift=45]ra) {$\ranfsql{\cdot}$};

  \draw [brace,decoration={raise=25}] (ranfra) -- (rasql);

  \draw [brace,decoration={raise=25}] (srrcsrnf) -- (srnfranf);

  \end{tikzpicture}
\vspace{-2ex}
\caption{Overview of the functions used in our implementation.}\label{fig:implementation}
%\vspace{-1.5ex}
\end{figure}

\subsection{Translation to SRNF}\label{sec:srnf-detail}
\begin{algorithm}[t]
  \SetKwInOut{Input}{input}
  \SetKwInOut{Output}{output}
  \SetKwProg{Fn}{function}{$\,=$}{}
  \Input{An \RC{} query $\qry$.}
  \Output{A SRNF query $\qrysrnf$ such that $\qry\equiv \qrysrnf$,
  $\fv{\qry}=\fv{\qrysrnf}$.}
  \BlankLine
  \Fn{$\pushnot{\qry}$}{
  \Switch{$\qry$}{
    \Case{$\neg \subqry$}{%
    	\Switch{$\subqry$}{%
    	  \lCase{$\neg \subqryc$} {%
    	    \Return $\pushnot{\subqryc}$\label{alg:push:negelim}%
    	  }
      	  \lCase{$\qrya\lor \qryb$}{%
            \Return $\pushnot{(\neg \qrya)\land(\neg \qryb)}$\label{alg:push:pushbegin}%
          }
          \lCase{$\qrya\land \qryb$}{%
            \Return $\pushnot{(\neg \qrya)\lor(\neg \qryb)}$\label{alg:push:pushend}%
          }%
          \Case{$\exists\varlist.\,\exqryvarlist$}{%
            \lIf{$\varlist\cap\fv{\exqryvarlist}=\emptyset$}{\Return $\pushnot{\neg \exqryvarlist}$}\label{alg:push:exelim}
            \Else{%
              \Switch{$\pushnot{\exqryvarlist}$}{%
        \lCase{$\qrya\lor \qryb$}{%
          \Return $\pushnot{(\neg\exists\varlist.\,\qrya)\land(\neg\exists\varlist.\,\qryb)}$%
        }
        \lOther{%
          \Return $\neg\exists\varlist\cap\fv{\exqryvarlist}.\,\pushnot{\exqryvarlist}$%
        }%
            }}
          }
          \lOther{\Return $\neg\pushnot{\subqry}$}
          }
    }
    \lCase{$\qrya\lor \qryb$}{%
      \Return $\pushnot{\qrya}\lor\pushnot{\qryb}$%
    }
    \lCase{$\qrya\land \qryb$}{%
      \Return $\pushnot{\qrya}\land\pushnot{\qryb}$%
    }
    \Case{$\exists\varlist.\,\exqryvarlist$}{
      \Switch{$\pushnot{\exqryvarlist}$}{
        \lCase{$\qrya\lor \qryb$}{%
          \Return $\pushnot{(\exists\varlist.\,\qrya)\lor(\exists\varlist.\,\qryb)}$\label{alg:push:distr}%
        }
        \lOther{%
          \Return $\exists\varlist\cap\fv{\exqryvarlist}.\,\pushnot{\exqryvarlist}$%
        }
      }
    }
    \lOther{\Return $\qry$}
  }
  }
\vspace*{-1ex}
\caption{Translation to SRNF.}
\label{alg:push}
\vspace*{-1ex}
\end{algorithm}

\looseness=-1
Figure~\ref{alg:push} defines the function $\pushnot{\qry}$
that yields a SRNF query equivalent to $\qry$.
The function $\pushnot{\qry}$ 
pushes negations downwards (Lines~\ref{alg:push:pushbegin}--\ref{alg:push:pushend}),
eliminates double negations (Line~\ref{alg:push:negelim}),
drops bound variables that do not occur in the query (Line~\ref{alg:push:exelim}),
and distributes existential quantifiers over disjunction (Line~\ref{alg:push:distr}).
The termination of $\pushnot{\qry}$
follows using the measure $\sz{\qry}$, shown in Figure~\ref{fig:size_measure},
that decreases for proper subqueries, after pushing negations and 
distributing existential quantification over disjunction.

Next we prove a lemma that we use as a precondition for 
translating safe-range queries in SRNF to queries in RANF.
\begin{lemma}\label{lem:srnf-to-ranf}
  Let $\qrysrnf$ be a query in SRNF. Then $\gen{\varx}{\neg \subqry}$
  does not hold for any variable~$\varx$
  and subquery $\neg\subqry$ of $\qrysrnf$.
\end{lemma}
\begin{proof}
  \looseness=-1
  Using Figure~\ref{fig:gen_con}, $\gen{\varx}{\neg \subqry}$
  can only hold if $\neg \subqry$ has the form $\neg\neg \qry$,
  $\neg(\qrya\lor \qryb)$,
  or $\neg(\qrya\land \qryb)$. The SRNF query $\qrysrnf$ cannot have a subquery
  $\neg \subqry$ that has any such form.
\end{proof}
\pagebreak[2]

\begin{figure}
    \small
    \[\begin{array}{l@{\;}c@{\;}l}
        \sz{\bot}=\sz{\top}&=&\sz{\varx\approx \termsymb}=1\\
        \sz{\predsymb(\termsymb_1, \ldots, \termsymb_{\arity(\predsymb)})}&=&1\\
        \sz{\neg \qry}&=&2\cdot\sz{\qry}\\
        \sz{\qrya\lor \qryb}&=& 2\cdot\sz{\qrya}+2\cdot\sz{\qryb}+2\\
        \sz{\qrya\land \qryb}&=& \sz{\qrya}+\sz{\qryb}+1\\
        \sz{\exists\varx.\,\exqryx}&=&2\cdot\sz{\exqryx}
    \end{array}\]
    \vspace{-2ex}
    \caption{Measure on \RC queries.}
    \label{fig:size_measure}
\end{figure}

\subsection{Translation to RANF}\label{sec:ranf-detail}

\begin{algorithm}[t!]
  \SetKwInOut{Input}{input}
  \SetKwInOut{Output}{output}
  \SetKwProg{Fn}{function}{$\,=$}{}
  \Input{A safe-range query $\qry\land\conjrestr$
  such that for all subqueries of the form $\neg \subqry$,
  $\gen{\varx}{\neg \subqry}$ does not hold for any variable $\varx$.}
  \Output{A RANF query $\safeqry$ and a subset of queries
  $\restra\subseteq\restr$ such that
  $\qry\land\conjrestr\equiv\safeqry\land\conjrestr$;
  for all $\str$ and $\valsymb$,
  $(\str, \valsymb)\models\safeqry\Longrightarrow(\str, \valsymb)\models\conjrestra$ holds;
  $\safeqry=\cp{\safeqry}$; and
  $\fv{\qry}\subseteq\fv{\safeqry}\subseteq\fv{\qry}\cup\fv{\restr}$,
  unless $\safeqry=\bot$.}
  \BlankLine
  \Fn{$\allowtosafe{\qry}{\restr}$}{
  \If{$\safe{\qry}$}{\Return $(\cp{\qry}, \emptyset)$}\label{alg:allowtosafe:ranf}
  \Switch{$\qry$}{
    \Case{$\varx\approx \vary$}{%
      \Return $\allowtosafe{\varx\approx \vary\land\conjrestr}{\emptyset}$%
    }%
    \Case{$\neg \subqry$}{%
      $\restra\leftarrow\{\restra\subseteq\restr\mid
      (\neg \subqry)\land\conjrestra\,\text{is safe range}\}$\;
      \If{$\restra=\emptyset$}{%
        $(\safesubqry, \_)\eq\allowtosafe{\subqry}{\emptyset}$\;
        \Return $(\cp{\neg \safesubqry}, \emptyset)$\;
      }%
      \lElse{%
        \Return $\allowtosafe{(\neg \subqry)\land\conjrestra}{\emptyset}$%
      }%
    }%
    \Case{$\qrya\lor \qryb$}{%
	  $\restra\leftarrow\{\restra\subseteq\restr\mid
      \bigvee_{\subqry\in\fldisj{\qry}}(\subqry\land\conjrestra)\,\text{is safe range}\}$\;%
      \lForEach{$\subqry\in\fldisj{\qry}$}{%
        $(\safesubqry, \_)\eq\allowtosafe{\subqry\land\conjrestra}{\emptyset}$%
      }%
      \Return $(\cp{\bigvee_{\subqry\in\fldisj{\qry}} \safesubqry}, \restra)$\;
    }
    \Case{$\qrya\land \qryb$}{
      $\negs\eq\{\subqry\in\flconj{\qry}\cup\restr\mid\isneg{\subqry}\}$;
      $\poss\eq(\flconj{\qry}\cup\restr)\setminus\negs$\;
      $\eqset\eq\{\subqry\in\poss\mid\iseq{\subqry}\}$;
      $\poss\eq\poss\setminus\eqset$\;
      $\neqset\eq\{\neg\subqry\in\negs\mid\iseq{\subqry}\}$;
      $\negs\eq\negs\setminus\neqset$\;
      \lForEach{$\subqry\in\poss$}{%
        $(\safesubqry, \restrf{\subqry})\eq\allowtosafe{\subqry}{(\poss\cup\eqset)\setminus\{\subqry\}}$}%
      \lForEach{$\neg \subqry\in\negs$}{%
        $(\safesubqry, \_)\eq\allowtosafe{\subqry}{\poss\cup\eqset}$}%
      $\restra\leftarrow\{\restra\subseteq\poss\mid
      \poss\subseteq\bigcup_{\subqry\in\restra} (\restrf{\subqry}\cup\{\subqry\})\}$\;
      \Return $(\cp{\sconj{\bigcup_{\subqry\in\restra} \{\safesubqry\}\cup
      \eqset\cup\bigcup_{\neg \subqry\in\negs} \{\neg \safesubqry\}\cup
      \neqset}},\allowbreak
      \bigcup_{\subqry\in\restra} (\restrf{\subqry}\cap\restr))$\;
    }
    \Case{$\exists \varlist.\,\exqryvarlist$}{
    \lIf{$\fv{\restr}\cap\varlist\neq\emptyset$}{%
    $\varlistb\leftarrow\{\varlistb\mid \len{\varlistb}=\len{\varlist}\,\text{and}\,((\fv{\exqryvarlist}\setminus\varlist)\cup\fv{\restr})\cap\varlistb=\emptyset\}$\label{alg:allowtosafe:capture}%
    }\label{alg:allowtosafe:nocapture}%
    \lElse{%
      $\varlistb:=\varlist$%
    }%
    $\exqryvarlistb\eq \exqryvarlist[\varlist\mapsto\varlistb]$\;
    $\restra\leftarrow\{\restra\subseteq\restr\mid
      \exqryvarlistb\land\conjrestra\,\text{is safe range}\}$\;
    $(\exsafeqryvarlistb, \_)\eq\allowtosafe{\exqryvarlistb\land\conjrestra}{\emptyset}$\;
    \Return $(\cp{\exists \varlistb.\,\exsafeqryvarlistb}, \restra)$\;}
    \lOther{\Return $(\cp{\qry}, \emptyset)$}
  }
}
  \caption{Translation of a safe-range query in SRNF to RANF.}
  \label{alg:allowtosafe}
\end{algorithm}

The function $\allowtosafe{\qry}{\restr}=(\safeqry, \restra)$,
defined in Figure~\ref{alg:allowtosafe},
where $\SRTORANFName$ stands for
\emph{safe range to relational algebra normal form},
takes a safe-range query $\qry\land\conjrestr$ in SRNF, or in 
existential normal form (ENF) (see Appendix~\ref{sec:nf})
and returns a \safeprop{} query $\safeqry$
%and a subset $\restra\subseteq\restr$
such that
$\qry\land\conjrestr\equiv\safeqry\land\conjrestr$.
To restrict variables in $\qry$, the function $\allowtosafe{\qry}{\restr}$
conjoins a subset of queries $\restra\subseteq\restr$ to $\qry$.
Given a safe-range query $\qry$, we first convert $\qry$ into SRNF
and set $\restr=\emptyset$.
Then we define $\allowtosafeqry{\qry}\eq \safeqry$, where $(\safeqry,\_)\eq\allowtosafe{\pushnot{\qry}}{\emptyset}$,
to be a \safeprop{} query $\safeqry$ equivalent to $\qry$.
The termination of $\allowtosafe{\qry}{\restr}$
follows from the lexicographic measure
$(2\cdot \sz{\qry}+\termphi{\qry}+2\cdot \sumszrestr+2\cdot \card{\restr}, \sz{\qry}+\sumszrestr)$.
Here $\sz{\qry}$ is defined in Figure~\ref{fig:size_measure},
$\termphi{\qry}\eq 1$ if $\qry$ is an equality between
two variables or the negation of a query, and $\termphi{\qry}\eq 0$ otherwise.

Next we describe the definition of $\allowtosafe{\qry}{\restr}$
that follows \cite[Algorithm~5.4.7]{DBLP:books/aw/AbiteboulHV95}.
Note that no constant propagation (Figure~\ref{fig:cp})
is needed in \cite[Algorithm~5.4.7]{DBLP:books/aw/AbiteboulHV95},
because the constants $\bot$ and $\top$ are not
in the query syntax~\cite[\S 5.3]{DBLP:books/aw/AbiteboulHV95}.
Because $\gen{\varx}{\bot}$ holds and $\varx\notin\fv{\bot}$,
we need to perform constant propagation to guarantee
that every disjunct has the same set of free variables
(e.g., the query $\bot\lor\predbrand(\varx)$ must be translated
to $\predbrand(\varx)$ to be in \safeprop{}).
We flatten the disjunction and conjunction
using $\fldisj{\cdot}$ and $\flconj{\cdot}$, respectively.
In the case of a conjunction $\qryconj$,
we first split the queries from $\flconj{\qryconj}$ and $\restr$
into queries $\poss$ that do not have the form of a negation
and queries $\negs$ that do.
Then we take out equalities between two variables and negations
of equalities between two variables from the sets $\poss$ and $\negs$,
respectively.
To partition $\flconj{\qryconj}\cup\restr$ this way, we define the predicates $\isneg{\qry}$
and $\iseq{\qry}$
characterizing equalities between two variables and negations, respectively, i.e.,
$\isneg{\qry}$ is true iff $\qry$ has the form $\neg\subqry$ and
$\iseq{\qry}$ is true iff $\qry$ has the form $\varx\approx \vary$.
Finally, the function $\sconj{\disjs}$
converts a set of queries into a RANF conjunction, 
defined in Figure~\ref{fig:safe},
i.e., a left-associative conjunction in RANF.
Note that the function $\sconj{\disjs}$ must order
the queries $\varx\approx \vary$ so that either $\varx$ or $\vary$
is free in some preceding conjunct,
e.g., $\predbrand(\varx)\land \varx\approx \vary\land \vary\approx \varz$ is in RANF,
but $\predbrand(\varx)\land \vary\approx\varz\land \varx\approx \vary$ is not.
In the case of an existentially quantified query
$\exists\varlist.\,\exqryvarlist$, we rename
the variables $\varlist$ to avoid a clash of the free variables in the set
of queries $\restr$ with the bound variables $\varlist$.

Finally, we resolve the nondeterministic choices in
$\allowtosafe{\qry}{\restr}$ by minimizing the cost of the resulting \safeprop{} query
with respect to a training database (\S\ref{sec:nondet}).

\subsection{Optimization using Count Aggregations}
\label{sec:aggs}

In this section, we introduce count aggregations
and describe a generalization of Clau\ss{}en et al.~\cite{DBLP:conf/vldb/ClaussenKMP97}'s
approach to evaluate RANF queries using count aggregations.
Consider the query
\[
\qryx\land\neg\exists \vary.\,(\qryx\land\qryy\land\neg\qryxy),
\]
where $\fv{\qryx}=\{\varx\}$, $\fv{\qryy}=\{\vary\}$,
and $\fv{\qryxy}=\{\varx, \vary\}$.
This query is obtained by applying our translation to the query
$\qryx\land\forall \vary.\,(\qryy\longrightarrow\qryxy)$.
The cost of the translated query is dominated by
the cost of the Cartesian product $\qryx\land\qryy$.
Consider the subquery
$\subqry\eq\exists \vary.\,(\qryx\land\qryy\land\neg\qryxy)$.
A assignment $\valsymb$ satisfies $\subqry$ iff $\valsymb$ satisfies $\qryx$
and there exists a value $\domval$ such that $\valsymb[\vary\mapsto\domval]$
satisfies $\qryy$, but not $\qryxy$, i.e.,
the number of values $\domval$ such that $\valsymb[\vary\mapsto\domval]$
satisfies $\qryy$ is not equal to the number of values $\domval$
such that $\valsymb[\vary\mapsto\domval]$ satisfies both $\qryy$ and $\qryxy$.
An alternative evaluation of $\subqry$
evaluates the queries $\qryx$, $\qryy$, $\qryy\land\qryxy$
and computes the numbers of values $\domval$
such that $\valsymb[\vary\mapsto\domval]$ satisfies
$\qryy$ and $\qryy\land\qryxy$, respectively,
i.e., computes count aggregations.
These count aggregations are then used to filter
assignments $\valsymb$ satisfying $\qryx$
to get assignments $\valsymb$ satisfying $\subqry$.
The asymptotic time complexity of the alternative evaluation
never exceeds that of the evaluation computing
the Cartesian product $\qryx\land\qryy$ and asymptotically improves it if
$\card{\sattup{\qryx}}+\card{\sattup{\qryy}}+\card{\sattup{\qryxy}}\ll
\card{\sattup{\qryx\land\qryy}}$.
Furthermore, we observe that a assignment $\valsymb$ satisfies
$\qryx\land\neg \subqry$ if $\valsymb$ satisfies $\qryx$, but not $\subqry$,
i.e., the number of values $\domval$ such that
$\valsymb[\vary\mapsto\domval]$ satisfies $\qryy$ is equal to the number
of values $\domval$ such that
$\valsymb[\vary\mapsto\domval]$ satisfies $\qryy\land\qryxy$.

Next we introduce the syntax and semantics of count aggregations.
We extend \RC{}'s syntax by $\aggqry$,
where $\qry$ is a query, $\aggresult$ is a variable
representing the result of the count aggregation,
and $\aggbound$ is a sequence of variables that are bound
by the aggregation operator.
The semantics of the count aggregation is defined as follows:
\[
(\str, \valsymb)\models \aggqry\;\;\text{iff}\;\;
(\agggroup=\emptyset\longrightarrow\fv{\qry}\subseteq\aggbound)\,\text{and}\, \valsymb(\aggresult)=\card{\agggroup},
\]
where $\agggroup=\{\domvallist\in \dom^{\len{\aggbound}}\mid
(\str, \valsymb[\aggbound\mapsto\domvallist])\models \qry\}$.
We use the condition
$\agggroup=\emptyset\longrightarrow\fv{\qry}\subseteq\aggbound$
instead of $\agggroup\neq\emptyset$
to set $\aggresult$ to a zero count if the group $\agggroup$ is empty and
there are no group-by variables (like in SQL).
The set of free variables in a count aggregation is
$\fv{\aggqry}=(\fv{\qry}\setminus\aggbound)\cup\{\aggresult\}$.
Finally, we extend the definition of $\safe{\qry}$
with the case of a count aggregation:
\[
\safe{\aggqry}\;\;\text{iff}\;\;\safe{\qry}
\,\text{and}\,\aggbound\subseteq\fv{\qry}
\,\text{and}\, \aggresult\notin\fv{\qry}.
\]

We formulate translations introducing count aggregations
in the following two lemmas.

\begin{lemma}\label{lem:aggs}
Given $\restr\neq\emptyset$, let $\exists\varlist.\,\exqryvarlist\land\aggrestr$
be a \safeprop{} query.
Let $\aggresult$, $\aggresultc$ be fresh variables that do not occur in
$\fv{\exqryvarlist}$.
Then
\[
\begin{array}{@{}l@{\,}c@{\,}l@{\,}l@{}}
(\exists\varlist.\,\exqryvarlist\land\aggrestr)&\equiv
&\multicolumn{2}{@{}l@{}}{((\exists\varlist.\,\exqryvarlist)
\land\aggrestrdist)\lor{}}\\
&&(\exists \aggresult, \aggresultc.&\aggqry\land{}\\
&&&\aggqryrestr\land\neg (\aggresult=\aggresultc)).
\end{array}
\tag{\lagg}
\]
Moreover, the right-hand side of (\lagg) is in \safeprop{}.
\end{lemma}

\begin{lemma}\label{lem:agggs}
Given $\restr\neq\emptyset$, let $\safeqry\land\neg\exists\varlist.\,\exqryvarlist\land\aggrestr$,
be a \safeprop{} query.
Let $\aggresult$, $\aggresultc$ be fresh variables that do not occur in
$\fv{\safeqry}\cup\fv{\exqryvarlist}$.
Then
\[
\begin{array}{@{}l@{\,}c@{\,}@{\,}l@{}l@{}}
(\safeqry\land\neg\exists\varlist.\,\exqryvarlist\land\aggrestr)&\equiv
&\multicolumn{2}{@{}l@{}}{(\safeqry\land\neg(\exists\varlist.\,\exqryvarlist))\lor{}}\\
&&(\exists \aggresult,\aggresultc.\,\safeqry\land{}&\aggqry\land{}\\
&&&\aggqryrestr\land(\aggresult=\aggresultc)).
\end{array}
\tag{\laggg}
\]
Moreover, the right-hand side of (\laggg) is in \safeprop{}.
\end{lemma}

Note that the query cost does not decrease after applying the translation
(\lagg) or (\laggg)
because of the subquery $\aggqry$
in which $\exqryvarlist$ is evaluated before the count aggregation is computed.
For the query $\exists \vary.\,((\qryx\land\qryy)\land\neg\qryxy)$ from before,
we would compute $\aggaux{\vary}{\aggresult}{\qryx\land\qryy}$,
i.e., we would not (yet) avoid computing the Cartesian product $\qryx\land\qryy$.
However, we could reduce the scope of the bound variable $\vary$
by further translating
\[
\aggaux{\vary}{\aggresult}{\qryx\land\qryy}\equiv
\qryx\land\aggaux{\vary}{\aggresult}{\qryy}.
\]
\looseness=-1
This technique, called \emph{mini-scoping},
can be applied to a count aggregation $\aggqry$
if the aggregated query $\exqryvarlist$
is a conjunction that can be split into two \safeprop{} conjuncts
and the variables $\aggbound$ do not occur free in one of the conjuncts
(that conjunct can be pulled out of the count aggregation).
Mini-scoping can be analogously applied to queries of the form
$\exists\varlist.\,\exqryvarlist$.

Moreover, we can split a count aggregation over a conjunction
$\qryvarlista\land\qryvarlistb$
into a product of count aggregations if
the conjunction can be split into two \safeprop{} conjuncts
with disjoint sets of bound variables, i.e.,
$\varlist\cap\fv{\qryvarlista}\cap\fv{\qryvarlistb}=\emptyset$:
\[
\aggaux{\varlist}{\aggresult}{\qryvarlista\land\qryvarlistb}\equiv
(\exists\aggresulta, \aggresultb.\,
\aggaux{\varlist\cap\fv{\qryvarlista}}{\aggresulta}{\qryvarlista}\land
\aggaux{\varlist\cap\fv{\qryvarlistb}}{\aggresultb}{\qryvarlistb}\land
\aggresult=\aggresulta\cdot\aggresultb).
\]
Here $\aggresulta$ and $\aggresultb$ are fresh variables that do not occur
in $\fv{\qryvarlista}\cup\fv{\qryvarlistb}\cup\{\aggresult\}$.
Note that mini-scoping is only a heuristic and it can both improve
and harm the time complexity of query evaluation.
We leave the application of other more general optimization algorithms~\cite{KhamisNR16, DBLP:journals/sigmod/OlteanuS16}) as future work.
%for factorized databases and functional aggregate queries can be readily 

We implement the translations from Lemmas~\ref{lem:aggs} and~\ref{lem:agggs}
and mini-scoping in the function $\cnt{\cdot}$.
Given a RANF query $\safeqry$,
$\cnt{\safeqry}$ is an equivalent RANF query after introducing count aggregations
and performing mini-scoping.
The function $\cnt{\safeqry}$ uses a training database to decide
how to apply the translations from Lemmas~\ref{lem:aggs} and~\ref{lem:agggs}
and mini-scoping.
More specifically, the function $\cnt{\safeqry}$
tries several possibilities and chooses one that minimizes
the query cost of the resulting RANF query.

\begin{example}\label{ex:aggs}
We show how to introduce count aggregations into the \safeprop{} query
\[
\safeqry\eq \qryx\land\neg\exists \vary.\,(\qryx\land\qryy\land\neg\qryxy).
\]
After applying the translation (\laggg) and mini-scoping to this query,
we obtain the following equivalent \safeprop{} query:
\[
\begin{array}{@{}l@{\,}c@{\,}l@{\,}l@{}}
\cnt{\safeqry}&\eq&\multicolumn{2}{@{}l@{}}{(\qryx\land\neg (\qryx\land\exists \vary.\,\qryy))\lor{}}\\
&&(\exists \aggresult, \aggresultc.\,\qryx\land{}&\aggaux{\vary}{\aggresult}{\qryy}\land%{}\\
%&&&
\aggaux{\vary}{\aggresultc}{\qryy\land\qryxy}\land(\aggresult=\aggresultc)).
\end{array}
\]
\end{example}

\subsection{Translating \safeprop{} to RA}\label{sec:ranf2ra}

Our translation of a RANF query into SQL has two steps: 
we first translate the query to an equivalent RA 
expression, which we then translate to SQL using
a publicly available RA interpreter \radb~\cite{radb}.

We define the function
$\ranfra{\safeqry}$ translating \safeprop{} queries $\safeqry$
into equivalent RA expressions $\ranfra{\safeqry}$.
The translation is based on 
Algorithm~5.4.8 by~\citeauthorAlice~\cite{DBLP:books/aw/AbiteboulHV95}, which we 
modify as follows.
We adjust the way closed \RC queries are handled.
\citeauthorIeee~\cite{DBLP:journals/tkde/ChomickiT95} observed that closed \RC queries 
cannot be handled by SQL, since SQL allows neither empty projections
nor 0-ary relations.
They propose to use a unary auxiliary predicate $\auxpredsymb\in \predsymbs$
whose interpretation $\auxpredsymb^{\str}=\{\auxpredtuple\}$
always contains exactly one tuple $\auxpredtuple$.
Every closed query $\exists \varx.\,\exqryx$ is then translated
into $\exists \varx.\,\auxpredsymb(\auxvar)\land \exqryx$
with an auxiliary free variable $\auxvar$.
Every other closed query $\safeqry$ is translated
into $\auxpredsymb(\auxvar)\land \safeqry$, e.g.,
$\predbrand(42)$ is translated into $\auxpredsymb(\auxvar)\land \predbrand(42)$.
We also use the auxiliary predicate $\auxpredsymb$ to translate queries
of the form $\varx\approx \cstsymb$ and $\cstsymb\approx \varx$ because
the single tuple $(\auxpredtuple)$ in $\auxpredsymb^{\str}$
can be mapped to any constant $\cstsymb$.
Finally, we extend \cite[Algorithm~5.4.8]{DBLP:books/aw/AbiteboulHV95}
with queries of the form $\aggqry$.

\subsection{Translating RA to SQL}\label{sec:ra2sql}

The \radb{} interpreter, abbreviated here by the function $\rasql{\cdot}$,
translates an RA expression into SQL
by simply mapping the RA connectives into their SQL counterparts.
The function $\rasql{\cdot}$ is primitive recursive on RA expressions.
We modify \radb{} to further improve performance of the query evaluation 
as follows.

A \safeprop{} query $\qrya\land\neg\qryb$, where $\safe{\qrya}$,
$\safe{\qryb}$, and $\fv{\qryb}\subseteq\fv{\qrya}$
is translated into RA expression $\gdiff{\ranfra{\qrya}}{\ranfra{\qryb}}$,
where $\gdiffop$ denotes the anti-join operator
and $\ranfra{\qrya}$, $\ranfra{\qryb}$ are the equivalent relational algebra
expressions for $\qrya$, $\qryb$, respectively.
The \radb{} interpreter only supports the anti-join operator
$\gdiff{\ranfra{\qrya}}{\ranfra{\qryb}}$ expressed as
$\ranfra{\qrya} - (\ranfra{\qrya} \bowtie \ranfra{\qryb})$,
where $-$ denotes the set difference operator
and $\bowtie$ denotes the natural join.
Alternatively, the anti-join operator
can be directly mapped to \verb|LEFT JOIN| in SQL.
We generalize \radb{} to use \verb|LEFT JOIN|
since it performs better in our empirical evaluation~\cite{artifact}.

The \radb{} interpreter introduces a separate SQL subquery 
in a \verb|WITH| clause for every subexpression in the RA expression.
We extend \radb{} to additionally perform common subquery elimination, i.e.,
to merge syntactically equal subqueries.
Common subquery elimination is also assumed
in our query cost (\S\ref{sec:cost}).

Finally, the function $\ranfsql{\safeqry}$ is defined as
$\ranfsql{\safeqry}\eq\rasql{\ranfra{\safeqry}}$, 
i.e., as a composition of the two translations from \safeprop{}
to RA and from RA to SQL.

\subsection{Resolving Nondeterministic Choices}\label{sec:nondet}

To resolve the nondeterministic choices in our algorithms,
we suppose that the algorithms have access
to a \emph{training database} $\strtrain$ of constant size.
The training database is used to compare the cost
of queries over the actual database and thus it should preserve
the relative ordering of queries by their cost
over the actual database as much as possible.
Still, our translation satisfies the correctness
and worst-case complexity claims (\S\ref{sec:free}
and~\ref{sec:complex}) for every choice of the training database.
The training databases used in our empirical evaluation
are obtained using the function $\dgname$ (\S\ref{sec:golf})
with $\card{\posset{}} = \card{\negset{}}=\traindg$.
Because of its constant size,
the complexity of evaluating a query over
the training database is constant and does not impact
the asymptotic time complexity of evaluating the query
over the actual database using our translation.
There are two types of nondeterministic choices in our algorithms:
\begin{itemize}
\item  \looseness=-1 Choosing some $\nondetobj\leftarrow\nondetset$ in a while-loop.
As the while-loops always update $\nondetset$
with $\nondetset\eq(\nondetset\setminus\{\nondetobj\})\cup \nondetfun(\nondetobj)$
for some $\nondetfun$, the order in which the elements of $\nondetset$
are chosen does not matter.
\item Choosing a subset of queries $\restra\subseteq\restr$
in the function $\allowtosafe{\qry}{\restr}$.
Because $\allowtosafe{\qry}{\restr}$ yields a RANF query,
we enumerate all \emph{minimal} subsets
(a subset $\restra\subseteq\restr$ is minimal
if there exists no proper subset
$\restra'\subsetneq\restra$ that could be used instead of $\restra$)
and choose one that minimizes the query cost of the RANF query.
\item \looseness=-1 Choosing a variable $\varx\in \varseta$ and a set $\cpreds$ such that
$\tcov{\varx}{\qryrb}{\cpreds}$, where $\qryrb$ is a query with range-restricted bound variables
and $\varseta\subseteq\fv{\qryrb}$.
Observe that the measure $\sz{\qry}$ on queries, defined in Figure~\ref{fig:size_measure},
decreases for the queries in the premises of the rules for
$\tgen{\varx}{\qryrb}{\qpreds}$ and $\tcov{\varx}{\qryrb}{\cpreds}$,
defined in Figures~\ref{fig:gen_con} and~\ref{fig:cov_extra}.
Hence, deriving $\tgen{\varx}{\qryrb}{\qpreds}$ and $\tcov{\varx}{\qryrb}{\cpreds}$
either succeeds or gets stuck after at most $\sz{\qryrb}$ steps.
In particular, we can enumerate all sets $\cpreds$ such that $\tcov{\varx}{\qryrb}{\cpreds}$
holds.
Because we derive one additional query $\qryrb[\varx\mapsto \vary]$
for every $\vary\in\coleqs{\varx}{\cpreds}$
and a single query $\qryrb\land\colpredsqry{\cpreds}$,
we choose $\varx\in \varseta$ and $\cpreds$
minimizing $\card{\coleqs{\varx}{\cpreds}}$ as the first objective
and $\sum_{\projoneatom\in\colpreds{\cpreds}}\cost{\projoneatom}{\strtrain}$
as the second objective.
Our particular choice of $\cpreds$ with $\tcov{\varx}{\qryrb}{\cpreds}$ is merely a heuristic and does not provide
any additional guarantees compared to every other choice
of $\cpreds$ with $\tcov{\varx}{\qryrb}{\cpreds}$.
\end{itemize}
This process can be further improved by adopting query plan search heuristics as is done by most of the
existing database management systems.

\section{Data Golf Benchmark}\label{sec:golf}

\begin{algorithm}[t]
  \SetKwInOut{Input}{input}
  \SetKwInOut{Output}{output}
  \SetKwProg{Fn}{function}{$\,=$}{}
  \Input{An \RC{} query $\qry$ satisfying
  \propcon, \propcst, \proppred, \proprep,
    $\dgstrat\in\{0, 1\}$.}
  \Output{Two sets of variables $\poseqs$ and $\negeqs$ whose values
must be equal in every tuple in $\posset{\varseq}$ and $\negset{\varseq}$
when computing a Data Golf structure.}
  \BlankLine
    \Fn{$\dgeqs{\qry}{\dgstrat}$}{%
    \Switch{$\qry$}{%
      \lCase{$\predsymb(\termsymb_1, \ldots, \termsymb_{\arity(\predsymb)})$}{%
        \Return $(\emptyset, \emptyset)$%
      }%
      \lCase{$\varx\approx\vary$}{%
        \Return $(\{\varx, \vary\}, \emptyset)$%
      }%
      \Case{$\neg\subqry$}{%
         $(\poseqs, \negeqs)\eq\dgeqs{\subqry}{\dgstrat}$\;
        \Return $(\negeqs, \poseqs)$\;
      }%
      \Case{$\subqrya\lor \subqryb$ or $\subqrya\land \subqryb$}{%
        $(\poseqsa, \negeqsa)\eq\dgeqs{\subqrya}{\dgstrat}$\;
        $(\poseqsb, \negeqsb)\eq\dgeqs{\subqryb}{\dgstrat}$\;
        \lIf{$\dgstrat=0$}{%
          \Return $(\poseqsa\cup\poseqsb, \negeqsa\cup\negeqsb)$%
        }%
        \lElseIf{$\qry=\subqrya\lor\subqryb$}{%
          \Return $(\poseqsa\cup\negeqsb, \negeqsa\cup\negeqsb)$%
        }%
        \lElseIf{$\qry=\subqrya\land\subqryb$}{%
          \Return $(\poseqsa\cup\poseqsb, \poseqsa\cup\negeqsb)$%
        }%
      }%
      \lCase{$\exists \vary.\,\exqryy$}{%
        \Return $\dgeqs{\exqryy}{\dgstrat}$%
      }%
    }%
  }%
  \caption{Computing sets of variables $\poseqs$ and $\negeqs$
  to reflect equalities in a query when computing a Data Golf structure.}
  \label{alg:golfeqs}
\end{algorithm}

\begin{algorithm}[t]
  \SetKwInOut{Input}{input}
  \SetKwInOut{Output}{output}
  \SetKwProg{Fn}{function}{$\,=$}{}
  \Input{An \RC{} query $\qry$ satisfying
  \propcon, \propcst, \proppred, \proprep,
  a sequence of pairwise distinct variables $\varseq$,
  $\av{\qry}\subseteq\varseq$,
  sets of tuples $\posset{\varseq}$ and $\negset{\varseq}$ over $\varseq$
  such that all values of variables from $\av{\qry}$
  in these tuples are pairwise distinct (also across tuples) except that,
  in every tuple in $\posset{\varseq}$ ($\negset{\varseq}$),
  the variables in $\poseqs$ ($\negeqs$)
  have the same value (which is different across tuples), where
  $\dgeqs{\qry}{\dgstrat}=(\poseqs, \negeqs)$,
  $\dgstrat\in\{0, 1\}$.}
  \Output{A structure $\str$ such that
  $\evaltup{\posset{\varseq}}{\fvseq{\qry}}\subseteq\sattup{\qry}$,
  $\evaltup{\negset{\varseq}}{\fvseq{\qry}}\cap\sattup{\qry}=\emptyset$,
  and $\sattup{\subqry}$ and $\sattup{\neg \subqry}$
  contain at least $\min\{\card{\posset{\varseq}}, \card{\negset{\varseq}}\}$
  tuples, for every $\subqry\sqsubseteq \qry$.}
  \BlankLine
  \Fn{$\dg{\qry}{\varseq}{\posset{\varseq}}{\negset{\varseq}}{\dgstrat}$}{%
  \Switch{$\qry$}{%
    \lCase{$\predsymb(\termsymb_1, \ldots, \termsymb_{\arity(\predsymb)})$}{%
      \Return $\{\predsymb^\str\mapsto
        \evaltup{\posset{\varseq}}{\termsymb_1, \ldots,
        \termsymb_{\arity(\predsymb)}}\}$%
      }%
    \lCase{$\varx\approx \vary$}{%
       \Return $\emptyset$%
    }%
    \lCase{$\neg \subqry$}{%
      \Return $\dg{\subqry}{\varseq}{\negset{\varseq}}
      {\posset{\varseq}}{\dgstrat}$%
    }%
    \Case{$\subqrya\lor \subqryb$ or $\subqrya\land \subqryb$}{%
        $(\poseqsa, \negeqsa)\eq\dgeqs{\subqrya}{\dgstrat}$;
        $(\poseqsb, \negeqsb)\eq\dgeqs{\subqryb}{\dgstrat}$\;
        \lIf{$\dgstrat=0$}{%
          $(\vareqsa, \vareqsb)\eq(\poseqsa\cup\negeqsb, \negeqsa\cup\poseqsb)$%
        }%
        \lElseIf{$\qry=\subqrya\land\subqryb$}{%
          $(\vareqsa, \vareqsb)\eq(\negeqsa\cup\negeqsb, \negeqsa\cup\poseqsb)$%
        }%
        \lElseIf{$\qry=\subqrya\lor\subqryb$}{%
          $(\vareqsa, \vareqsb)\eq(\poseqsa\cup\poseqsb, \negeqsa\cup\poseqsb)$%
        }%
      $(\auxseta{\varseq}, \auxsetb{\varseq})\leftarrow
      \{(\auxseta{\varseq}, \auxsetb{\varseq})\mid
        \card{\auxseta{\varseq}}=\card{\auxsetb{\varseq}}=
        \min\{\card{\posset{\varseq}}, \card{\negset{\varseq}}\},\ 
        \text{all values in tuples in}\allowbreak
        \posset{\varseq}, \negset{\varseq},
        \auxseta{\varseq}, \auxsetb{\varseq}
        \text{ are chosen to be pairwise distinct (also across tuples)
        except that,}\allowbreak
        \text{in every tuple in }
        \auxseta{\varseq}\ (\auxsetb{\varseq}),
        \text{ the variables in } \vareqsa\ (\vareqsb)
        \text{ have the same value}
        \allowbreak
        \text{(which is different across tuples)}\}$\;%
    \lIf{$\dgstrat=0$}{%
      \Return $\dg{\subqrya}{\varseq}{\posset{\varseq}\cup \auxseta{\varseq}}
      {\negset{\varseq}\cup \auxsetb{\varseq}}{\dgstrat}\cup
      \dg{\subqryb}{\varseq}{\posset{\varseq}\cup \auxsetb{\varseq}}
      {\negset{\varseq}\cup \auxseta{\varseq}}{\dgstrat}$%
    }%
    \lElseIf{$\qry=\subqrya\lor \subqryb$}{%
      \Return $\dg{\subqrya}{\varseq}{\posset{\varseq}\cup \auxseta{\varseq}}
      {\negset{\varseq}\cup \auxsetb{\varseq}}{\dgstrat}\cup
      \dg{\subqryb}{\varseq}{\auxseta{\varseq}\cup \auxsetb{\varseq}}
      {\negset{\varseq}\cup\posset{\varseq}}{\dgstrat}$%
    }%
    \lElseIf{$\qry=\subqrya\land \subqryb$}{%
      \Return $\dg{\subqrya}{\varseq}{\posset{\varseq}\cup \negset{\varseq}}
      {\auxseta{\varseq}\cup \auxsetb{\varseq}}{\dgstrat}\cup
      \dg{\subqryb}{\varseq}{\posset{\varseq}\cup \auxsetb{\varseq}}
      {\negset{\varseq}\cup \auxseta{\varseq}}{\dgstrat}$%
    }%
    }%
    \lCase{$\exists \vary.\,\exqryy$}{%
      \Return $\dg{\exqryy}{\varseq}{\posset{\varseq}}
      {\negset{\varseq}}{\dgstrat}$%
    }%
  }%
  }
  \caption{Computing the Data Golf structure.}
  \label{alg:golf}
\end{algorithm}

\looseness=-1
In this section, we describe the \emph{Data Golf} benchmark, which we use to generate  
structures (i.e., database instances) for our empirical evaluation.
The technical description of this benchmark is only needed to fully understand the setup of our
empirical evaluation (\S\ref{sec:eeval}), but its details are independent of our query translation
(\S\ref{sec:translation}--\ref{sec:detail}).

Given an RC query, we seek a structure that yields a nontrivial evaluation result for the overall
query and for all its subqueries. Intuitively, the structure makes query evaluation
potentially more challenging compared to the case where some
subquery evaluates to a trivial (e.g., empty) result.
More specifically, Data Golf has two objectives.
The first resembles the \emph{regex golf} game's objective~\cite{regexgolf} (hence the name) and aims to find a structure on which the result of a given query contains a given \emph{positive} set of tuples and does not contain any tuples from another given \emph{negative} set. The second objective is to ensure that all the query's subqueries evaluate to a non-trivial result.

\looseness=-1
Formally, given a query $\qry$ and two sets of tuples $\posset{}$ and
$\negset{}$ over a fixed domain $\dom$,
representing assignments of $\av{\qry}$
and satisfying further assumptions on their values, Data Golf produces a
structure $\str$ (represented as a partial mapping from predicate symbols
to their interpretations) such that
the projections of tuples in $\posset{}$ ($\negset{}$, respectively)
to $\fvseq{\qry}$ are in $\sattup{\qry}$ (disjoint from $\sattup{\qry}$, respectively)
and $\card{\sattup{\subqry}}$ and $\card{\sattup{\neg \subqry}}$
are at least $\min\{\card{\posset{}}, \card{\negset{}}\}$, for every $\subqry\sqsubseteq \qry$.
To be able to produce such a structure $\str$,
we make the following assumptions on $\qry$:
\begin{enumerate}[(\proppred)]
\item[(\propcon)]
for every subquery $\exists \vary.\,\exqryy$ of $\qry$ we have $\con{\vary}{\exqryy}{\apreds}$
%\begin{short}\cite[Figure~5]{DBLP:journals/tods/GelderT91}\end{short}%
(Figure~\ref{fig:gen_con_vgt})
for some set of atomic predicates $\apreds$ and, moreover,
$\{\vary\}\subsetneq\fv{\oneatom}$ holds for every $\oneatom\in\apreds$;
these conditions prevent subqueries like $\exists \vary.\,\neg \predexb(\varx, \vary)$
and $\exists \vary.\,(\predexb(\varx, \vary)\lor \predexc(\vary))$, respectively;
\item[(\propcst)] $\qry$ contains no subquery of the form $\varx\approx\cstsymb$,
which is satisfied by exactly one tuple;
\item[(\proppred)] $\qry$ contains no closed subqueries, e.g.,
$\predexc(42)$, because a closed subquery is either satisfied
by all possible tuples or no tuple at all; and
\item[(\proprep)] $\qry$ contains no repeated predicate symbols
and no equalities $\varx\approx\vary$ in $\qry$
share a variable;
this avoids subqueries like $\predexc(\varx)\land\neg \predexc(\varx)$
and $\varx\approx\vary\land\neg\varx\approx\vary$.
\end{enumerate}

Given a sequence of pairwise distinct variables $\varseq$
and a tuple $\domvallist$ of the same length,
%$\len{\varseq}=\len{\domvallist}$,
we may interpret the tuple $\domvallist$ as a
\emph{tuple over} $\varseq$, denoted as $\tupover{\domvallist}{\varseq}$.
Given a sequence $\termsymb_1, \ldots, \termsymb_\listlength\in
\varseq\cup \consts$ of terms, we denote by
$\evaltup{\tupover{\domvallist}{\varseq}}{\termsymb_1, \ldots, \termsymb_\listlength}$
the tuple obtained by evaluating the terms
$\termsymb_1, \ldots, \termsymb_\listlength$
over $\tupover{\domvallist}{\varseq}$.
Formally, we define
$\evaltup{\tupover{\domvallist}{\varseq}}{\termsymb_1, \ldots, \termsymb_\listlength}\eq
(\domvalc_\iter)_{\iter=1}^\listlength$, where
$\domvalc_\iter=\domvallist_\idx$ if $\termsymb_\iter=\varseq_\idx$
and $\domvalc_\iter=\termsymb_\iter$ if $\termsymb_\iter\in \consts$.
We lift this notion to sets of tuples over $\varseq$
in the standard way.

Data Golf is formalized by the function
$\dg{\qry}{\varseq}{\posset{\varseq}}{\negset{\varseq}}{\dgstrat}$,
defined in Figure~\ref{alg:golf},
where $\varseq$ is a sequence of pairwise distinct variables
containing all variables in $\qry$, i.e., $\av{\qry}\subseteq\varseq$,
$\posset{\varseq}$ and $\negset{\varseq}$ are sets of tuples
over~$\varseq$, and $\dgstrat\in\{0, 1\}$ is a \emph{strategy}.
To reflect $\qry$'s equalities in the sets $\posset{\varseq}$
and $\negset{\varseq}$, given a strategy $\dgstrat$,
we define the function $\dgeqs{\qry}{\dgstrat}=(\poseqs, \negeqs)$
(Figure~\ref{alg:golfeqs})
that computes two sets of variables $\poseqs$ and $\negeqs$ whose values
must be equal in every tuple in $\posset{\varseq}$
and $\negset{\varseq}$, respectively.
The values of the remaining variables ($\varseq\setminus\poseqs$
and $\varseq\setminus\negeqs$, respectively) must be pairwise distinct
and also different from the value of the variables in $\poseqs$
and $\negeqs$, respectively.
In the case of a conjunction or a disjunction,
we add disjoint sets $\auxseta{\varseq}$, $\auxsetb{\varseq}$
of tuples over $\varseq$ to $\posset{\varseq}$, $\negset{\varseq}$
so that the intermediate results for the subqueries are neither equal
nor disjoint.
We implement two strategies (parameter $\dgstrat$)
to choose these sets $\auxseta{\varseq}$, $\auxsetb{\varseq}$.

Let $\str$ be a Data Golf structure
computed by $\dg{\qry}{\varseq}{\posset{\varseq}}{\negset{\varseq}}{\dgstrat}$. We justify why $\str$
satisfies $\evaltup{\posset{\varseq}}{\fvseq{\qry}}\subseteq\sattup{\qry}$
and $\evaltup{\negset{\varseq}}{\fvseq{\qry}}\cap\sattup{\qry}=\emptyset$.
We proceed by induction on the query $\qry$.
Because of (\proprep), the Data Golf structures for the subqueries
$\qrya$, $\qryb$ of a binary query $\qrya\lor\qryb$ or $\qrya\land\qryb$
can be combined using the union operator.
The only case that does not follow immediately is that
$\evaltup{\negset{\varseq}}{\fvseq{\qry}}\cap\sattup{\qry}=\emptyset$
for a query $\qry$ of the form $\exists\vary.\,\exqryy$.
We prove this case by contradiction.
Without loss of generality we assume that
$\fvseq{\exqryy}=\fvseq{\qry}\cdot \vary$.
Suppose that $\domvallist\in\evaltup{\negset{\varseq}}{\fvseq{\qry}}$
and $\domvallist\in\sattup{\qry}$.
Because $\domvallist\in\evaltup{\negset{\varseq}}{\fvseq{\qry}}$,
there exists some $\domval$ such that
$\domvallist\cdot\domval\in\evaltup{\negset{\varseq}}{\fvseq{\exqryy}}$.
Because $\domvallist\in\sattup{\qry}$, there exists some $\domvalc$ such that
$\domvallist\cdot\domvalc\in\sattup{\qryy}$.
By the induction hypothesis,
$\domvallist\cdot\domval\notin\sattup{\exqryy}$ and
$\domvallist\cdot\domvalc\notin\evaltup{\negset{\varseq}}{\fvseq{\exqryy}}$.
Because $\con{\vary}{\exqryy}{\apreds}$ holds for some $\apreds$
satisfying (\propcon), the query $\exqryy$ is equivalent to
$
%\exqryy\equiv
(\exqryy\land\bigor{\oneatom\in\apreds}\;\oneatom)\lor
%\bigor{\varz\in\coleqs{\vary}{\cpreds}} (\exqryy[\vary\mapsto \varz])\lor
\exqryy[\vary/\bot]
$.
We have $\domvallist\cdot\domvalc\in\sattup{\exqryy}$.
If the tuple $\domvallist\cdot\domvalc$ satisfies
$\exqryy[\vary/\bot]$,
then $\domvallist\cdot\domval\in\sattup{\exqryy}$ (contradiction)
because the variable $\vary$ does not occur in the
query $\exqryy[\vary/\bot]$
and thus its assignment in $\domvallist\cdot\domvalc$ can be arbitrarily changed.
Otherwise, the tuple $\domvallist\cdot\domvalc$ satisfies some
atomic predicate $\oneatom\in\apreds$
and (\propcon) implies $\{\vary\}\subsetneq\fv{\oneatom}$.
Hence, the tuples $\domvallist\cdot\domval$
and $\domvallist\cdot\domvalc$ agree on the assignment
of a variable $\varx\in\fv{\oneatom}\setminus\{\vary\}$.
Let $\possetc{\varseq}$ and $\negsetc{\varseq}$
be the sets in the recursive call of $\dgname$ on the atomic predicate
from $\oneatom$.
Because
$\domvallist\cdot\domval\in\evaltup{\negset{\varseq}}{\fvseq{\exqryy}}$
and $\evaltup{\negset{\varseq}}{\fvseq{\exqryy}}\subseteq
\possetc{\varseq}[\fvseq{\exqryy}]\cup\negsetc{\varseq}[\fvseq{\exqryy}]$,
the tuple $\domvallist\cdot\domval$ is in $\possetc{\varseq}[\fvseq{\exqryy}]\cup\negsetc{\varseq}[\fvseq{\exqryy}]$.
Because $\domvallist\cdot\domvalc$ satisfies
the quantified predicate $\oneatom$,
the tuple $\domvallist\cdot\domvalc$ is in $\possetc{\varseq}[\fvseq{\exqryy}]$.
Next we observe that the assignments of every variable
(in particular, $\varx$)
in the tuples from the sets $\possetc{\varseq}$, $\negsetc{\varseq}$
are pairwise distinct (there can only be equal values of variables
within a single tuple).
Because the tuples $\domvallist\cdot\domval$
and $\domvallist\cdot\domvalc$ agree on the assignment of $\varx$,
they must be equal, i.e.,
$\domvallist\cdot\domval=\domvallist\cdot\domvalc$ (contradiction).

The sets $\posset{\varseq}$, $\negset{\varseq}$ only grow
in $\dgname$'s recursion and
the properties (\propcon), (\propcst), (\proppred), and (\proprep{})
imply that $\qry$ has no closed subquery. Hence,
$\evaltup{\posset{\varseq}}{\fvseq{\qry}}\subseteq\sattup{\qry}$
and $\evaltup{\negset{\varseq}}{\fvseq{\qry}}\cap\sattup{\qry}=\emptyset$ imply that
$\card{\sattup{\subqry}}$ and $\card{\sattup{\neg \subqry}}$
contain at least $\min\{\card{\posset{\varseq}}, \card{\negset{\varseq}}\}$
tuples, for every $\subqry\sqsubseteq \qry$.

\begin{example}
Consider the query
$\qry\eq \neg\exists\vary.\,\predexb(\varx, \vary)\land\neg\predexa(\varx, \vary, \varz)$.
This query $\qry$ satisfies
\text{\normalfont(\propcon)}, \text{\normalfont(\propcst)}, \text{\normalfont(\proppred)}, and \text{\normalfont(\proprep)}. In particular,
$\con{\vary}{\predexb(\varx, \vary)\land\neg\predexa(\varx, \vary, \varz)}{\apreds}$
holds for $\apreds=\{\predexb(\varx, \vary)\}$ with
$\{\vary\}\subsetneq\fv{\predexb(\varx, \vary)}$.
We choose $\varseq=(\varx, \varz, \vary)$,
$\posset{\varseq}=\{(0, 4, 8),\allowbreak (2, 6, 10)\}$, and
$\negset{\varseq}=\{(12, 16, 20), (14, 18, 22)\}$. The function
$\dg{\qry}{\varseq}{\posset{\varseq}}{\negset{\varseq}}{\dgstrat}$ first
flips $\posset{\varseq}$ and $\negset{\varseq}$
because $\qry$'s main connective is negation.

For conjunction (a binary operator),
two additional sets of tuples are computed:
$\auxseta{\varseq}=\{(24, 28, 32), (26, 30, 34)\}$
and
$\auxsetb{\varseq}=\{(36, 40, 44), (38, 42, 46)\}$.
Depending on the strategy ($\dgstrat=0$ or $\dgstrat=1$),
one of the following structures is computed:
$\str_0=\{\predexb\mapsto\{(12, 20), (14, 22),\allowbreak (24, 32), (26, 34)\},
\predexa\mapsto\posset{\varx\vary\varz}\}$, or
$\str_1=\{\predexb\mapsto\{(12, 20), (14, 22), (0, 8), (2, 10)\},
\predexa\mapsto\posset{\varx\vary\varz}\}$,
where $\posset{\varx\vary\varz}=\{(0, 8, 4),\allowbreak (2, 10, 6), (24, 32, 28), (26, 34, 30)\}$.

The query $\predexc(\varx)\land\qry$ is satisfied
by the finite set of tuples $\posset{\varseq}[\varx, \varz]$ under the structure
$\str_1\cup\{\predexc\mapsto\{(0), (2)\}\}$
obtained by extending $\str_1$ ($\dgstrat=1$).
In contrast, the same query $\predexc(\varx)\land\qry$ is satisfied
by an infinite set of tuples including $\posset{\varseq}[\varx, \varz]$
and disjoint from $\negset{\varseq}[\varx, \varz]$ under the structure
$\str_0\cup\{\predexc\mapsto\{(0), (2)\}\}$
obtained by extending $\str_0$ ($\dgstrat=0$).
\end{example}

\section{Empirical Evaluation}
\label{sec:eeval}

We empirically validate the evaluation performance of the queries output by \tool{}.
We also assess \tool{}'s translation time, the average-case time complexity of query evaluation, scalability to large
databases,
and DBMS interoperability.
To this end, we answer the following research questions:
\begin{enumerate}[RQ1]
\item\label{rq:compare} How does \tool{}'s query evaluation perform compared to the state-of-the-art tools on both
domain-independent and domain-dependent queries?
\item\label{rq:scale} How does \tool{}'s query evaluation scale on large synthetic databases?
\item\label{rq:real} How does \tool{}'s query evaluation perform on real-world databases?
\item\label{rq:cnt} How does the count aggregation optimization impact \tool{}'s performance?
\item\label{rq:dbms} Can \tool{} use different DBMSs for query evaluation?
\item\label{rq:translate} How long does \tool{} take to translate different queries (without query evaluation)?
\end{enumerate}

We organize our evaluation into five experiments. Four experiments (\smallexp{}, \mediumexp{},
\largeexp{}, and \realexp) focus on the type and size of the structures we use for query evaluation. The fifth experiment (\infexp{}) focuses on the
evaluation of non-evaluable (i.e., domain-dependent)
queries that may potentially produce infinite evaluation results.

To answer~\ref{rq:compare}, we compare our tool with 
the translation-based approach by~\citeauthorVanGT~\cite{DBLP:journals/tods/GelderT91} (\vgtool{}), 
the structure reduction approach by~\citeauthorAil~\cite{ail}, and 
the 
  \ddd{}~\cite{DBLP:conf/csl/MollerLAH99, DBLP:conf/cade/Moller02}, 
  \ldd{}~\cite{DBLP:conf/fmcad/ChakiGS09}, and
  \mpreg{}~\cite{DBLP:journals/jacm/BasinKMZ15} tools 
that evaluate \RC{} queries directly using infinite relations encoded as binary decision diagrams.
We could not find a publicly available implementation of \citeauthorVanGT's translation. Therefore,
the tool \vgtool{} for evaluable \RC queries is derived from our implementation by modifying the
function $\allowed{\cdot}$ in Figure~\ref{alg:allowed} to use the 
%\begin{short}$\mathsf{con}$\end{short} 
relation
%\begin{short}\cite[Figure~5]{DBLP:journals/tods/GelderT91}\end{short}
$\con{\varx}{\qry}{\apreds}$ (Appendix~\ref{sec:eval}, Figure~\ref{fig:gen_con_vgt})
instead of $\tcov{\varx}{\qry}{\cpreds}$ (Figure~\ref{fig:cov_extra}) 
and to use the generator $\bigor{\oneatom\in\apreds}\exists\fvseq{\qry}\setminus\{\varx\}.\,\oneatom$ instead of $\colpredsqry{\cpreds}$.
Evaluable queries $\qry$ are always translated into $(\qryfin, \bot)$ by $\rw{\cdot}$ because all of
$\qry$'s free variables are range restricted.
We exclude \vgtool{} from the comparison on non-evaluable queries (experiment \infexp{}).
Similarly, the implementation of
~\citeauthorAil's approach % that uses an extended active domain as the generators
was not available; 
hence we used our formally-verified implementation~\cite{DBLP:journals/afp/EvalFO}.
The implementations of the remaining tools were publicly available. 

We use Data Golf structures of growing size (experiments \smallexp{}, \mediumexp{}, 
and \largeexp{}) to answer~\ref{rq:scale}. In contrast, to answer~\ref{rq:real}, we use 
real-world structures obtained from the Amazon review dataset~\cite{DBLP:conf/emnlp/NiLM19}
(experiment \realexp{}).

To answer~\ref{rq:cnt}, we also consider variants of the translation-based approaches without the
step that uses count aggregation optimization $\cnt{\cdot}$, superscripted with a minus ($^-$).

SQL queries computed by the translations are evaluated using the \psql{} and \msql{}
DBMS (\ref{rq:dbms}). We superscript the tool names with \psqlsub{} and \msqlsub{} accordingly. 
In the \largeexp{} experiment, we only use \psql{} because it 
consistently performed better
than \msql{} in the \mediumexp{} experiment.
In all our experiments, the translation-based tools used a Data Golf structure with
$\card{\posset{}} = \card{\negset{}}=\traindg$
as the training database.
We run our experiments on an AMD Ryzen 7 PRO 4750U computer with 32 GB RAM.
The relations in \psql{} and \msql{} are recreated before each
invocation to prevent optimizations based on caching recent query evaluation
results. We measure the query evaluation times 
of all the tools and the translation time of our \tool{} tool (\ref{rq:translate}).
We provide all our experiments
in an easily reproducible and publicly available artifact~\cite{artifact}.

\begin{table}
  \setlength{\tabcolsep}{10pt}

\begin{tabular}{r c c c c c}
  \toprule
  Experiments: & \smallexp & \mediumexp & \largeexp & \infexp & \realexp \\

  \midrule

  \tool & \checkmark & \checkmark & \checkmark & \checkmark$^*$ & \checkmark \\
  \vgtool & \checkmark & \checkmark & TO & N/A & \checkmark \\
  \ddd & \checkmark & TO & TO & \checkmark & \checkmark \\
  \ldd & \checkmark & TO & TO & \checkmark & \checkmark\\
  \mpreg & \checkmark & TO & TO & \checkmark & \checkmark \\
  Ailamazyan et al. & TO & TO & TO & TO & TO \\
  \bottomrule
  \multicolumn{6}{l}{$*$ Only states that the result is infinite.}
\end{tabular}
\caption{Applicability and performance of all the tools on all the experiments. TO = Timeout of
\tosusp{} on all experiment runs, N/A = Not applicable}
\label{fig:tools}
\end{table}

In the \smallexp{}, \mediumexp{}, and \largeexp{} experiments, 
we generate ten pseudorandom queries (denoted as $Q_i, 1\leq i \leq 10$, see Appendix~\ref{sec:queries}) with a fixed size $\evalsz$
and  Data Golf structures $\str$ (strategy $\dgstrat=1$).
The queries satisfy the Data Golf assumptions along with a few additional ones: the queries are not safe range, every bound variable actually occurs in its scope,
disjunction only appears at the top-level, and only pairwise distinct variables appear as terms in predicates.
The queries have $\evalnfv$ free variables and every subquery has at most
$\evalmaxn$ free variables.
We control the size of the Data Golf structure $\str$ in our experiments
using a parameter $\eparam = \card{\posset{}} = \card{\negset{}}$.
Because the sets $\posset{}$ and $\negset{}$ grow in the recursion on subqueries,
relations in a Data Golf structure typically have
more than $\eparam$ tuples. 
The values of the parameter $\eparam$ for Data Golf structures are summarized
in Figure~\ref{fig:golfexps}.

The \infexp{} experiment consists of five pseudorandom queries $\qry^I_i, 1\leq i \leq 5$
(Appendix~\ref{sec:queries}) that are 
\emph{not} evaluable and $\rw{\qry^I_i}=(\qry_{i,fin}, \qry_{i,inf})$, 
where $\qry_{i,inf}\neq\bot$. Specifically, the queries are of the form
$\cpreda \land \forall \varx, \vary.\; \cpredb \longrightarrow \cpredc$, where $\cpreda, \cpredb$, and $\cpredc$
are either atomic predicates or equalities.
We choose the queries so that the number of their satisfying tuples
is not too high,
e.g., quadratic in the parameter $\eparam$,
because no tool can possibly enumerate so many tuples within the timeout.
For each $1\leq i \leq 5$, we compare the performance
of our tool to tools that directly evaluate $\qry^I_i$ on  
structures generated by the two Data Golf strategies (parameter $\dgstrat$), which trigger infinite or finite evaluation results on the considered queries.
For infinite results, our tool outputs this fact (by evaluating $\qry_{i,inf}$), whereas the other tools also output 
a finite representation of the infinite result. For finite results, all tools produce the same output.

\begin{figure}
\small
\def\expspace{8pt}
\begin{tabular}{@{}l@{\hspace*{2em}}r@{\cspace}r@{\cspace}r@{\cspace}r@{\cspace}r@{\cspace}r@{\cspace}r@{\cspace}r@{\cspace}r@{\cspace}r@{}l}
\toprule
\multicolumn{1}{@{}r@{\cspace}@{\cspace}}{Query}&$Q_1$&$Q_2$&$Q_3$&$Q_4$&\multicolumn{1}{r@{\cspace}@{\cspace}}{$Q_5$}&$Q_6$&$Q_7$&$Q_8$&$Q_9$&$Q_{10}$\\
\midrule
\multicolumn{12}{@{}c@{}}{Experiment \smallexp{}, Evaluable pseudorandom queries 
$\qry$, $\cntsub{\qry}=14$, $\eparam=500$:}\\
\midrule
\multicolumn{1}{@{}r@{\cspace}@{\cspace}}{\trtime}&0.6&0.0&0.1&0.0&\multicolumn{1}{r@{\cspace}@{\cspace}}{0.0}&0.0&0.0&0.1&0.0&0.0\\
\midrule
\tool\psqlsub&\textbf{0.2}&0.2&0.3&0.2&\multicolumn{1}{r@{\cspace}@{\cspace}}{0.2}&\textbf{0.2}&0.2&\textbf{0.1}&\textbf{0.2}&\textbf{0.1}\\
\tool\msqlsub&0.3&0.2&0.3&0.2&\multicolumn{1}{r@{\cspace}@{\cspace}}{0.2}&\textbf{0.2}&\textbf{0.1}&0.2&0.3&0.2\\
\toolnonopt\psqlsub&\textbf{0.2}&\textbf{0.1}&\textbf{0.2}&0.2&\multicolumn{1}{r@{\cspace}@{\cspace}}{\textbf{0.1}}&\textbf{0.2}&\textbf{0.1}&\textbf{0.1}&\textbf{0.2}&\textbf{0.1}\\
\toolnonopt\msqlsub&\textbf{0.2}&\textbf{0.1}&0.3&0.2&\multicolumn{1}{r@{\cspace}@{\cspace}}{\textbf{0.1}}&\textbf{0.2}&\textbf{0.1}&\textbf{0.1}&0.4&0.3\\
\midrule
\vgtool\psqlsub&1.5&0.2&1.8&1.3&\multicolumn{1}{r@{\cspace}@{\cspace}}{1.9}&8.6&2.5&1.5&15.6&4.8\\
\vgtool\msqlsub&0.3&0.2&0.3&\textbf{0.1}&\multicolumn{1}{r@{\cspace}@{\cspace}}{0.2}&56.3&6.1&0.2&155.7&13.5\\
\vgtoolnonopt\psqlsub&16.4&6.0&10.7&6.3&\multicolumn{1}{r@{\cspace}@{\cspace}}{TO}&6.2&3.1&2.3&48.0&8.8\\
\vgtoolnonopt\msqlsub&129.1&73.7&97.3&66.8&\multicolumn{1}{r@{\cspace}@{\cspace}}{TO}&52.1&19.1&12.8&TO&61.1\\
\midrule
\ddd&3.8&2.4&5.5&RE&\multicolumn{1}{r@{\cspace}@{\cspace}}{1.2}&3.7&4.7&2.2&17.9&5.3\\
\ldd&35.9&16.0&46.2&15.6&\multicolumn{1}{r@{\cspace}@{\cspace}}{9.0}&28.3&12.9&17.4&206.0&32.6\\
\mpreg&31.3&11.2&29.7&10.5&\multicolumn{1}{r@{\cspace}@{\cspace}}{10.2}&31.2&10.7&21.0&103.3&24.0\\
\noalign{\vspace{\expspace}}
\midrule
\multicolumn{12}{@{}c@{}}{\smallskip Experiment \mediumexp{}, Evaluable pseudorandom queries $\qry$,
$\cntsub{\qry}=14$, $\eparam=20000$:}\\
\midrule
\tool\psqlsub&2.1&1.1&2.2&1.2&\multicolumn{1}{r@{\cspace}@{\cspace}}{1.1}&1.2&\textbf{0.5}&0.9&1.8&\textbf{1.0}\\
\tool\msqlsub&6.2&2.6&6.6&2.7&\multicolumn{1}{r@{\cspace}@{\cspace}}{2.4}&4.4&1.6&2.7&9.0&3.3\\
\toolnonopt\psqlsub&\textbf{1.4}&\textbf{0.8}&\textbf{1.4}&\textbf{0.8}&\multicolumn{1}{r@{\cspace}@{\cspace}}{\textbf{0.8}}&\textbf{1.0}&\textbf{0.5}&\textbf{0.6}&\textbf{1.7}&\textbf{1.0}\\
\toolnonopt\msqlsub&4.3&1.9&5.0&1.8&\multicolumn{1}{r@{\cspace}@{\cspace}}{1.9}&3.3&1.6&2.2&8.3&3.2\\
\midrule
\vgtool\psqlsub&2.9&1.0&3.0&2.2&\multicolumn{1}{r@{\cspace}@{\cspace}}{2.8}&TO&TO&2.3&TO&TO\\
\vgtool\msqlsub&4.6&2.4&5.2&2.6&\multicolumn{1}{r@{\cspace}@{\cspace}}{2.7}&TO&TO&3.0&TO&TO\\
\vgtoolnonopt\psqlsub&TO&TO&TO&TO&\multicolumn{1}{r@{\cspace}@{\cspace}}{TO}&TO&TO&TO&TO&TO\\
\vgtoolnonopt\msqlsub&TO&TO&TO&TO&\multicolumn{1}{r@{\cspace}@{\cspace}}{TO}&TO&TO&TO&TO&TO\\
\noalign{\vspace{\expspace}}
\midrule
\multicolumn{12}{@{}c@{}}{\smallskip Experiment \largeexp{}, Evaluable pseudorandom queries $\qry$,
$\cntsub{\qry}=14$, tool = \tool\psqlsub:}\\
\midrule
$\eparam=40000$&4.5&2.3&4.4&2.3&\multicolumn{1}{r@{\cspace}@{\cspace}}{2.2}&2.5&1.0&1.7&3.7&1.8\\
$\eparam=80000$&8.8&4.5&8.7&4.8&\multicolumn{1}{r@{\cspace}@{\cspace}}{4.5}&5.0&1.8&3.3&7.2&3.7\\
$\eparam=120000$&14.1&6.8&12.8&7.2&\multicolumn{1}{r@{\cspace}@{\cspace}}{7.0}&7.1&2.8&5.1&10.8&4.9\\
\noalign{\vspace{\expspace}}
\midrule
\multicolumn{1}{@{}r@{\cspace}@{\cspace}}{Query}&$Q^I_1$&$Q^I_2$&$Q^I_3$&$Q^I_4$&\multicolumn{1}{r@{\cspace}@{\cspace}}{$Q^I_5$}&$Q^I_1$&$Q^I_2$&$Q^I_3$&$Q^I_4$&$Q^I_5$\\
\midrule
\multicolumn{12}{@{}c@{}}{\smallskip Experiment \infexp{}, Non-evaluable pseudorandom queries $\qry$,
$\cntsub{\qry}=7$, $\eparam=4000$:}\\
\midrule
&\multicolumn{5}{c@{\cspace}@{\cspace}}{$\strategya$}&\multicolumn{5}{c}{$\strategyb$}\\
\multicolumn{1}{@{}r@{\cspace}@{\cspace}}{\trtime}&0.0&0.0&0.0&0.0&\multicolumn{1}{r@{\cspace}@{\cspace}}{0.0}&0.0&0.0&0.0&0.0&0.0\\
\midrule
\tool\psqlsub&0.5&0.5&0.5&0.5&\multicolumn{1}{r@{\cspace}@{\cspace}}{0.5}&0.5&1.5&0.7&0.7&1.4\\
\tool\msqlsub&\textbf{0.3}&0.5&\textbf{0.3}&0.5&\multicolumn{1}{r@{\cspace}@{\cspace}}{0.5}&\textbf{0.4}&\textbf{0.7}&0.4&0.6&\textbf{0.6}\\
\toolnonopt\psqlsub&\textbf{0.3}&\textbf{0.3}&\textbf{0.3}&\textbf{0.3}&\multicolumn{1}{r@{\cspace}@{\cspace}}{\textbf{0.3}}&\textbf{0.4}&TO&\textbf{0.3}&\textbf{0.5}&TO\\
\toolnonopt\msqlsub&0.4&1.0&0.5&0.7&\multicolumn{1}{r@{\cspace}@{\cspace}}{0.9}&0.6&TO&0.6&0.6&TO\\
\midrule
\ddd&32.4&81.0&32.7&60.9&\multicolumn{1}{r@{\cspace}@{\cspace}}{81.6}&32.1&68.6&31.9&59.4&68.2\\
\ldd&TO&TO&TO&TO&\multicolumn{1}{r@{\cspace}@{\cspace}}{TO}&288.5&TO&TO&TO&TO\\
\mpreg&175.0&TO&175.4&TO&\multicolumn{1}{r@{\cspace}@{\cspace}}{TO}&160.3&299.1&160.0&TO&TO\\
\bottomrule\\
\end{tabular}
\vspace{-3ex}
%TO = Timeout of \togolf{}, RE = Runtime Error
\caption{Experiments \smallexp{}, \mediumexp{}, \largeexp{}, and \infexp{}.
We use the following abbreviations: TO = Timeout of \togolf{}, RE = Runtime Error. Reported
translation time for \tool.}
\label{fig:golfexps}
\end{figure}

Figure~\ref{fig:golfexps} shows the empirical evaluation results for the experiments \smallexp{}, \mediumexp{}, \largeexp{}, and \infexp{}.
All entries are execution times in seconds,
TO is a timeout, and RE is a runtime error.
In the experiments \smallexp{}, \mediumexp{}, and \largeexp{},
the columns correspond to ten unique pseudorandom queries
(the same queries are used in all the three experiments).
In the \infexp{} experiment, we use five
unique pseudorandom queries and two Data Golf strategies.
The time it takes for our translation \tool{} to translate
each query is shown in the first line
for the experiments \smallexp{} and \infexp{}
because the queries in the experiments \mediumexp{} and \largeexp{}
are the same as in \smallexp{}.
The remaining lines show evaluation times
with the lowest time for a query typeset in bold.
We omit the rows for tools that time out or crash on all queries of an experiment, e.g.,
~\citeauthorAil~\cite{ail}. Table~\ref{fig:tools} summarizes which tools were used in which
experiment. 
We conclude that our translation~\tool{} significantly outperforms
all other tools on all queries (\ref{rq:compare})
(except \vgtool{} on the fourth query, but on the smallest structure)
and scales well to higher values of $\eparam$,
i.e., larger relations in the Data Golf structures, on all queries (\ref{rq:scale}).
The count aggregation optimization (\ref{rq:cnt}) provides no conclusive benefit 
for our translation, while it consistently improves \vgtool{}'s performance, especially
in the \mediumexp{} experiment.

\begin{figure}[t]
\small
\begin{tabular}{@{}l@{\hspace*{2em}}rr@{\cspace}rr@{\cspace}rr@{\cspace}l@{\;\;\;}l@{\cspace}rr@{\cspace}rr@{\cspace}rr@{}}
\toprule
\multicolumn{1}{r}{Query}&\multicolumn{2}{c}{$\qsusp$}&\multicolumn{2}{c}{$\qsuspusr$}&\multicolumn{2}{c}{$\qsuspa$}&\multicolumn{1}{r@{\cspace}}{}&\multicolumn{2}{c}{$\qsusp$}&\multicolumn{2}{c}{$\qsuspusr$}&\multicolumn{2}{c}{$\qsuspa$}\\
\multicolumn{1}{r}{Param. $\eparam$}&$10^3$&$10^4$&$10^3$&$10^4$&$10^3$&$10^4$&\multicolumn{1}{r@{\cspace}}{Dataset}&GC&MI&GC&MI&GC&MI\\
\midrule
\multicolumn{1}{@{}r@{\cspace}@{\cspace}}{\trtime}&\multicolumn{2}{c@{\cspace}}{0.0}&\multicolumn{2}{c@{\cspace}}{0.0}&\multicolumn{2}{c}{0.3}&&\multicolumn{2}{c@{\cspace}}{0.0}&\multicolumn{2}{c@{\cspace}}{0.0}&\multicolumn{2}{c}{0.3}\\
\midrule
\tool\psqlsub&1.2&1.4&1.8&2.2&3.5&4.1&&\textbf{1.6}&\textbf{8.4}&2.5&\textbf{11.9}&4.9&\textbf{51.3}\\
\tool\msqlsub&\textbf{0.3}&\textbf{1.1}&\textbf{0.3}&\textbf{1.4}&\textbf{0.5}&\textbf{2.6}&&1.9&54.2&\textbf{2.4}&71.0&\textbf{4.5}&TO\\
\toolnonopt\psqlsub&28.1&TO&28.6&TO&228.6&TO&&132.9&TO&131.8&TO&TO&TO\\
\toolnonopt\msqlsub&TO&TO&TO&TO&TO&TO&&TO&TO&TO&TO&TO&TO\\
\midrule
\vgtool\psqlsub&1.5&1.7&\vgtna&\vgtna&204.3&TO&&1.9&10.0&\vgtna&\vgtna&TO&TO\\
\vgtool\msqlsub&\textbf{0.3}&1.4&\vgtna&\vgtna&TO&TO&&1.7&60.9&\vgtna&\vgtna&TO&TO\\
\vgtoolnonopt\psqlsub&TO&TO&\vgtna&\vgtna&TO&TO&&TO&TO&\vgtna&\vgtna&TO&TO\\
\vgtoolnonopt\msqlsub&TO&TO&\vgtna&\vgtna&TO&TO&&TO&TO&\vgtna&\vgtna&TO&TO\\
\midrule
\ddd&4.0&TO&4.1&TO&18.6&TO&&61.3&TO&61.0&TO&126.7&TO\\
\ldd&22.3&TO&22.2&TO&148.5&TO&&TO&TO&TO&TO&TO&TO\\
\mpreg&22.4&TO&22.6&TO&84.1&TO&&TO&TO&TO&TO&TO&TO\\
\bottomrule\\
\end{tabular}
\vspace{-3ex}
%GC = Gift Cards dataset, MI = Musical Instruments dataset, TO = Timeout of \tosusp{}
\caption{\looseness=-1 Experiment \realexp{} with the queries $\qsusp$, $\qsuspusr$, $\qsuspa$.
We use the following abbreviations: GC = Gift Cards, MI = Musical Instruments, TO = Timeout of
\tosusp{}. Reported translation time for \tool.}
\label{fig:suspexps}
\end{figure}

We also evaluate the tools on the queries
$\qsusp$ and $\qsuspusr$
from the introduction and on the more challenging query
$\qsuspa\eq \predbrand(\varbrand) \land \exists \varuser, \varscore, \vartext.\;\forall \varprod.\;\predprod(\varbrand,\varprod) \longrightarrow \predscore(\varprod, \varuser, \varscore)\lor\predtext(\varprod, \varuser, \vartext)$
with an additional relation~$\predtext$
that relates user's review text (variable $\vartext$) to a product.
The query $\qsuspa$ computes all brands for which
there is a user, a score, and a review text such that all the
brand's products were reviewed by that user with that score
or by that user with that text.
We use both  Data Golf structures
(strategy $\dgstrat=1$) and real-world structures obtained from
the Amazon review dataset~\cite{DBLP:conf/emnlp/NiLM19}.
The real-world relations $\predprod$, $\predscore$, and $\predtext$ are obtained
by projecting the respective tables from the Amazon review dataset
for two chosen product categories: gift cards (abbreviated GC) consisting of 147\,194 reviews of 1\,548 products and musical instruments (MI)  consisting of 1\,512\,530  reviews of 120\,400 products. The relation $\predbrand$ contains 
all brands from $\predprod$ that have at least three products.
Because the tool by \citeauthorAil{}, \ddd{}, \ldd{}, and \mpreg{} only support
integer data, we injectively remap the string
and floating-point values from the Amazon review dataset to integers.

\looseness=-1
Figure~\ref{fig:suspexps} shows the empirical evaluation results:
the time it takes for our translation \tool{} to translate
each query is shown in the first line
and the execution times on Data Golf structures (left)
and on structures derived from the real-world dataset
for two specific product categories (right)
are shown in the remaining lines.
We remark that \vgtool{} cannot handle the query $\qsuspusr$
as it is not evaluable~\cite{DBLP:journals/tods/GelderT91}, hence we
mark the correspond cells in Figure~\ref{fig:suspexps} with $-$.
Our translation \tool{}
significantly outperforms all other tools
(except \vgtool{} on $\qsusp$, where \tool{} and \vgtool{}
have similar performance)
on both Data Golf and real-world structures (\ref{rq:real}).
\vgtoolnonopt{} translates $\qsusp$ into a \safeprop{} query
with a higher query cost than \toolnonopt{}.
However, the optimization $\cnt{\cdot}$ manages to rectify this inefficiency (\ref{rq:cnt})
and thus \vgtool{} exhibits a comparable performance as \tool{}.
Specifically,
the factor of $80\times$ in query cost between \vgtoolnonopt{} and \toolnonopt{}
improves to $1.1\times$ in query cost between \vgtool{} and \tool{}
on a Data Golf structure with $\eparam = 20$~\cite{artifact}.
Nevertheless, \vgtool{} does not finish evaluating
the query $\qsuspa$ on GC and MI datasets within 5 minutes, unlike \tool{}.
Finally, \tool{}'s translation took less than 1 second on all the queries (\ref{rq:translate}).

\section{Conclusion}
\looseness=-1
We presented a translation-based approach to evaluating arbitrary relational calculus queries over an infinite domain with improved time complexity over existing approaches. This contribution is an important milestone towards making the relational calculus a viable query  language for practical databases. In future work, we plan to integrate into our base language features that database practitioners love,
such as inequalities, bag semantics, and aggregations.

\def\ackname{\relax Acknowledgments} {\paragraph{\ackname}\relax This research
    has been supported by the Swiss National Science Foundation grants ``Big Data
    Monitoring`` (167162) and ``Model-driven Security \& Privacy'' (204796) as well
    as by a Novo Nordisk Fonden start package grant (NNF20OC0063462). }

\bibliographystyle{alphaurl}
\bibliography{main}

\clearpage
\appendix

\section{Evaluable Queries}\label{sec:eval}

The classes of \emph{evaluable} queries~\cite[Definition~5.2]{DBLP:journals/tods/GelderT91}
and \emph{allowed} queries~\cite[Definition~5.3]{DBLP:journals/tods/GelderT91}
are decidable subsets of \di{} \RC queries.
The evaluable queries characterize
exactly the \di{} queries with no repeated predicate
symbols~\cite[Theorem~10.5]{DBLP:journals/tods/GelderT91}.
Every evaluable query can be translated
to an equivalent allowed query~\cite[Theorem 8.6]{DBLP:journals/tods/GelderT91}
and every allowed query can be translated to an equivalent
\safeprop{} query~\cite[Theorem 9.6]{DBLP:journals/tods/GelderT91}.

\begin{figure}[t]
  \small
  \[
  \begin{array}{@{}l@{\;\;}c@{\;\;}l@{}}
  \ruleif{\vgen{\varx}{\qry}{\{\qry\}}}{\edb{\qry}\,\text{and}\,\free{\varx}{\qry}}\\
  \ruleif{\vgen{\varx}{\neg\neg \qry}{\apreds}}{\vgen{\varx}{\qry}{\apreds}}\\
  \ruleif{\vgen{\varx}{\neg (\qrya\lor \qryb)}{\apreds}}{\vgen{\varx}{(\neg \qrya)\land(\neg \qryb)}{\apreds}}\\
  \ruleif{\vgen{\varx}{\neg (\qrya\land \qryb)}{\apreds}}{\vgen{\varx}{(\neg \qrya)\lor(\neg \qryb)}{\apreds}}\\
  \ruleif{\vgen{\varx}{\neg\exists \vary.\,\exqryy}{\apreds}}{\varx\neq \vary\,\text{and}\,\vgen{\varx}{\neg \exqryy}{\apreds}}\\
  \ruleif{\vgen{\varx}{\qrya\lor \qryb}{\apredsa\cup \apredsb}}{\vgen{\varx}{\qrya}{\apredsa}\,\text{and}\,\vgen{\varx}{\qryb}{\apredsb}}\\
  \ruleif{\vgen{\varx}{\qrya\land \qryb}{\apreds}}{\vgen{\varx}{\qrya}{\apreds}}\\
  \ruleif{\vgen{\varx}{\qrya\land \qryb}{\apreds}}{\vgen{\varx}{\qryb}{\apreds}}\\
  \ruleif{\vgen{\varx}{\exists \vary.\,\exqryy}{\apreds}}{\varx\neq \vary\,\text{and}\,\vgen{\varx}{\exqryy}{\apreds}}\\[10pt]
  \ruleif{\con{\varx}{\qry}{\emptyset}}{\absent{\varx}{\qry}}\\
  \ruleif{\con{\varx}{\qry}{\{\qry\}}}{\edb{\qry}\,\text{and}\,\free{\varx}{\qry}}\\
  \ruleif{\con{\varx}{\neg\neg \qry}{\apreds}}{\con{\varx}{\qry}{\apreds}}\\
  \ruleif{\con{\varx}{\neg (\qrya\lor \qryb)}{\apreds}}{\con{\varx}{(\neg \qrya)\,\text{and}\,(\neg \qryb)}{\apreds}}\\
  \ruleif{\con{\varx}{\neg (\qrya\land \qryb)}{\apreds}}{\con{\varx}{(\neg \qrya)\lor(\neg \qryb)}{\apreds}}\\
  \ruleif{\con{\varx}{\neg\exists \vary.\,\exqryy}{\apreds}}{\varx\neq \vary\,\text{and}\,\con{\varx}{\neg \exqryy}{\apreds}}\\
  \ruleif{\con{\varx}{\qrya\lor \qryb}{\apredsa \cup \apredsb}}{\con{\varx}{\qrya}{\apredsa}\,\text{and}\,\con{\varx}{\qryb}{\apredsb}}\\
  \ruleif{\con{\varx}{\qrya\land \qryb}{\apreds}}{\vgen{\varx}{\qrya}{\apreds}}\\
  \ruleif{\con{\varx}{\qrya\land \qryb}{\apreds}}{\vgen{\varx}{\qryb}{\apreds}}\\
  \ruleif{\con{\varx}{\qrya\land \qryb}{\apredsa \cup \apredsb}}{\con{\varx}{\qrya}{\apredsa}\,\text{and}\,\con{\varx}{\qryb}{\apredsb}}\\
  \ruleifdot{\con{\varx}{\exists \vary.\,\exqryy}{\apreds}}{\varx\neq \vary\,\text{and}\,\con{\varx}{\exqryy}{\apreds}}\\
  \end{array}
  \]
  %\vspace*{-3ex}
  \caption{The relations $\vgen{\varx}{\qry}{\apreds}$ and $\con{\varx}{\qry}{\apreds}$~\cite{DBLP:journals/tods/GelderT91}.}\label{fig:gen_con_vgt}
  \end{figure}

\begin{definition}\label{def:eval}
A query $\qry$ is called \emph{evaluable} if
\begin{itemize}
\item every variable $\varx\in\fv{\qry}$ satisfies
$\xgen{\varx}{\qry}$ and
\item the bound variable $\vary$ in every subquery $\exists\vary.\,\exqryy$
of $\qry$ satisfies $\xcon{\vary}{\exqryy}$.
\end{itemize}
A query $\qry$ is called \emph{allowed} if
\begin{itemize}
\item every variable $\varx\in\fv{\qry}$ satisfies
$\xgen{\varx}{\qry}$ and
\item the bound variable $\vary$ in every subquery $\exists \vary.\,\exqryy$
of $\qry$ satisfies $\xgen{\vary}{\exqryy}$,
\end{itemize}
where the relation $\xgen{\varx}{\qry}$ is defined to hold iff
there exists a set of atomic predicates $\apreds$ such that $\vgen{\varx}{\qry}{\apreds}$
and the relation $\xcon{\varx}{\qry}$ is defined to hold iff
there exists a set of atomic predicates $\apreds$ such that $\con{\varx}{\qry}{\apreds}$,
respectively.
The relations $\vgen{\varx}{\qry}{\apreds}$
and $\con{\varx}{\qry}{\apreds}$ are defined in Figure~\ref{fig:gen_con_vgt}.
\end{definition}

The termination of the rules
in Figure~\ref{fig:gen_con_vgt}
follow using the measure $\sz{\qry}$ (Figure~\ref{fig:size_measure}).
We now relate the definitions from Figure~\ref{fig:gen_con} and Figure~\ref{fig:gen_con_vgt} with the following lemmas.
\begin{lemma}\label{lem:vgen-neg-pos}
Let $\varx$ and $\vary$ be free variables in a query $\qry$ such that
$\xgen{\varx}{\neg \qry}$ and $\xgen{\vary}{\qry}$ hold. Then we get a contradiction.
\end{lemma}
\begin{proof}
This is proved by induction on the query $\qry$
using the measure $\sz{\qry}$ on queries defined in Figure~\ref{fig:size_measure},
which decreases in every case of the definition in Figure~\ref{fig:gen_con_vgt}.
\end{proof}
\begin{lemma}\label{lem:vgen-gen-one}
Let $\qry$ be a query
such that $\xgen{\vary}{\exqryy}$ holds for the bound variable $\vary$
in every subquery $\exists \vary.\,\exqryy$ of $\qry$.
Suppose that $\xgen{\varx}{\qry}$ holds for a free variable $\varx\in\fv{\qry}$.
Then $\gen{\varx}{\qry}$ holds.
\end{lemma}
\begin{proof}
This is proved by induction on the query $\qry$
using the measure $\sz{\qry}$ on queries defined in Figure~\ref{fig:size_measure},
which decreases in every case of the definition in Figure~\ref{fig:gen_con_vgt}.

Lemma~\ref{lem:vgen-neg-pos} and
the assumption that $\xgen{\vary}{\exqryy}$ holds for the bound variable $\vary$
in every subquery $\exists \vary.\,\exqryy$ of $\qry$ imply
that $\xgen{\varx}{\qry}$ cannot be derived
using the rule $\xgen{\varx}{\neg\exists \vary.\,\exqryy}$, i.e.,
$\qry$ cannot be of the form $\neg\exists \vary.\,\exqryy$.
Every other case in the definition of $\xgen{\varx}{\qry}$
has a corresponding case in the definition of $\gen{\varx}{\qry}$.
\end{proof}
\begin{lemma}\label{lem:vgen-gen}
Let $\qry$ be an \emph{allowed} query, i.e.,
$\xgen{\varx}{\qry}$ holds for every free variable $\varx\in\fv{\qry}$ and
$\xgen{\vary}{\exqryy}$ holds for the bound variable $\vary$ in every subquery
$\exists \vary.\,\exqryy$ of $\qry$.
Then $\qry$ is a safe-range query, i.e.,
$\gen{\varx}{\qry}$ holds for every free variable $\varx\in\fv{\qry}$ and
$\gen{\vary}{\exqryy}$ holds for the bound variable $\vary$ in every subquery
$\exists \vary.\,\exqryy$ of $\qry$.
\end{lemma}
\begin{proof}
The lemma is proved by applying Lemma~\ref{lem:vgen-gen-one}
to every free variable of $\qry$
and to the bound variable $\vary$ in every subquery of $\qry$
of the form $\exists \vary.\,\exqryy$.
\end{proof}
Lemma~\ref{lem:vgen-gen} shows that every allowed query is safe range.
But there exist safe-range queries that are not allowed,
e.g., $\predbrand(\varx)\land \varx\approx \vary$.

\section{Existential Normal Form}\label{sec:nf}

Recall that our translation uses a standard approach to obtain 
RANF queries from safe-range queries via SRNF~\cite{DBLP:books/aw/AbiteboulHV95}.
In this section we introduce existential normal form (ENF), an alternative 
normal form to SRNF that can also be used to translate safe-range queries
to RANF queries, and we discuss why we opt for using SRNF instead.

\begin{figure}[t]
\begin{tikzpicture}
\tikzset{pt/.style={draw,rounded corners=3pt,inner sep=5pt}, link/.style={-latex,line width=0.1mm,draw=black} }
\node[pt] (sr) at (0,2*\x) {Safe-range \RC{}};
\node[pt] (srnf) at (-1*\x,1*\x) {SRNF};
\node[pt] (enf) at (1*\x,1*\x) {ENF};
\node[pt] (ranf) at (0,0) {\safeprop};
% \node[pt] (ra) at (0,-1*\x) {SQL};

\draw[link] (sr)-- node[label=left:{$\pushnot{\qry}$}]{} (srnf);
\draw[link] (srnf)-- node[label=left:{$\allowtosafe{\qrysrnf}{\emptyset}$}]{} (ranf);
\draw[link] (sr)-- node[label=right:{$\enf{\qry}$}]{} (enf);
\draw[link] (enf)-- node[label=right:{$\allowtosafe{\qryenf}{\emptyset}$}]{} (ranf);
% \draw[link] (ranf)-- node[label=right:{$\ranfsql{\safeqry}$}]{} (ra);

\path [link,dashed] (sr.185) .. controls ($(sr.185)-(35mm,10mm)$) and ($(ranf.185)-(35mm,0mm)$) .. ([yshift=0mm]ranf.185);
   \node at ([yshift=-3mm, xshift=-15mm]ranf.180) {
    $\allowtosafeqry{\qry}$};
\end{tikzpicture}
%\vspace*{-2ex}
\caption{Alternative translation from safe-range RC to RANF query via ENF.}\label{fig:nf}
\end{figure}

Figure~\ref{fig:nf} shows an overview of the \RC fragments and query 
normal forms (nodes) and the functions we use to translate between 
them (edges).
The dashed edge shows the translation of a safe-range query to \safeprop{}
we opt for in this article. It is the composition of the two translations
from safe-range \RC{} to SRNF and from SRNF to \safeprop{}, respectively.
In the rest of this section we introduce ENF and the corresponding translations 
to RANF.

ENF was introduced by~Van Gelder and Topor~\cite{DBLP:journals/tods/GelderT91}
to translate an allowed query~\cite{DBLP:journals/tods/GelderT91} into an equivalent RANF query.
Given a safe-range query in ENF,
the rules $\rulea\text{--}\rulec$ from \S\ref{sec:ranf} can be applied
to obtain an equivalent RANF query~\cite[Lemma~7.8]{DBLP:conf/pods/Escobar-MolanoHJ93}.
We remark that the rules $\rulea\text{--}\rulec$
are not sufficient to yield an equivalent RANF query
for the original definition of ENF~\cite{DBLP:journals/tods/GelderT91}.
This issue has been identified and fixed by
\citeauthorEMHJ~\cite{DBLP:conf/pods/Escobar-MolanoHJ93}.
Unlike SRNF, a query in ENF can have a subquery of the form
$\neg(\qrya\land\qryb)$, but no subquery of the form
$\neg\qrya\lor\qryb$ or $\qrya\lor\neg\qryb$.
A function $\enf{\qry}$ that yields an ENF query
equivalent to $\qry$ can be defined in terms of subquery rewriting
using the rules in~\cite[Figure~2]{DBLP:conf/pods/Escobar-MolanoHJ93}.

Analogously to \cite[Lemma~7.4]{DBLP:conf/pods/Escobar-MolanoHJ93},
if a query $\qry$ is safe range, then $\enf{\qry}$ is also safe range.
Next we prove the following lemma that we could use as a precondition for
translating safe-range queries in ENF to queries in RANF.
\begin{lemma}\label{lem:enf-to-ranf}
  Let $\qryenf$ be a query in ENF.
  Then $\gen{\varx}{\neg \subqry}$ does not hold for any variable~$\varx$
  and subquery $\neg\subqry$ of $\qryenf$.
\end{lemma}
\begin{proof}
    \looseness=-1
  Assume that $\gen{\varx}{\neg \subqry}$ holds for a variable $\varx$
  in a subquery $\neg \subqry$ of $\qryenf$.
  We derive a contradiction by induction on $\sz{\qryenf}$.
  According to Figure~\ref{fig:gen_con} and by definition of ENF,
  $\gen{\varx}{\neg \subqry}$
  can only hold if $\subqry$ is a conjunction.
  Then $\gen{\varx}{\neg\subqry}$ implies $\gen{\varx}{\neg \qrya}$
  for some query $\qrya\in\flconj{\subqry}$
  that is not a negation (by definition of ENF)
  or conjunction (by definition of $\flconj{\cdot}$), i.e.,
  $\qrya$ is a disjunction (according to Figure~\ref{fig:gen_con}).
  Then $\gen{\varx}{\neg \qrya}$ implies $\gen{\varx}{\neg \qryb}$
  for some query $\qryb\in\fldisj{\qrya}$
  that is not a negation (by definition of ENF)
  or disjunction (by definition of $\fldisj{\cdot}$), i.e.,
  $\qryb$ is a conjunction (according to Figure~\ref{fig:gen_con}).
  Next we observe that $\neg\qryb$ is in ENF
  because $\qryb$ is a subquery of the ENF query $\qryenf$,
  $\qryb$ is a conjunction, and $\qryb$ is a subquery
  of a disjunction ($\qrya$) in $\qryenf$. Moreover,
  $\sz{\neg\qryb}<\sz{\qrya}<\sz{\subqry}<\sz{\qryenf}$.
  This allows us to apply the induction hypothesis to the ENF query $\neg\qryb$
  and its subquery $\neg\qryb$ (note that a query is a subquery of itself)
  and derive that $\gen{\varx}{\neg\qryb}$ does not hold, which is a contradiction.
\end{proof}

Although applying the rules $\rulea\text{--}\rulec$
to $\enf{\qry}$ instead of $\pushnot{\qry}$
may result in a RANF query with fewer subqueries,
the query cost, i.e., the time complexity of query evaluation,
can be arbitrarily larger.
We illustrate this in the following example
that is also included in our artifact~\cite{artifact}.
We thus opt for using SRNF instead of ENF for translating 
safe-range queries into RANF.

\begin{example}\label{ex:enf}
The safe-range query $\qryenf\eq \predexb(\varx, \vary)\land\neg (\predexc(\varx)\land \predexc(\vary))$ is in ENF and RANF,
but not SRNF.
Applying the rule $\rulea$ to $\pushnot{\qryenf}$ yields
the RANF query $\qrysrnf\eq (\predexb(\varx, \vary)\land\neg \predexc(\varx))\lor(\predexb(\varx, \vary)\land\neg \predexc(\vary))$
that is equivalent to $\qryenf$.
The costs of the two queries over a structure $\str$ are 
$ \cost{\qryenf}{\str}=2\cdot\card{\sattup{\predexb(\varx, \vary)}} + \allowbreak
           \card{\sattup{\predexc(\varx)}} + \allowbreak
           \card{\sattup{\predexc(\vary)}} + \allowbreak
           2\cdot\card{\sattup{\predexc(\varx)\land \predexc(\vary)}} + \allowbreak
           2\cdot\card{\sattup{\qryenf}}$  
and
$
\cost{\qrysrnf}{\str}=2\cdot\card{\sattup{\predexb(\varx, \vary)}} + \allowbreak
          \card{\sattup{\predexc(\varx)}} + \allowbreak
		  2\cdot\card{\sattup{\predexb(\varx, \vary)}} + \allowbreak
          \card{\sattup{\predexc(\vary)}} + \allowbreak
          2\cdot\card{\sattup{\predexb(\varx, \vary)\land\neg \predexc(\varx)}} + \allowbreak
          2\cdot\card{\sattup{\predexb(\varx, \vary)\land \neg \predexc(\vary)}} + \allowbreak
          2\cdot\card{\sattup{\qrysrnf}}$, respectively.
Note that the cost of $\qryenf$ can be arbitrarily
larger if $\predexc(\varx)\land \predexc(\vary)$ evaluates to a large intermediate result,
i.e., $\card{\sattup{\predexc(\varx)\land \predexc(\vary)}}\gg\card{\sattup{\predexb(\varx, \vary)}}$.
In contrast, the cost of $\qrysrnf$
can only be larger by a constant factor.
\end{example}

\section{Queries used in the evaluation}\label{sec:queries}

\begin{figure}
{
\footnotesize
\begin{tabular}{@{}r@{}l}
  $Q_{1}\eq\,$ & $ \neg (\exists x_2.\, \neg (\exists x_3.\, ((\mathsf{P}_3(x_1, x_0, x_3)) \land (\neg (\exists x_4.\, \mathsf{P}_4(x_1, x_3, x_4, x_2)))) \land (\exists x_4.\, \mathsf{P}_2(x_1, x_4)))) \lor (\mathsf{Q}_2(x_1, x_0))$ \\
  $Q_{2}\eq\,$ & $ \neg (\exists x_2.\, \neg (\exists x_3.\, (\neg (\mathsf{P}_4(x_1, x_0, x_2, x_3))) \land (\mathsf{P}_3(x_1, x_3, x_0)))) \land (\exists x_2.\, \exists x_3.\, (x_1 = x_2) \land (\mathsf{Q}_3(x_0, x_2, x_3)))$ \\
  $Q_{3}\eq\,$ & $ (\exists x_2.\, \neg (\exists x_3.\, ((\mathsf{P}_3(x_1, x_0, x_3)) {\land} (\mathsf{P}_1(x_0))) {\land} ((x_0 = x_1) {\land} (\neg (\mathsf{P}_4(x_1, x_2, x_3, x_0)))))){\land}(\exists x_2.\, \mathsf{Q}_3(x_1, x_2, x_0))$ \\
  $Q_{4}\eq\,$ & $ (\exists x_2.\, \mathsf{P}_3(x_1, x_2, x_0)) \land (\neg (x_0 = x_1)) \land (\neg (\exists x_2.\, \neg (\exists x_3.\, (\mathsf{P}_2(x_0, x_3)) \land (\neg (\mathsf{P}_4(x_1, x_3, x_2, x_0))))))$ \\
  $Q_{5}\eq\,$ & $ (\exists x_2.\, \exists x_3.\, \neg (\exists x_4.\, (\exists x_5.\, \mathsf{P}_3(x_0, x_5, x_4)) \land (\neg (\mathsf{P}_4(x_0, x_4, x_2, x_3))))) \land ((\neg (x_0 = x_1)) \land (\mathsf{P}_2(x_0, x_1)))$ \\
  $Q_{6}\eq\,$ & $ (\exists x_2.\, \neg (\exists x_3.\, ((\neg (\exists x_4.\, \mathsf{P}_3(x_0, x_3, x_4))) \land (\mathsf{Q}_3(x_0, x_1, x_3))) \land (\neg (\mathsf{R}_3(x_0, x_2, x_1))))) \land (\exists x_2.\, \mathsf{S}_3(x_0, x_1, x_2))$ \\
  $Q_{7}\eq\,$ & $ (\exists x_2.\, (\mathsf{P}_3(x_0, x_2, x_1)) \land (x_0 = x_2)) \land (\exists x_2.\, \neg (\exists x_3.\, (\neg (\exists x_4.\, \exists x_5.\, \mathsf{P}_4(x_0, x_2, x_5, x_4))) \land (\mathsf{P}_2(x_1, x_3))))$ \\
  $Q_{8}\eq\,$ & $ \neg (\exists x_2.\, \neg (\exists x_3.\, (\exists x_4.\, (\mathsf{P}_3(x_0, x_3, x_1)) {\land} ((\mathsf{P}_4(x_0, x_3, x_1, x_4)) {\land} (x_3 = x_4))) {\land} (\neg (\exists x_4.\, \mathsf{Q}_4(x_0, x_4, x_3, x_2)))))$ \\
  $Q_{9}\eq\,$ & $ \mathsf{P}_2(x_1, x_0) \land (\exists x_2.\, \neg (\exists x_3.\, ((\neg (\mathsf{P}_4(x_1, x_3, x_2, x_0))) \land (\mathsf{P}_3(x_0, x_3, x_1))) \land (\neg (\mathsf{Q}_2(x_0, x_3))))) \land (x_0 = x_1)$ \\
  $Q_{10}\eq\,$ & $ (\exists x_2.\, \mathsf{P}_3(x_1, x_2, x_0)) \lor (\neg (\exists x_2.\, \neg (\exists x_3.\, (\neg (\mathsf{P}_2(x_1, x_2))) \land (\exists x_4.\, (\mathsf{P}_1(x_0)) \land (\mathsf{P}_4(x_0, x_3, x_1, x_4))))))$ \\
  $Q^I_1\eq\,$ & $\mathsf{P}_{1}(x_0) \land \neg (\exists x_2.\, \exists x_3.\, (\mathsf{P}_{3}(x_0, x_2, x_3) \land \neg (x_0 = x_1))) $\\
  $Q^I_2\eq\,$ & $\mathsf{P}_{1}(x_0) \land \neg (\exists x_2.\, \exists x_3.\, (\mathsf{P}_{3}(x_0, x_2, x_3) \land \neg \mathsf{P}_{2}(x_1, x_3)))  $\\
  $Q^I_3\eq\,$ & $\mathsf{P}_{1}(x_0) \land \neg (\exists x_2.\, \exists x_3.\,
  (\mathsf{P}_{3}(x_0, x_3, x_2) \land \neg (x_1 = x_2))) $\\
  $Q^I_4\eq\,$ & $\mathsf{P}_{1}(x_1) \land \neg (\exists x_2.\, \exists x_3.\, (\mathsf{P}_{3}(x_1, x_3, x_2) \land \neg \mathsf{P}_{4}(x_1, x_0, x_3, x_2))) $\\
  $Q^I_5\eq\,$ & $\mathsf{P}_{1}(x_0) \land \neg (\exists x_2.\, \exists x_3.\, (\mathsf{P}_{3}(x_0, x_2, x_3) \land \neg \mathsf{Q}_{3}(x_1, x_2, x_3))) $
\end{tabular}
}
\caption{Randomly generated queries.}
\label{fig:queries}
\end{figure}

Figure~\ref{fig:queries} shows the randomly generated queries that we use in our evaluation. The index of each predicate indicates the predicate's arity.

% this is just a hack to move the license to the correct (last) page
\vspace*{4cm}
~

\end{document}